\documentclass[12pt]{article}

\usepackage{mathtools}
\usepackage{times}
\usepackage{soul}
\usepackage{url}
\usepackage[utf8]{inputenc}
\usepackage{amsmath}
\usepackage{booktabs}
\usepackage{multirow}
\usepackage{tabularx}
\usepackage{enumitem}

\usepackage{amsthm}
\usepackage{amssymb}
\usepackage{hyperref}
\usepackage{colortbl}
\usepackage{xcolor}
\usepackage{pict2e,picture,graphicx}

\makeatletter
\DeclareRobustCommand{\Arrow}[1][]{%
\check@mathfonts
\if\relax\detokenize{#1}\relax
\settowidth{\dimen@}{$\m@th\rightarrow$}%
\else
\setlength{\dimen@}{#1}%
\fi
\sbox\z@{\usefont{U}{lasy}{m}{n}\symbol{41}}%
\begin{picture}(\dimen@,\ht\z@)
\roundcap
\put(\dimexpr\dimen@-.7\wd\z@,0){\usebox\z@}
\put(0,\fontdimen22\textfont2){\line(1,0){\dimen@}}
\end{picture}%
}
\makeatother
\newcommand{\rightarrowshort}{\hspace{.2mm}\scalebox{.8}{\Arrow[.15cm]}\hspace{.2mm}}

\newtheorem{theorem}{Theorem}
\newtheorem{corollary}{Corollary}
\newtheorem{lemma}{Lemma}

\newtheorem{observation}{Observation}
\newtheorem{definition}{Definition}

\newcommand{\np}{{\sf{NP}}}

\newcommand{\nph}{{\np}-hard}
\newcommand{\conph}{{\sf{coNP}}-hard}
\newcommand{\cowah}{{\sf{coW[1]}}-hard}
\newcommand{\conphshort}{{\sf{coNP}}-h}
\newcommand{\conphns}{{\sf{coNP}}-hardness}
\newcommand{\nphns}{{\np}-hardness}

\newcommand{\nphshort}{{\np}-h}
\newcommand{\poly}{{\sf{P}}}

\newcommand{\wa}{\sf{W[1]}}
\newcommand{\wah}{{\sf{W[1]}}-hard}

\newcommand{\wb}{{\sf{W[2]}}}

\newcommand{\xp}{{\sf{XP}}}
\newcommand{\paranph}{{\sf{paraNP}}-hard}
\newcommand{\fpt}{{\sf{FPT}}}

\newcommand{\yes}{Yes}
\newcommand{\no}{No}
\newcommand{\yesins}{Yes-instance}
\newcommand{\noins}{No-instance}

\newcommand{\bigo}[1]{O(#1)}
\newcommand{\bigos}[1]{O^*(#1)}
\newcommand{\abs}[1]{|#1|}
\newcommand{\edge}[2]{\{#1,#2\}}
\newcommand{\xc}{\mathcal{H}}
\newcommand{\xce}{H}
\newcommand{\xs}{A}
\newcommand{\xse}{a}
\newcommand{\xsize}{\kappa}
\newcommand{\w}{w} 
\newcommand{\onlyfull}[1]{}
\newcommand{\prob}[1]{{\sc{#1}}}

\newcommand{\cbcm}{CBCM}
\newcommand{\sbcm}{SBCM}
\newcommand{\sdcm}{SDCM}
\newcommand{\probb}[2]{\prob{#1-#2}}

\newcommand{\memph}[1]{\emph{#1}}

\newcommand{\nemph}[1]{\emph{#1}}

\newcommand{\fnotion}[1]{{\it{#1}}}  

\newcommand{\abmv}{ABMV}

\newcommand{\vset}{N} 
\newcommand{\eset}{A}  
\newcommand{\vere}{u}

\newcommand{\candc}[1]{C({#1})} 
\newcommand{\vaes}[2]{{#1^{\star}(#2)}}  
\DeclareMathOperator*{\argmax}{arg\,max}
\newcommand{\setmid}{:}

\newcommand{\swinspe}[3]{{W}_{#2, #3}^+(#1)} 
\newcommand{\swinshort}[1]{{W}^+(#1)} %
\newcommand{\pwinspe}[3]{{W}_{#2, #3}(#1)} 
\newcommand{\pwinshort}[1]{{W}(#1)} %
\newcommand{\slosespe}[3]{{L}_{#2, #3}(#1)} 
\newcommand{\sloseshort}[1]{{L}(#1)} %

\newcommand{\score}[4]{{\textsf{sc}}^{#3}_{#4}(#1,#2)} 

\newtheorem{claimm}{Claim}
\newcommand{\EP}[4]
{\begin{center}
{
\begin{tabularx}{0.98\columnwidth}{ll}
\toprule
\multicolumn{2}{l}{{#1} (#2)} \\ \midrule
{\bf Input:}   & \parbox[t]{0.83\columnwidth}{#3\vspace*{1mm}}  \\
{\bf Question:}& \parbox[t]{0.83\columnwidth}{#4\vspace*{.5mm}} \\ \bottomrule
\end{tabularx}
}
\end{center}}

\newcommand{\EPP}[3]
{\begin{center}
{
\begin{tabularx}{0.98\columnwidth}{ll}
\toprule
\multicolumn{2}{l}{{#1}} \\ \midrule
{\bf Input:}   & \parbox[t]{0.83\columnwidth}{#2\vspace*{1mm}}  \\
{\bf Question:}& \parbox[t]{0.83\columnwidth}{#3\vspace*{.5mm}} \\ \bottomrule
\end{tabularx}
}
\end{center}}

 \usepackage{geometry}
\newcounter{exo}
\makeatletter
\newenvironment{example}[1][]
    {\refstepcounter{exo}%
\begin{center}
   \begin{tabular}{|>{\columncolor{gray!10}}p{0.967\textwidth}|}
    \hline \\[-3mm] {\bfseries{Example \theexo}: #1}\\ \hline
    }
    {
    \\[1.5mm] \hline
  \end{tabular}
    \end{center}
    }

\begin{document}
\title{Complexity of Manipulating and Controlling Approval-Based Multiwinner Voting\thanks{A preliminary version of the paper has appeared in the Proceedings of the 28th International Joint Conference on Artificial Intelligence~\protect\cite{DBLP:conf/ijcai/Yang19}. The current version includes all proofs and some new results.}}

\author{Yongjie Yang}

\date{Chair of Economic Theory, Saarland University, Saarb\"{u}cken, Germany \\
yyongjiecs@gmial.com}

\maketitle

\begin{abstract}
We investigate the complexity of several manipulation\onlyfull{, bribery} and control problems under numerous prevalent approval-based multiwinner voting rules. Particularly, the rules we study include  approval voting (AV), satisfaction approval voting (SAV), net-satisfaction approval voting (NSAV), proportional approval voting (PAV), approval-based Chamberlin-Courant voting (ABCCV),  minimax approval voting (MAV), etc. We show that these rules generally resist the strategic types scrutinized in the paper, with only a few exceptions. In addition, we also obtain many fixed-parameter tractability results for these problems with respect to several natural parameters, and derive polynomial-time algorithms for certain special cases.
\medskip

\noindent{\textbf{keywords:}} {approval-based multiwinner voting, manipulation,  control, fixed-parameter tractability, NP-hardness, integer linear programming}
\end{abstract}

\section{Introduction}
%
Investigating the complexity of strategic voting problems has been a vibrant research topic in Computational Social Choice over the last three decades. Since the pioneering works of Bartholdi and Orlin~\cite{BartholdiJames1991SocialChoiceWelfareSTV}, and Bartholdi, Tovey, and Trick~\cite{BARTHOLDI89,Bartholdi92howhard}, many manipulation, control, and bribery problems have been proposed for single-winner voting rules. Initiated by Meir~et~al.~\cite{DBLP:journals/jair/MeirPRZ08}, these problems have been extended to multiwinner voting in recent years.
Enriching this line of research, we propose some natural manipulation\onlyfull{, bribery,} and control problems for approval-based multiwinner voting ({\abmv}) rules and investigate the complexity of these problems for the prevalent rules {{approval voting}} (AV), {{satisfaction approval}} voting (SAV),  {{net-satisfaction approval voting}} (NSAV), proportional approval voting (PAV), approval-based Chamberlin-Courant voting (ABCCV), minimax approval voting (MAV), etc. 
\onlyfull{In general, in these problems, there are either several voters who attempt to coordinate their votes in order to improve the results\onlyfull{ in their favor} ({\it{manipulation}})\onlyfull{, or an external agent who aims to make some distinguished candidates win by either bribing some vulnerable voters ({\it{bribery}})}, or by directly modifying the election ({\it{control}}).
The study on strategy voting problems for single-winner voting rules has dominated computational social choice for a decade and recently such study has been extended to multi-winner voting rules~\cite{DBLP:conf/ijcai/BredereckKN17,DBLP:journals/jair/MeirPRZ08,DBLP:conf/atal/Peters18}. Nevertheless, our models differ from the previous studied ones in many aspects.}

Manipulation models the scenario where some voters participating in an election, called manipulators, want to improve the election result in their favor by misreporting their preferences.  A necessary notion in the study of manipulation is the preference of a voter over all possible outcomes. Unlike ranking-based single-winner voting where every voter has a linear preference over all candidates and the outcome is a single winner (see, e.g.,~\cite{smith2006votingsystems}), in the setting of {\abmv} how to deduce voters' preferences over committees (subsets of candidates) from their dichotomous preferences over candidates is already a question without a clear-cut answer. Some prominent approaches for this purpose have been proposed in the literature (see, e.g.,~\cite{DBLP:conf/ijcai/LacknerS18,DBLP:journals/scw/LaslierS16,DBLP:conf/atal/Peters18}). For example, a voter may prefer a committee to another one if the former contains more of her approved candidates. In a more conservative situation, a voter is more satisfied with the former one only if the former contains not only more of her truly approved candidates but also all of her approved candidates included in the latter one. The two approaches lead to the concepts of {\it{cardinality-strategyproofness}} and {\it{subset-strategyproofness}} respectively when only one manipulator is considered (see, e.g.,~\cite{DBLP:conf/atal/Peters18}). In contrast to the celebrated Gibbard-Satterthwaite theorem for single-winner voting~\cite{Gibbard73,Satterthwaite75}, there exist natural {\abmv} rules (like AV) that are strategyproof with respect to the above two concepts. However, many {\abmv} rules are not strategyproof. For example, Peters~\cite{DBLP:conf/atal/Peters18} showed that any {\abmv} rules that satisfy certain proportional properties are not cardinality- and subset-strategyproof. Aziz~et~al.~\cite{DBLP:conf/atal/AzizGGMMW15} showed that SAV is not cardinality-strategyproof. Motivated by these non-strategyproofness results, many multiwinner manipulation problems have been proposed recently. Particularly, Aziz~et~al.~\cite{DBLP:conf/atal/AzizGGMMW15} studied the {\prob{Winner Manipulation}} and {\prob{Winner Set Manipulation}} problems for resolute multiwinner voting rules, i.e., rules that return exactly one wining committee (precisely, the resoluteness of the rules studied in their paper is achieved by using  a linear order on candidates to break ties). In the {\prob{Winner Manipulation}} problem, given are an election, a distinguished candidate, and an integer~$\ell$, and the question is whether it is possible to add~$\ell$ additional votes so that the distinguished candidate is included in the winning committee. In the {\prob{Winner Set Manipulation}} problem, there are multiple distinguished candidates, and the question is whether we can add~$\ell$ additional votes so that the distinguished candidates are exactly the winners. In all of these problems, it is assumed that the manipulators have one clear target set of candidates whom they want to make the winners.
This assumption, however, does not seem to be very realistic since there may exist exponentially many outcomes which are more preferred to the current outcome by the manipulators.

In this paper, we study two new manipulation problems,
where manipulators judge the results with respect to the two preference extension approaches discussed above. 
Concretely, given an election~$E$ and a winning committee of~$E$, in the {\prob{Cardinality-Based Coalition Manipulation}} problem (\cbcm), manipulators aim to obtain winning committees including more of their truthfully approved candidates, and in the {\prob{Subset-Based Coalition Manipulation}} problem (\sbcm), manipulators aim to obtain winning committees including all their current approved winning candidates and at least one more approved candidate that is not in the current winning $k$-committee. 

Lackner and Skowron~\cite{DBLP:conf/ijcai/LacknerS18} put forward the notion of stochastic domination strategyproofness (SD-strategyproofness), and called for investigating the algorithmic challenge of finding successful SD-manipulations. Several of our reductions for the above two manipulation problems also apply to the SD-strategyproofness based manipulation problem (\sdcm). (see Section~\ref{sec-preliminaries} for formal definitions of the notions).

Besides, we study several election control problems named {\prob{Constructive Control by Adding Voters}} ({\prob{CCAV}}),  {\prob{Constructive Control by Adding Candidates}} ({\prob{CCAC}}),  {\prob{Constructive Control by Deleting Voters}} ({\prob{CCDV}}), and {\prob{Constructive Control by Deleting Candidates}} ({\prob{CCDC}}) for multiwinner voting rules. These problems model the scenario where a powerful external agent would like to make a given subset of candidates be contained in all winning $k$-committees by adding or deleting a limited number of voters or candidates.
They are direct generalizations of the extensively studied control problems for single-winner voting rules.

Our results offer a comprehensive understanding of the complexity of the manipulation and election control problems. In particular, we show that the manipulation problems are generally {\nph}, and show that most of the election control problems are also {\nph} even when we aim to merely select one winner or there is just one distinguished candidate. However, there are several interesting exceptions:
\begin{itemize}
\item We show that the computationally demanding\footnote{Rules under which winners are {\nph} to compute.} rules PAV and ABCCV are immune to {\prob{CCAC}} as long as the number of distinguished candidates equals that of the size~$k$ of the winning committee. For SAV and NSAV under which the winners can be efficiently computed, {\prob{CCAC}} is already {\nph} even when we aim to select only one winner.

\item We show that {\prob{CCDC}} for PAV and and {\prob{CCDC}} for ABCCV are {\nph} if and only if $k\geq 2$. It is worth mentioning that {\abmv} rules in the special case where $k=2$ has been demonstrated to play a distinctive role in runoff elections~\cite{DBLP:conf/ijcai/DelemazureLLS22}.

\item We show that, for MAV, {\prob{CCAV}} is {\nph} even when $k=1$, but the dual problem {\prob{CCDV}} is polynomial-time solvable for~$k$ being any constant. The complexity differential between the two problems for the same rule is quite rare.
\end{itemize}
We refer to Table~\ref{tab-results-summary} for a summary of our concrete results.

To complement the hardness results, we explored the fixed-parameter tractability of the manipulation and election control problems with respect to several natural parameters, including the number of candidates, the number of votes, the number of manipulators, etc. Our concrete results regarding this issue are summarized in Table~\ref{tab-fpt-results}.

\subsection{Related Works}
\label{sec-related-works}
In this section, we discuss important related works that have not been mentioned or well-elaborated above.

Meir~et~al.~\cite{DBLP:journals/jair/MeirPRZ08} initiated the study of the complexity of control and manipulation problems for multiwinner voting rules, but they mainly considered ranking-based voting rules, and their treatment for manipulation assumes the presence of only a single manipulator. Besides, in their models, strategic agents (the manipulator or the powerful external agent) derive utilities from  candidates, and they attempt to achieve a winning~$k$-committee yielding the maximum total utilities. These models nicely bypass the question of extending voters' preferences over candidates to preferences over committees, but captures many real-world scenarios. 
The work of Meir~et~al.~\cite{DBLP:journals/jair/MeirPRZ08} on manipulation was later complemented by Obraztsova, Zick, and Elkind~\cite{DBLP:conf/atal/ObraztsovaZE13} who considered further ranking-based voting rules and investigated how different tie-breaking schemes change the complexity of the problem. Bredereck, Kaczmarczyk, and Niedermeier~\cite{DBLP:journals/aamas/BredereckKN21} take a step further by firstly extending the manipulation model by considering the presence of multiple manipulators. Notably, Bredereck, Kaczmarczyk, and Niedermeier~\cite{DBLP:journals/aamas/BredereckKN21} adopted three different functions (utilitarian, egalitarian, candidate-wise egalitarian) to merge the utilities of manipulators obtained from different committees.
Compared to these works, we focus on {\abmv} rules. On top of that, our model departs from previously utility-based manipulation 
by assuming that manipulators resort to preference extensions to evaluate committees. We remark that {\nphns} proofs of our problems can be modified to show the {\nphns} of the utility-involved variants by assigning to certain candidates in the constructed elections in the {\nphns} proofs very high utilities and assign to other candidates utility zero.

Our study on manipulations is also related to the works of Lackner and Skowron~\cite{DBLP:conf/ijcai/LacknerS18}, Laslier and van der Straeten~\cite{DBLP:journals/scw/LaslierS16}, Peters~\cite{DBLP:conf/atal/Peters18}, and Yang and Wang~\cite{DBLP:conf/ijcai/YangW18} who studied strategyproofness of {\abmv}. However, they were concerned with mathematical properties of these rules and focused only on one manipulator.

Another line of research concerning manipulation in {\abmv} is as follows. In this setting, it is assumed that voters have linear preferences (but they are only allowed to submit dichotomous preferences) over candidates and the question is whether voters have incentive to submit nonsincere votes in order to improve the result. A vote is sincere with respect to a linear order on candidates if the approved candidates are exactly those ranked above some threshold candidate in the linear order. Moreover, voters compare different outcomes based on some preference extension principles such as the Kelly extension principle, G\"{a}rdenfors extension principle, etc. We refer to~\cite{Barber2004,EndrissTD2013} and references therein for further discussions.

Our control problems are straightforward generalizations of four standard election control problems first proposed in~\cite{Bartholdi92howhard} for single-winner voting. The number of papers covering this theme is huge. We refer to~\cite{BaumeisterR2016,handbookofcomsoc2016Cha7FR} and references therein for important progress by 2016, and refer to~\cite{DBLP:conf/atal/CarletonCHNTW23a,DBLP:conf/atal/CarletonCHNTW23,DBLP:journals/aamas/ErdelyiNRRYZ21,DBLP:journals/ai/NevelingR21,DBLP:conf/atal/Yang17,DBLP:conf/ecai/000120,DBLP:conf/atal/000123a} for some recent results.

In addition, our control problems are related to the group control problems in the setting of group identification~\cite{DBLP:journals/aamas/ErdelyiRY20,DBLP:journals/aamas/YangD18}. Group identification models the scenario where voters and candidates (termed as individuals) coincide and the goal is to classify the individuals into two groups: the group of socially qualified individuals and the group of not socially qualified individuals. The group control problems consist in making some given distinguished individuals socially qualified by adding/deleting/partitioning the individuals.


Besides, we would like to point out that other types of strategic voting problems in multiwinner voting have also been  studied recently. Faliszewski, Skowron, and Talmon~\cite{DBLP:conf/atal/FaliszewskiST17} studied bribery problems for {\abmv} rules, where the goal is to ensure one distinguished candidate to be included in the winning $k$-committee by applying a limited number of modification operations. Complementing this work, Yang~\cite{DBLP:conf/atal/000120} studied the complexity of the counterparts of these bribery problems. Erd\'{e}lyi, Reger, and Yang~\cite{DBLP:journals/aamas/ErdelyiRY20}, and Boehmer~et~al.~\cite{DBLP:conf/ijcai/BoehmerBKL20} studied the complexity of bribery in group identification.

Finally, we remark that in a companion paper~\cite{DBLP:journals/corr/abs-2304-11927}, we study destructive counterparts of the election control problems for the aforementioned rules and their  sequential variants.


\subsection{Organization}
The remainder of the paper is organized as follows. In Section~\ref{sec-preliminaries}, we provide important notions and definitions that are used in our study.
Then, we unfold our detailed results in Sections~\ref{sec-manipulation}--\ref{sec-fpt}, where Section~\ref{sec-manipulation} is devoted to the complexity of manipulation problems, Section~\ref{sec-control} covers our complexity results for control problems, and Section~\ref{sec-fpt} studies a variety of {\fpt}-algorithms for problems studied in Sections~\ref{sec-manipulation} and~\ref{sec-control}. Section~\ref{sec-conclusion} concludes our main contributions and lay out interesting topics for future research.

\begin{table}
\caption{A summary of the complexity of constructive control problems for {\abmv} rules. Here, ``{\nphshort}'' stands for ``{\nph}'', and ``{\poly}'' for ``polynomial-time solvable''.  Additionally,~$J$ denotes the set of distinguished candidates, and~$k$ denotes the size of the winning committee. Notice that $k=1$ implies $\abs{J}=1$.
Results labeled by~$\spadesuit$ are from \protect\cite{DBLP:journals/jair/MeirPRZ08}, by~$\heartsuit$ are from~\protect\cite{baumeisterapproval09,DBLP:journals/jcss/HemaspaandraH07}, and~by $\clubsuit$ are from~\protect\cite{baumeisterapproval09,DBLP:journals/ai/HemaspaandraHR07,DBLP:conf/icaart/Lin11}. Several of our hardness results hold with even further restrictions on the input. Related discussions  are provided after the corresponding theorems. The results for manipulation hold for all the three manipulation problems {\cbcm}, {\sbcm}, and {\sdcm}.}
\label{tab-results-summary}
\renewcommand{\tabcolsep}{0.6mm}
\scriptsize{
\begin{center}
\begin{tabular}{|l|l|l|l|l|l|}\toprule
          & manipulation & {\prob{CCAV}} & {\prob{CCDV}}  & {\prob{CCAC}} & {\prob{CCDC}} \\ \toprule

AV
& {{\nphshort} (Thm.~\ref{thm-ccm-av-np-hard})}
& {\nphshort} ($k=1$, $\clubsuit$)
& {\nphshort} ($k=1$, $\clubsuit$)
& {immune} ($\heartsuit$)
& {\poly} ($\spadesuit$)\\  \midrule

NSAV
& {{\nphshort} (Thm.~\ref{thm-manipulation-np-hard-many-rules})}
& {\nphshort} ($k=1$, Thm.~\ref{thm-ccav-sav-np-hard})
& {\nphshort} ($k=1$, $\clubsuit$)
& {\nphshort} ($k=1$, Thm.~\ref{thm-ccac-sav-np-hard})
& {\nphshort} ($k=1$, Thm.~\ref{thm-ccdc-sav-np-hard})\\ \midrule

PAV/
&
& \multicolumn{4}{l|}{{\conphshort} ($\abs{J}=1 \& \ell=0$, Thm.~\ref{thm-cc-abccv-pav-co-np})}\\ \cline{3-6}

ABCCV&
& \multicolumn{2}{l|}{{\nphshort} ($k=1$, $\clubsuit$)}
&  {immune} ($\abs{J}=k$, Thm.~\ref{thm-pav-abbcv-mav-immue-to-ccac-k-equal-distinguished-candidates})
&{\poly} ($k=1$, $\spadesuit$)\\

&
& \multicolumn{2}{l|}{}
&
& {\nphshort} ($k=2$\&$\abs{J}=1$, Thm.~\ref{thm-ccdc-abccv-pav-nph-k-2})\\ \midrule

MAV
&  {{\nphshort} (Thm.~\ref{thm-manipulation-mav-np-hard})}
&   \multicolumn{4}{l|}{{\nphshort} ($\abs{J}=1 \& \ell=0$, Thm.~\ref{thm-cc-mav-nph})}\\ \cline{3-6}

&
& {\nphshort} ($k=1$\&$\abs{V}=1$, Thm.~\ref{thm-ccav-mav-nph-k-1})
& {\poly} ($k=\bigo{1}$, Thm.~\ref{thm-ccav-ccdv-mav-polynomial-time-solvable-k-constant})
& {\nphshort} ($k=1$\&$\abs{C}=2$, Thm.~\ref{thm-ccac-mav-nph-k-1})
& {\nphshort} ($k=1$, Thm.~\ref{thm-ccdc-mav-nph-k-1})\\ \bottomrule
\end{tabular}
\end{center}
}
\end{table}

\begin{table}
\caption{A summary of the fixed-parameter tractability of the manipulation and control problems with respect to several natural parameters. Here,~$t$ denotes the number of manipulators,~$m$ denotes the number of (registered and unregistered) candidates,~$n$ denotes the number of (registered and unregistered) votes, and~$\ell$ is the solution size.
{\fpt}-results without a reference mean that they are trivial.}
\label{tab-fpt-results}
\scriptsize{
\begin{center}
\begin{tabular}{|l|l|l|l|l|l|l|l|}\toprule
          & \multicolumn{2}{l|}{{\cbcm}/{\sbcm}} & \multicolumn{2}{l|}{{\prob{CCAV}/\prob{CCDV}}}  & \multicolumn{3}{l|}{\prob{CCAC}/\prob{CCDC}} \\ \cline{2-8}

          & $m$ & $t$                         & $m$   & $n$                                     & $m$   & $n$   & $n+\ell$ \\ \midrule

AV
& {{\fpt}} (Thms.~\ref{thm-manipulation-fpt-wrt-candidate},~\ref{thm-new-manipulation-sav-nsav-fpt-wrt-candidate})    & \underline{{\xp} (Thm.~\ref{thm-maniuplation-sav-nsav-polynomial-time-solvable-constant-number-manipulators})}
& {\fpt} (Thm.~\ref{thm-ccav-ccdv-av-sav-nsav-fpt-candidates})    & {\fpt}
& \multicolumn{3}{l|}{{\fpt} (polynomial-time solvable)} \\ \hline

SAV/NSAV
& {\fpt} (Thms.~\ref{thm-manipulation-sav-nsav-fpt-wrt-candidate},~\ref{thm-new-manipulation-sav-nsav-fpt-wrt-candidate})  &  \underline{{\xp} (Thm.~\ref{thm-maniuplation-sav-nsav-polynomial-time-solvable-constant-number-manipulators})}
& {\fpt} (Thm.~\ref{thm-ccav-ccdv-av-sav-nsav-fpt-candidates})  & {\fpt}
& {\fpt} & open & {\fpt} (Thm.~\ref{thm-ccac-ccdc-many-rules-fpt-ell-plus-n})\\ \hline

PAV/ABCCV & {\fpt} & open
&   {\fpt} (Thm.~\ref{thm-ccadv-abccv-pav-fpt-m}) & {\fpt}
& {\fpt} & open & {\fpt} (Thm.~\ref{thm-ccac-ccdc-many-rules-fpt-ell-plus-n})\\ \hline

MAV
& {\fpt} & open
& {\fpt} (Thm.~\ref{thm-ccav-mav-fpt-m})  & {\fpt}
& {\fpt} & open & {\fpt} (Thm.~\ref{thm-ccac-ccdc-many-rules-fpt-ell-plus-n})\\ \bottomrule
\end{tabular}
\end{center}
}
\end{table}

\section{Preliminaries}
\label{sec-preliminaries}
We assume the reader is familiar with the basics in graph theory and (parameterized) complexity theory~\cite{DBLP:books/sp/CyganFKLMPPS15,DBLP:conf/lata/Downey12,DBLP:journals/interfaces/Tovey02,Douglas2000}.

\subsection{Approval-Based Multiwinner Voting}
In the approval model, an {\memph{election}} is a tuple~$(C, V)$ where~$C$ is a set of candidates, and~$V$ is a multiset of votes cast by a set of voters. Each {\fnotion{vote}}~$v\in V$ is defined as a 
subset of~$C$, consisting of the candidates approved by the corresponding voter.
For ease of exposition, we interchangeably use the terms vote and voter. So, by saying that a vote~$v$ {\fnotion{approves}} a candidate~$c$, we simply mean that~$c\in v$. By saying that a vote~$v$ {\fnotion{approves}}~$C'\subseteq C$, we mean that~$v$ approves all candidates in~$C'$ (and probably also approves some other candidates not in~$C'$). A~{\fnotion{$k$-set}} is a set  of cardinality~$k$. A subset of candidates is also called a {\fnotion{committee}}, and a $k$-subset of candidates is called a {\fnotion{$k$-committee}}.
 Some other important notations are summarized in Table~\ref{tab-notations}.
\begin{table}[h]
 \caption{Several important notations. Here, $(C, V)$ is an election,~$C'$ is a subset of~$C$,~$v$ is a vote from $V$, and $k$ is an integer.}
    \centering
    \begin{tabular}{llp{0.5\textwidth}}\toprule
    notations & formal definitions & descriptions \\ \midrule

    $v_{C'}$ & $v\setminus C'$ & the vote obtained from~$v$ by removing all candidates not contained in~$C'$\\ \hline

    $V_{C'}$ & $\{v_{C'} \setmid v\in V\}$ & the multiset of votes obtained from~$V$ by replacing each $v\in V$ by~$v_{C'}$\\ \hline

    $(C', V)$ &  $(C', V_{C'})$ & the election~$(C, V)$ restricted to~$C'$\\ \hline

    $V(c)$  & $\{v\in V \setmid c\in v\}$ & the multiset of votes in~$V$ approving~$c$\\ \hline

    $V(C')$ & $\bigcup_{c\in C'}V(c)$ & the multiset of votes in~$V$ approving at least one candidate from~$C'$ \\ \hline

    $\vaes{V}{C'}$ & $\{v\in V \setmid v=C'\}$ & the multiset of votes in~$V$ approving exactly the candidates from~$C'$\\ \hline

    $C^{\vee}(V')$ & $\bigcup_{v\in V'}v$ & the set of candidates approved by at least one vote in~$V'$\\  \hline

    $C^{{\star}}_V(V')$ &$\{c\in C \setmid V(c)=V'\}$ & the set of candidates in~$C$ approved by all votes in~$V'$ but disapproved by any vote in~$V\setminus V'$\\ \hline

    $m^{\star}_V(V')$ & $\abs{C_V^{\star}(V')}$ & the cardinality of $C_V^{\star}(V')$\\ \hline

    $\mathcal{C}_{k, C}(C')$ & $\{X\subseteq C \setmid \abs{X}=k, C'\subseteq X\}$ & the set of all $k$-committees of~$C$ containing~$C'$\\ \hline

    ${\overline{\mathcal{C}_{k, C}}}({C'})$ & $\{X\subseteq C \setmid \abs{X}=k, C'\setminus X\neq\emptyset\}$ & the set of all $k$-committees of~$C$ not containing~$C'$\\ \bottomrule
    \end{tabular}
    \label{tab-notations}
\end{table}


An {\memph{{\abmv} rule}} (aka.\ {\memph{approval-based~committee selection rule}}) maps each election $(C, V)$ and an integer~$k$ to a collection of~$k$-committees of~$C$, which are called {\memph{winning~$k$-committees}} of this rule at~$(C, V)$. In practice, when a rule returns multiple winning $k$-committees,  a certain {\fnotion{tie-breaking scheme}} is often used to select exactly one winning $k$-committee.

We study some important {\abmv} rules which can be categorized into two groups. Under these rules, each~$k$-committee receives a score based on the votes, and winning $k$-committees are those either maximizing or minimizing the corresponding scoring functions. For the first group of rules the score of a committee is the sum of the scores of its members. Due to this property, these rules are referred to as {\memph{additive rules}} in the literature~\cite{DBLP:conf/atal/AzizGGMMW15,Kilgour2010,DBLP:conf/ijcai/YangW18}.
In the following, let $(C, V)$ be an election.

\renewcommand*\descriptionlabel[1]{\hspace\labelsep
  \normalfont\bfseries
  {#1}}
\begin{description}
\item[Approval voting (AV)] The score of a candidate~$c\in C$ in $(C, V)$ is the number of votes in~$V$ approving~$c$, and winning $k$-committees are those with the highest total scores of their members. Formally, a $k$-committee~$w$ is a winning $k$-committee of AV at $(C, V)$ if $\sum_{v\in V} \abs{v\cap w}\geq \sum_{v\in V} \abs{v\cap w'}$ for all $w'\subseteq C$ such that $\abs{w'}=k$.
\end{description}

Due to its simplicity and intuitiveness, AV has been widely used in practice both as a single-winner voting rule and as a multiwinner voting rule~\cite{BramsF2007approvalvotingsecondedition}. 

\begin{description}
\item[Satisfaction approval voting (SAV)] Each candidate~$c\in C$ receives~$\frac{1}{\abs{v}}$ points from each vote~$v$ approving~$c$, and the SAV  score of~$c$ in the election $(C, V)$ is $\sum_{c\in v\in V}\frac{1}{\abs{v}}$. 
    Similar to AV, winning $k$-committees of SAV at $(C, V)$ are those with the highest SAV score, i.e., $k$-committees~$\w\subseteq C$ with the maximum possible value of $\sum_{\emptyset\neq v\in V}\frac{\abs{v\cap \w}}{\abs{v}}$.
\end{description}

Using SAV as a multiwinner voting rule was advocated by Brams and Kilgour~\cite{Bram2014Kilgour}. Interestingly, the application of SAV scores goes beyond the setting of {\abmv}. For instance, SAV scores are related to the Shapley values of players in a ranking game which is relevant to the settings of ranking wines and allocating profits among museums~\cite{DehezG2018,DBLP:journals/geb/GinsburghZ03,GinsburghZ2012}.

\begin{description}
\item[Net-satisfaction approval voting (NSAV)] This rule is a variant of SAV which captures the principle that if addition of approved candidates of a vote in a committee increases the satisfaction of the corresponding voter, then addition of disapproved candidates decreases the satisfaction. In particular, each candidate~$c\in C$ receives $\frac{1}{\abs{v}}$ points from every vote $v\in V$ approving~$c$, and receives $\frac{1}{m-\abs{v}}$ points from every vote $v\in V$ not approving~$c$, and the NSAV score of~$c$ in $(C, V)$ is defined as $\sum_{c\in v\in V} \frac{1}{\abs{v}}-\sum_{c\not\in v\in V}\frac{1}{m-\abs{v}}$. The satisfaction degree of a vote~$v$ derived from a committee~$\w\subseteq C$ is measured by the total points of candidates in~$\w$ received from~$v$, and this rule aims to maximize voters' total satisfaction. Precisely, winning $k$-committees of NSAV at $(C, V)$ are $k$-committees $\w\subseteq C$  maximizing
$\sum_{v\in V, v\neq\emptyset} \frac{\abs{w\cap v}}{\abs{v}}-\sum_{v\in V, v\neq C}\frac{\abs{w\setminus v}}{m-\abs{v}}$.
\end{description}

NSAV was proposed by Kilgour and Marshall~\cite{Kilgour2014Marshall}. 

We call an additive rule {\fnotion{polynomial computable}} if given a vote and a candidate, the score of the candidate received from the vote can be computed in polynomial time in the size of the vote. To the best of our knowledge, all natural rules studied so far in the literature are polynomial computable.
\medskip

Now we give the definitions of the second group of rules where the score of each committee is combinatorially determined by its members.
\begin{description}
\item[Approval-based Chamberlin-Courant voting (ABCCV)] A voter is satisfied with a committee if this committee includes at least one of her approved candidates. The score of a committee is the number of voters satisfied with the committee, and winning $k$-committees are those with the maximum score.
\end{description}

ABCCV was first suggested by Thiele~\cite{Thiele1985} and then independently proposed by Chamberlin and Courant~\cite{ChamberlinC1983APSR10.2307/1957270}.

\begin{description}
\item [Proportional approval voting (PAV)] The score of a committee~$\w\subseteq C$ is $\sum_{v\in V, v\cap \w\neq \emptyset}\left(\sum_{i=1}^{\abs{v\cap \w}}\frac{1}{i}\right)$. Winning $k$-committees are those with the maximum score. \onlyfull{PAV fulfills several proportional properties~\cite{DBLP:journals/scw/AzizBCEFW17}. }

\item[Minimax approval voting (MAV)] The Hamming distance between two subsets~$\w\subseteq C$ and~$v\subseteq C$ is $\abs{\w\setminus v}+\abs{v\setminus \w}$. The score of a committee~$\w$ is the maximum Hamming distance between~$\w$ and the votes, i.e., $\max_{v\in V} (\abs{\w\setminus v}+\abs{v\setminus \w})$. Winning $k$-committees are those having the smallest score.
\end{description}

PAV was first studied by Thiele~\cite{Thiele1985}, and MAV was proposed by Brams~\cite{Bramsminimaxapproval2007}.
It should be pointed out that calculating a winning $k$-committee with respect to the second group of rules is computationally hard~\cite{DBLP:conf/atal/AzizGGMMW15,LeGrand2004Tereport,DBLP:journals/scw/ProcacciaRZ08}, standing in contrast to the polynomial-time solvability for many additive rules~\cite{DBLP:conf/atal/AzizGGMMW15,DBLP:conf/atal/YangW19}.%

Now we introduce the class of Thiele's rules which contain AV, ABCCV, and PAV.

\begin{description}
\item[$\omega$-Thiele] Each Thiele's rule is characterized by a function $\omega: \mathbb{N}\rightarrow \mathbb{R}$ so that $\omega(0)=0$ and $\omega(i)\leq \omega(i+1)$ for all nonnegative integers~$i$. The score of a committee $\w\subseteq C$ is defined as $\sum_{v\in V}\omega(\abs{v\cap w})$. The rule selects $k$-committees with the maximum score.\footnote{Thiele's rules are also studied under some other names including weighted PAV rules, generalized approval procedures, etc.~\cite{DBLP:journals/scw/AzizBCEFW17,Kilgour2014Marshall}.}
\end{description}

Obviously, AV is the $\omega$-Thiele rule where $\omega(i)=i$, ABCCV is the $\omega$-Thiele's rule such that $\omega(i)=1$ for all $i>0$, and PAV is the $\omega$-Thiele's rule such that $\omega(i)=\sum_{j=1}^i 1/j$ for all $i>0$. Many of our results apply to a subclass of Thiele's rules such that $\omega(2)<2\omega(1)$. In particular, ABCCV and PAV belong to this subclass.

In the description of our algorithms, the following notations are consistently used.
Let~$\varphi$ be an {\abmv} rule defined above except~MAV. Let $E=(C, V)$ be an election. For  a vote~$v\in V$ and a committee~$\w\subseteq C$, we use~$\score{v}{\w}{E}{\varphi}$ to denote the~$\varphi$ score of~$\w$ received from~$v$ in~$E$. For~$V'\subseteq V$, we define $\score{V'}{\w}{E}{\varphi}=\sum_{v\in V}\score{v}{\w}{E}{\varphi}$. If~$\w=\{c\}$ is a singleton, we simply write~$c$ instead of~$\{c\}$ in the above notations. Additionally, we drop the subscript~$\varphi$ whenever it is clear from the context which rule~$\varphi$ is discussed.

The {\nemph{$k$-winning-threshold}} of~$\varphi$ at an election $E=(C, V)$ is defined as follows. Let~$\rhd$ be a linear order on~$C$ so that for $c, c'\in C$ it holds that $c\rhd c'$ implies $\score{v}{c}{E}{\varphi}\geq \score{v}{c'}{E}{\varphi}$, i.e., candidates are ordered in~$\rhd$ with respect to their~$\varphi$ scores received from~$V$, from those with the highest scores to those with the lowest scores. The $k$-winning-threshold of~$\varphi$ at~$E$ is the~$\varphi$ score of the $k$-th candidate in~$\rhd$.

Let~$s$ be the $k$-winning-threshold of~$\varphi$ at an election~$E=(C, V)$, and let~$x$ be the number of candidates from~$C$ whose~$\varphi$ scores in~$E$ are exactly~$s$. Then,~$\swinspe{E}{\varphi}{k}$ is defined as the set of candidates from~$C$ whose~$\varphi$ scores in~$E$ are at least~$s$ if $x=1$, and is defined as the set of candidates from~$C$ whose~$\varphi$ scores in~$E$ are strictly larger than~$s$ if $x\geq 2$. Additionally,  we  define $\slosespe{E}{\varphi}{k}\subseteq C$ as the set of all candidates whose~$\varphi$ scores in~$E$ are strictly smaller than~$s$, and define $\pwinspe{E}{\varphi}{k}=C\setminus (\swinspe{E}{\varphi}{k} \cup \slosespe{E}{\varphi}{k}$ as the set of all the other candidates in~$C$. Obviously, $\pwinshort{E}=\emptyset$ if and only if $x=1$ and, moreover, if $\pwinshort{E} \neq \emptyset$ then every candidate from~$\pwinshort{E}$ has~$\varphi$ score exactly~$s$ in~$E$.
We usually drop~$k$ from~ the notations~$\swinspe{E}{\varphi}{k}$, $\slosespe{E}{\varphi}{k}$, and $\pwinspe{E}{\varphi}{k}$, because the value of~$k$ is usually clearly in the context. On top of that, if it is clear from the context which rule~$\varphi$ or which election~$E$ are discussed, we also drop~$\varphi$, or~$E$, or both from these notations.
It is easy to see that, for an additive {\abmv} rule~$\varphi$, a $k$-committee $\w\subseteq C$ is a winning $k$-committee of~$\varphi$ at $(C, V)$ if and only if $\swinshort{E} \subseteq \w \subseteq (\swinshort{E} \cup \pwinshort{E})$.
It is obvious that for AV, SAV, and NSAV, given an election~$E$, $\swinshort{E}$, $\pwinshort{E}$, and $\sloseshort{E}$ can be computed in polynomial time.

\subsection{Problem Formulations}
Now we formulate the manipulation and control problems studied in this paper. Let~$\varphi$ be a multiwinner voting rule.

\subsubsection{The Manipulation Problems}
\label{subsec-manipulation}
\EP
{\prob{Cardinality-Based Coalition Manipulation} for~$\varphi$}{\probb{CBCM}{$\varphi$}}
{A set of candidates~$C$, two multisets~$V$ and~$V_{\text{M}}$ of votes over~$C$, a winning committee $\w\in \varphi(C, V\cup V_{\text{M}}, k)$.}
{Is there a multiset~$U$ of~$\abs{V_{\text{M}}}$ votes over~$C$ such that for all $\w'\in \varphi(C, V\cup U, k)$ and all~$v\in V_{\text{M}}$, it holds that~$\abs{v\cap \w'}>\abs{v\cap \w}$?}

If we replace $\abs{v\cap \w'}>\abs{v\cap \w}$ in the above definition by $(v\cap \w)\subsetneq (v\cap \w')$, we obtain {\prob{Subset-Based Coalition Manipulation}} for~$\varphi$ ({\probb{SBCM}{$\varphi$}}). In the definitions of \probb{CBCM}{$\varphi$} and {\probb{SBCM}{$\varphi$}}, votes in~$V_{\text{M}}$ are called manipulative votes. Particularly, each $v\in V_{\text{M}}$ consists of candidates whom the corresponding voter (manipulator) truthfully approves. 

The above manipulation problems are relevant to the setting of iterative voting. In this setting, after voters submit their preferences to a central platform, a winning $k$-committee is announced. After this, voters are allowed to change their preferences at will in several rounds. The above manipulation problems model the situation where in a particular round of the voting process some voters (manipulators) form a coalition and want to replace the current winning $k$-committee~$\w$ with another more favorable one by misreporting their preferences. In cases where the tie-breaking schemes are publicity unknown, or when ties are broken randomly, the manipulators may want to ensure that their coordination results in an improved result without taking any risk. This is why in the question we demand that every manipulator prefers all winning $k$-committees of the new election to~$\w$. The problems assume that except the manipulators, other voters do not change their preferences, and manipulators know the submitted preferences of other voters. This may not be very realistic all the time. However, our main message of the study is that even when this is the case, for many rules the manipulators are faced with a computationally hard problem to solve.

Now we give the notion of stochastic domination introduced by Lackner and Skowron~\cite{DBLP:conf/ijcai/LacknerS18}, and the respective manipulation problem.

Let~$C$ be a set of candidates, $S\subseteq C$, and let~$\mathcal{A}$ and~$\mathcal{B}$ be two collections of committees of~$C$. We say that~$\mathcal{A}$ {\memph{stochastically dominates}}~$\mathcal{B}$ subject to~$S$ if and only if for every integer~$i$ it holds that
\[\frac{\abs{\w\in \mathcal{A} \setmid \abs{\w\cap S}\geq i}}{\abs{\mathcal{A}}}\geq \frac{\abs{\w\in \mathcal{B} \setmid \abs{\w\cap S}\geq i}}{\abs{\mathcal{B}}},\]
and, moreover, there exists at lest one~$i$ for which the above inequality is strict.

Based on the above notion, a natural manipulation problem can be formally defined as follows.

\EP
{\prob{SD-Coalition Manipulation} for~$\varphi$}{\probb{SDCM}{$\varphi$}}
{A set of candidates~$C$, two multisets~$V$ and~$V_{\text{M}}$ of votes over~$C$, and an integer $k\leq \abs{C}$.}
{Is there a multiset~$U$ of~$\abs{V_{\text{M}}}$ votes such that $\varphi(C, V\cup U, k)$ stochastically dominates $\varphi(C, V\cup V_{\text{M}}, k)$ subject to  every $v\in V_{\text{M}}$?}

Notice that in the above manipulation problem we do not have a winning $k$-committee in the input.

\onlyfull{
Now we introduce the bribery problems.
For an election~$(C, V)$, the current winning $k$-committee~$\w$ of~$(C, V)$, and a subset~$\w'\subseteq C$, we say a vote~$v\in V$ is radically (respectively, conservatively)~$\w'$-{\it{vulnerable}} at~$(C, V)$ if~$\abs{v\cap \w'}>\abs{v\cap \w}$ (respectively, $(v\cap \w)\subsetneq (v\cap \w')$).

\EP
{Radical/Conservative Frugal Bribery}{RCFB/CCFB}
{An election $(C, V)$, the current winning $k$-committee~$\w\in \varphi(C, V, k)\subseteq C$ such that~$\abs{w}=k$, a nonempty subset $J\subseteq C$ of at most~$k$ candidates.}
{Is there a subset $\w'\subseteq C$ such that $J\subseteq \w'$ and there exists a change of the votes cast by the radically/conservatively~$\w'$-vulnerable voters so that~$\w'$ becomes the unique winning $k$-committee?}
\medskip
}

\subsubsection{The Election Control Problems}
Now we extend the definitions of four standard single-winner election control problems to multiwinner voting.
These control problems model the scenario where some {\fnotion{external agent}} (e.g., the election chair) aims to make some distinguished candidates the winners by modifying the election.

\EP
{\prob{Constructive Control by Adding Voters} for~$\varphi$}{\probb{CCAV}{$\varphi$}}
{A set of candidates~$C$, two multisets~$V$ and~$U$ of votes over~$C$, a positive integer~$k\leq \abs{C}$, a nonempty subset~$J\subseteq C$ of at most~$k$ distinguished candidates, and a nonnegative integer~$\ell$.}
{Is there~$U'\subseteq U$ such that~$|U'|\leq \ell$ and~$J$ belongs to all winning $k$-committees of~$\varphi$ at $(C, V\cup U', k)$?}

In the above definition, votes in~$V$ are called {\memph{registered votes}} and votes in~$U$ are called {\memph{unregistered votes}}. Generally speaking, the above problem consists in determining if we can add a limited number of unregistered votes into the multiset of registered votes so that all distinguished candidates are entirely contained in all winning $k$-committees with respect to the final multiset of registered votes.

\EP
{\prob{Constructive Control by Deleting Voters} for~$\varphi$}{\probb{CCDV}{$\varphi$}}
{A set of candidates~$C$, a multiset of votes~$V$ over~$C$, a positive integer~$k\leq \abs{C}$, a nonempty subset~$J\subseteq C$ of at most~$k$ distinguished candidates, and a nonnegative integer~$\ell$.}
{Is there~$V'\subseteq V$ such that~$\abs{V'}\leq \ell$ and~$J$ belongs to all winning $k$-committees of~$\varphi$ at $(C, V\setminus V', k)$?}

Generally speaking, {\prob{CCDV}} aims to determine if we can delete a limited number of votes so that all distinguished candidates are contained in all winning $k$-committees of the resulting election.

\EP
{\prob{Constructive Control by Adding Candidates} for~$\varphi$}{\probb{CCAC}{$\varphi$}}
{Two disjoint subsets~$C$ and~$D$ of candidates, a multiset~$V$ of votes over $C\cup D$, a positive integer $k\leq \abs{C}$, a nonempty subset $J\subseteq C$ of at most~$k$ distinguished candidates, and a nonnegative integer~$\ell$.}
{Is there~$D'\subseteq D$ of at most~$\ell$ candidates such that~$J$ belongs to all winning $k$-committees of~$\varphi$ at $(C\cup D', V, k)$?}

In the above definitions, we call candidates in~$C$ {\memph{registered candidates}} and call candidates in~$D$ {\memph{unregistered candidates}}. In plain words, the above problem determines if we can add a limited number of unregistered candidates into the set of registered candidates so that all distinguished candidates are contained in all winning $k$-committees.

\EP
{\prob{Constructive Control by Deleting Candidates} for~$\varphi$}{\probb{CCDC}{$\varphi$}}
{A set of candidates~$C$, a multiset of votes~$V$ over~$C$, a positive integer~$k\leq \abs{C}$, a nonempty subset~$J\subseteq C$ of at most~$k$ distinguished candidates, and a nonnegative integer~$\ell$.}
{Is there a subset~$C'\subseteq C\setminus J$ of at most~$\ell$ candidates such that~$\abs{C\setminus C'}\geq k$ and~$J$ belongs to all winning $k$-committees of $\varphi(C\setminus C', V, k)$?}

In the above control problems, the goal of the external agent is to incorporate the distinguished candidates into all winning $k$-committees.
This is natural in many situations. For example, when the external agent is unaware of the tie-breaking scheme, or when a randomized tie-breaking scheme is used, the external agent may want to securely ensure that her favorite candidates become winners.
{The above situation also motivates our requirement on~$\w'$ in the manipulation problem\onlyfull{ and bribery problems}.} \onlyfull{Especially, in the bribery problem, the briber wants to first figure out a committee~$\w'$ to determine the $\w'$-vulnerable voters and, moreover, guarantees these voters that once they follow her guidance~$\w'$ will be the winning $k$-committee regardless of tie-breaking schemes.}

Notice that {\probb{CCAV}{$\varphi$}}, {\probb{CCDV}{$\varphi$}}, {\probb{CCAC}{$\varphi$}}, and {\probb{CCDC}{$\varphi$}} restricted to~$\abs{J}=k=1$ are exactly the extensively studied constructive control (by adding/deleting voters/candidates) problems for single-winner voting (more precisely, they are the unique-winner models of these constructive control problems)~\cite{Bartholdi92howhard,DBLP:journals/jair/FaliszewskiHHR09,handbookofcomsoc2016Cha7FR,DBLP:conf/atal/Yang17}.

A multiwinner voting rule is {\it{immune}} to a constructive control type (adding/deleting voters/candidates) if it is impossible to make any $J\subseteq C$, which is not contained in the current winning $k$-committee, be included in all winning $k$-committees, by performing the corresponding operation. A multiwinner voting rule is {\it{susceptible}} to a control type if it is not immune to this control type.

\subsubsection{Two Important Special Cases}
\label{sec-two-speical-cases}
Now we discuss two interesting special cases of the problems defined in the previous section. The reasons that we separately define them are as follows. First, we believe that they are of independent interest. Second, several of our hardness results are established for the special cases. Third, several our polynomial-time algorithms and {\fpt}-algorithms rely on algorithms solving these special cases.

\EPP{{\probb{$J$-CC}{$\varphi$}}}
{An election $(C, V)$, an integer~$k\leq \abs{C}$, and a subset~$J$ of at most~$k$ candidates.}
{Is~$J$ contained in all winning $k$-committees of~$\varphi$ at $(C, V)$?}

We denote by {\probb{$p$-CC}{$\varphi$}} the special case of {\probb{$J$-CC}{$\varphi$}} where~$J$ is a singleton.
Obviously, {\probb{$p$-CC}{$\varphi$}} is a special case of {\probb{CCAV}{$\varphi$}}, {\probb{CCDV}{$\varphi$}}, {\probb{CCAC}{$\varphi$}}, and {\probb{CCDC}{$\varphi$}}, where $\ell=0$.

For an {\yesins} of the above problems defined in Sections~\ref{subsec-manipulation}--\ref{sec-two-speical-cases}, a  subset that satisfies all conditions given in the corresponding {\textbf{Question}} is referred to as a {{feasible solution}} of the instance.

\onlyfull{{\bf{Parameterized complexity.}} A parameterized problem instance consists of a main part~$I$ and a parameter which is often (but not necessarily to be) an integer~$k$. A parameterized problem is fixed-parameter tractable (\fpt) if every instance~$(I, k)$ of the problem can be solved in time~$f(k)\cdot \abs{I}^{O(1)}$ where~$f(k)$ is a computable function in~$k$. Fixed-parameter intractable problems are categorized into a hierarchy of classes. The most widely studied classes are arguably the {\wah}. Unless parameterized complexity collapses at some level, {\wah} problems do not admit {\fpt}-algorithms.
We assume the reader is familiar with the basics in parameterized complexity such as fixed-parameter tractability ({\fpt}).
We refer\onlyfull{ the reader} to~\cite{DBLP:books/sp/CyganFKLMPPS15} for detailed discussion of the parameterized complexity.}

\subsection{Useful Hardness Problems}
Our hardness results are based on reductions from the following problems. Let~$G$ be a graph, and let~$\vset$ be the vertex set of~$G$. A {\fnotion{vertex cover}} of~$G$ is a subset of vertices whose removal results in a graph without any edge. A subset~$S\subseteq N$ is an {\fnotion{independent set}} in~$G$ if $\vset\setminus S$ is a vertex cover of~$G$. A {\fnotion{clique}} in~$G$ is a subset of pairwise adjacent vertices.

\EPP
{\prob{Vertex Cover}}
{A graph~$G=(N, A)$ where~$N$ is the set of vertices and~$A$ is the set of edges of~$G$, and an integer~$\kappa$.}
{Does~$G$ have a vertex cover of~$\kappa$ vertices?}

It is well-known that {\prob{Vertex Cover}}  restricted to graphs of maximum degree three is {\nph}. Indeed, the {\nphns} remains even when the input graphs are $3$-regular planar graphs~\cite{DBLP:journals/tcs/GareyJS76,DBLP:journals/jct/Mohar01}.

\EP
{\prob{Restricted Exact Cover by Three Sets}}{\prob{RX3C}}
{A universe $\xs=\{\xse_1,\xse_2,\dots,\xse_{3\kappa}\}$ and a multiset $\xc=\{\xce_1,\xce_2,\dots,\xce_{3\kappa}\}$ of $3$-subsets of~$\xs$ such that every~$\xse\in \xs$ occurs in exactly three elements of~$\xc$.}
{Is there~$\xc'\subseteq \xc$ such that~$\abs{\xc'}=\kappa$ and every~$\xse\in \xs$ occurs in exactly one element of~$\xc'$?}

In an instance $(\xs, \xc)$ of {\prob{RX3C}}, we say that an $\xce\in \xc$ covers $\xse\in \xs$ if $\xse\in \xce$, and say that an $\xc'\subseteq \xc$ covers some $\xs'\subseteq \xs$ if every $\xse\in \xs'$ is covered by at least one element of~$\xc'$. Therefore, the problem determines if~$\xc$ contains an~$\xc'$ of cardinality~$\kappa$ which covers~$\xs$.
It is known that {\prob{RX3C}} is {\nph}~\cite{DBLP:journals/tcs/Gonzalez85}.

\EPP
{\prob{Clique}}
{A graph $G=(\vset, \eset)$ and an integer~$\kappa$.}
{Does~$G$ have a clique of~$\kappa$ vertices?}

\EPP
{\prob{Independent Set}}
{A graph $G=(\vset, \eset)$ and an integer~$\kappa$.}
{Does~$G$ have an independent set of~$\kappa$ vertices?}

It is known that {\prob{Clique}} and {\prob{Independent Set}} are {\nph} even when restricted to regular graphs~\cite{DBLP:journals/cj/Cai08,DBLP:conf/iwpec/Marx04,DBLP:conf/cats/MathiesonS08}.

\subsection{Parameterized Complexity}
A {\fnotion{parameterized problem}} is a subset of $\Sigma^*\times \mathbb{N}$, where~$\Sigma$ is a finite alphabet. Parameterized problems can be grouped into different classes. The arguably most  thought-after classes are those in the following hierarchy:
\[\fpt\subseteq \wa\subseteq \wb \subseteq \cdots \subseteq \xp.\]
Particularly, a parameterized problem is {\fnotion{fixed-parameter tractable}} ({\fpt}) if there is an algorithm which correctly determines for each instance $(I, \kappa)$ of the problem whether
$(I, \kappa)$ is a {\yesins} in time $\bigo{f(\kappa)\cdot \abs{I}^{\bigo{1}}}$, where~$f$ is a computable function and~$\abs{I}$ is the size of~$I$.
A parameterized problem is in the class {\xp} if there is an algorithm which correctly determines for each instance $(I, \kappa)$ of the problem whether
$(I, \kappa)$ is a {\yesins} in time $\bigo{\abs{I}^{f(\kappa)}}$, where~$f$ is a computable function in~$\kappa$.
A parameterized problem is {\wah} if all problems in {\wa} are parameter reducible to it.
Unless {\fpt}$ = ${\wa} which is widely believed to be unlikely, {\wah} problems do not admit any {\fpt}-algorithms. Parameterized problems could also go beyond {\xp}.
A problem is {\paranph} if there is a constant~$c$ so that for every constant $k\geq c$ the problem is {\nph} when fixing~$k$ as the parameter.
For greater details on parameterized complexity theory, we refer to~\cite{DBLP:books/sp/CyganFKLMPPS15}.

\subsection{Remarks}
It should be pointed out that our hardness results also hold for corresponding multiwinner voting rules which always select exactly one winning $k$-committee by utilizing a specific tie-breaking scheme (see, e.g.,~\cite{DBLP:journals/aamas/BredereckKN21} for descriptions of some tie-breaking schemes). 
In fact, our hardness reductions are established carefully to avoid ties. \onlyfull{For example, the prevalent lexicographic tie-breaking scheme selects the lexicographically smallest tied committees with respect to a fixed order on candidates.
Other widely studied tie-breaking schemes include the one against the strategic agent and the one in favor of the strategic agent. When several committees are tied, the former one select one including the minimum number of distinguished candidates and the latter one select.}%
Moreover, our {\poly}- and {\fpt}-algorithms for the additive rules can be adapted for these variants when ties are broken lexicographically.

\section{Manipulation\onlyfull{ and Bribery}}
\label{sec-manipulation}
In this section, we \onlyfull{study the manipulation and bribery problems formulated in Preliminaries. We first }consider the manipulation problems and show that in general these problems are {\nph}.

\begin{theorem}
\label{thm-ccm-av-np-hard}
{\probb{\cbcm}{AV}}, {\probb{\sbcm}{AV}}, and {\probb{\sdcm}{AV}} are {\memph{{\nph}}}. Moreover, this holds even if there are four nonmanipulative votes.
\end{theorem}

\begin{proof}
We give a reduction from {\prob{Vertex Cover}} that applies to all of {\probb{\cbcm}{AV}}, {\probb{\sbcm}{AV}}, and {\probb{\sdcm}{AV}}.
Let $(G, \kappa)$ be a {\prob{Vertex Cover}} instance where~$G=(\vset, \eset)$ is a $3$-regular graph. Let~$n=\abs{\vset}$ and~$m=\abs{\eset}$. Without loss of generality, we assume that~$m>4$ (otherwise the instance can be solved in polynomial time). For each vertex~$\vere\in \vset$, we create one candidate denoted still by~$\vere$ for simplicity. Additionally, we create~$\kappa$ candidates~$c_1$,~$c_2$,~$\dots$,~$c_{\kappa}$. Let~$C$ be the set of the above~$n+\kappa$ created candidates. We create the following votes. First, we create four nonmanipulative votes in~$V$ each of which approves exactly the~$\kappa$ candidates~$c_1$, $\dots$,~$c_{\kappa}$. In addition, for each edge~$\edge{b}{b'}\in A$, we create one manipulative vote~$v(b,b')=\{b,b'\}$ in~$V_{\text{M}}$. Finally, we set~$k=\kappa$. Let $\w=\{c_1,\dots,c_{\kappa}\}$. Then, $(C, V, V_{\text{M}}, \w)$ is the instance of both {\probb{\cbcm}{AV}} and {\probb{\sbcm}{AV}}, and $(C, V, V_{\text{M}})$ is the instance of {\probb{\sdcm}{AV}}. Let $E=(C, V\cup V_{\text{M}})$. As~$G$ is $3$-regular, every candidate in~$\vset$ has AV score~$3$ in~$E$. As every candidate in~$\w$ has AV score~$4$ in~$E$,~$\w$ is the unique winning $k$-committee of AV at~$E$.

Now we prove the correctness of the reduction. Notice that~$\w$ does not include any candidate approved by any manipulators. As a consequence, the instance of {\probb{\cbcm}{AV}} and {\probb{\sbcm}{AV}} (respectively, {\probb{\sdcm}{AV}}) is a {\yesins} if and only if the manipulators can coordinate their votes so that every (at least one) AV winning $k$-committee of the resulting election contains at least one approved candidate of every manipulative vote.

$(\Rightarrow)$ Suppose that~$G$ has a vertex cover~$S\subseteq N$ of~$\kappa$ vertices. If all manipulators turn to approve exactly the candidates corresponding to~$S$, the AV score of each candidate in~$S$ increases from~$3$ to~$m>4$, and hence~$S$ forms the unique winning $k$-committee. As~$S$ is a vertex cover of~$G$, by the construction of the votes, for every manipulative vote $v(b,b')=\{b,b'\}$ where $\edge{b}{b'}\in A$, at least one of~$b$ and~$b'$
 is included in~$S$.

$(\Leftarrow)$ If there exists no vertex cover of~$\kappa$ vertices in~$G$, then no matter which~$\kappa'\leq \kappa$ candidates from~$N$ are in a winning $k$-committee after the manipulators change their votes, there is at least one manipulative vote~$v(b, b')$, $\edge{b}{b'}\in A$, such that neither~$b$ nor~$b'$ is contained in this winning~$k$-committee.
\end{proof}

Similar to the reduction in the proof of Theorem~\ref{thm-ccm-av-np-hard}, we can show the {\nphns} of {\cbcm}, {\sbcm}, and {\sdcm} for SAV. Clearly, after all manipulators change to approve candidates corresponding to a vertex cover~$S$ of~$\kappa$ vertices, the SAV score of each candidate in~$S$ increases from~$3/2$ to~$m/\kappa$. Given this, to prove the {\nphns} of {\cbcm}, {\sbcm}, and {\sdcm} for SAV, we need only to create~$m-5$ further nonmanipulative votes approving~$c_1$,~$\dots$,~$c_{\kappa}$ so that in the original election each~$c_i$ has SAV score~$(m-1)/\kappa$.\footnote{The correctness of the reduction also relies on the assumption  $\kappa<\frac{2}{3}\cdot (m-1)$ which does not change the {\nphns} of {\prob{Vertex Cover}} restricted to $3$-regular graphs~\cite{DBLP:journals/tcs/GareyJS76,DBLP:journals/jct/Mohar01}.}
\onlyfull{One can check that the reduction also applies to NSAV if we assume that~$\kappa$ is substantially smaller than~$m$, say~$\kappa^5\leq m$, which does not change the complexity of the {\prob{Vertex Cover}} problem.} It is easy to see that the same reduction works for NSAV too. We arrive at the following theorem.

\begin{theorem}
\label{thm-manipulation-np-hard-many-rules}
For $\varphi\in \{\memph{\text{SAV}}, \memph{\text{NSAV}}\}$, {\probb{\cbcm}{$\varphi$}}, {\probb{\sbcm}{$\varphi$}}, and {\probb{\sdcm}{$\varphi$}} are {\memph{\nph}}.
\end{theorem}

By establishing a different reduction, we can show similar hardness results for MAV.

\begin{theorem}
\label{thm-manipulation-mav-np-hard}
{\probb{\cbcm}{MAV}}, {\probb{\sbcm}{MAV}}, and {\probb{\sdcm}{MAV}} are {\memph{\nph}}, even if there is only one nonmanipulator, and every candidate is approved by at most three votes.
\end{theorem}

\begin{proof}
Our proof is based on a reduction from {\prob{Vertex Cover}} which applies to  all of {\probb{\cbcm}{MAV}}, {\probb{\sbcm}{MAV}}, and {\probb{\sdcm}{MAV}}.
Let~$(G, \kappa)$ be an instance of {\prob{Vertex Cover}} where~$G=(N, A)$ is a~$3$-regular graph. Without loss of generality, we assume that $|A|>\kappa$.

The candidates are as follows. First, for each~$a\in N$, we create one candidate denoted still by~$a$ for simplicity. Then, we create a set~$X$ of~$2\kappa+1$ candidates. In addition, for each edge~$\edge{a}{b}\in A$, we create a set~$Y(\{a,b\})$ of~$3\kappa+1$ candidates. Let $C=\bigcup_{\edge{a}{b}\in A}Y(\edge{a}{b})\cup N\cup X$.  
Now we  construct the votes. First, for each edge~$\edge{a}{b}\in A$, we create one vote~$v(\edge{a}{b})=\{a,b\}$. These votes are the manipulative votes, i.e., we let~$V_{\text{M}}=\{v(a,b) \setmid \edge{a}{b}\in A\}$. In addition, we create one nonmanipulative vote approving exactly the~$2\kappa+1$ candidates in~$X$.
Let~$k=\kappa$, and let~$\w$ be any arbitrary $k$-subset of~$X$. The instances of {\probb{\cbcm}{MAV}} and {\probb{\sbcm}{MAV}} are the same $(C, V, V_{\text{M}}, \w)$, and the instance of {\probb{\sdcm}{MAV}} is $(C, V, V_{\text{M}})$. The constructions of the instances clearly can be done in polynomial time.

One can check that winning $k$-committees of MAV at $(C, V\cup V_{\text{M}})$ are exactly all $k$-subsets of~$X$. In fact, any such a $k$-committee is of Hamming distance $k+2$ from every manipulative vote, and is of Hamming distance $k+1$ from the nonmanipulative vote. So, the MAV score of such a committee is $k+2$. In addition, every $k$-committee which contains at most $k-1$ candidates from~$X$ is of Hamming distance at least $1+(k+2)=k+3$ from the nonmanipulative vote, and hence has MAV score at least $k+3$.

It remains to show the correctness of the reduction. Note that as none of the winning $k$-committees of~$(C, V\cup V_{\text{M}})$ intersects any manipulative votes, the instance of {\probb{\cbcm}{MAV}} and {\probb{\sbcm}{MAV}} (respectively, {\probb{\sdcm}{MAV}}) is a  {\yesins} if and only if the manipulators can coordinate their votes so that every (respectively, at least one) winning $k$-committee of the resulting election intersects every manipulative vote from~$V_{\text{M}}$.

$(\Rightarrow)$ Assume that there is a vertex cover~$S\subseteq N$ of~$\kappa$ vertices in~$G$. Then, we change each manipulative vote~$v(\{a,b\})$ so that after the change it exactly approves all candidates in~$S\cup Y(a,b)$. For the sake of clarity, we use $v'(\{a, b\})$ to denote $v(\{a, b\})$ after the change. Let~$E$ denote the election after these changes. One can check that~$S$ has an MAV score~$3\kappa+1$ in~$E$. We complete the proof by showing that~$S$ is the unique winning $k$-committee of~$E$, i.e., any other~$k$-committee has Hamming distance at least~$3\kappa+2$ in~$E$. Consider first a $k$-committee which contains some candidate in $\bigcup_{\edge{a}{b}\in A}Y(\edge{a}{b})$. Due to the assumption $|A|>k$, there exists an edge $\edge{a'}{b'}\in A$  such that none of $Y(\edge{a'}{b'})$ is in this $k$-committee. It is easy to check that the Hamming distance between this $k$-committee and the vote $v'(\edge{a'}{b'})$ is at least $1+1+(3k+1)=3k+3$. It remains to consider $k$-committees which do not contain any candidate from $\bigcup_{\edge{a}{b}\in A}Y(\edge{a}{b})$ but contain some candidate from~$X$. The Hamming distance between such a $k$-committee and every manipulative vote is at least $(3k+1)+1+1=3k+3$. It follows that~$S$ is the unique winning $k$-committee of~$E$. Finally, as~$S$ intersects each of the manipulative votes, we can conclude that in this case the constructed instance of {\cbcm}/{\sbcm}/{\sdcm} is a {\yesins}.

$(\Leftarrow)$ Similar to the argument in the proof of Theorem~\ref{thm-ccm-av-np-hard}, if there exists no vertex cover of~$\kappa$ vertices in the graph~$G$, for any $k$-committee $\w'\subseteq C$, at least one of the manipulative votes is disjoint from~$\w'$. So, the constructed {\cbcm}/{\sbcm}/{\sdcm} instance is a {\noins}.
\end{proof}

In the proofs of Theorems~\ref{thm-ccm-av-np-hard}--\ref{thm-manipulation-mav-np-hard}, manipulators considerably outnumber nonmanipulators. This stands in contrast to the polynomial-time solvability of the canonical coalition manipulation problem (see~\cite{handbookofcomsoc2016Cha7FR} for the definitions) for many ranking-based single-winner voting rules, where, when there are more manipulators than nonmanipulators, the manipulators can always make the distinguished candidate the winner by ranking the distinguished candidate in the top and ranking other candidates greedily. On the other hand, the canonical coalition manipulation problem for many single-winner voting rules is already {\nph} even when there is only one or two manipulators~\cite{BARTHOLDI89,DBLP:journals/ai/DaviesKNWX14}. However, we show that this is not always the case for our problems. In particular, we show that {\cbcm} and {\sbcm} for AV are polynomial-time solvable when the number of manipulators is a constant.


%



\begin{theorem}
\label{thm-maniuplation-polynomial-time-solvable-constant-number-manipulators}
{\probb{\cbcm}{AV}} and {\probb{\sbcm}{AV}} are polynomial-time solvable if there are a constant number of manipulators.
\end{theorem}

\begin{proof}
We consider first {\probb{\cbcm}{AV}}. Let $I=(C, V, V_{\text{M}}, \w)$ be an instance of {\probb{\cbcm}{AV}}, where $\w\subseteq C$ is a winning $k$-committee of AV at $(C, V)$, and~$V_{\text{M}}$ is the multiset of manipulative votes. Let $m=\abs{C}$ denote the number of candidates, let $k=\abs{\w}$ be the cardinality of~$\w$, and let~$t=\abs{V_{\text{M}}}$ be the number of manipulators which is a constant.
We derive a polynomial-time algorithm by splitting the instance~$I$ into polynomially many subinstances, so that the original instance~$I$ is a {\yesins} if and only if at least one of the subinstances is a {\yesins}. 
Before presenting the algorithm, we first study a property of the solution space of {\probb{\cbcm}{AV}}.

\begin{claimm}
\label{claim-b}
If~$I$ is a {\yesins}, then~$I$ has a feasible solution so that all manipulators turn to approve the same candidates  which are all from~$C^{\vee}(V_{\memph{\text{M}}})$, i.e., there exists a subset $S\subseteq C^{\vee}(V_{\memph{\text{M}}})$ so that every $v\in V_{\memph{\text{M}}}$ prefers every AV winning $k$-committee of $(C, V\cup V')$ to~$\w$, where~$V'$ is the multiset of~$t$ votes each of which approves exactly the candidates from~$S$.
\end{claimm}%
\smallskip

{\noindent{\textit{Proof of Claim~\ref{claim-b}.}}}
Assume that~$I$ is a {\yesins}, and let~$V'$ be a multiset of~$t$ votes over~$C$, one-to-one corresponding to votes of~$V_{\text{M}}$, so that every $v\in V_{\text{M}}$ prefers every AV winning $k$-committee of $E=(C, V\cup V')$ to~$\w$. For every $v\in V_{\memph{\text{M}}}$, let~$v'$ denote the vote in~$V'$ corresponding to~$v$. 
If all votes from~$V'$ approve the same candidates and they are all from~$C^{\vee}(V_{\text{M}})$, we are done. Otherwise, one of the two cases described below may occur.
\begin{description}
\item[Case~1:] $\exists v'\in V'$ so that~$v'$ approves at least one candidate $c\in C\setminus C^{\vee}(V_{\text{M}})$.

Let $\tilde{v}=v'\setminus \{c\}$, and let~$E'$ be the election obtained from~$E$ by replacing~$v'$ with~$\tilde{v}$. Now we prove that every $v\in V_{\text{M}}$ prefers every AV winning $k$-committee of~$E'$ to~$\w$. Let~$\w'$ be a winning $k$-committee of~$E'$. The proof proceeds by distinguishing the following cases.
\begin{itemize}
\item $c\in \swinshort{E}$

In this case, if $c\in \swinshort{E'}$, the AV winning $k$-committees of~$E$ and those of~$E'$ are the same; we are done. If $c\in \pwinshort{E'} \cup \sloseshort{E'}$, then $\swinshort{E'} \subseteq \swinshort{E}\cap \w'$ and $\w'\setminus \swinshort{E'} \subseteq \pwinshort{E'}$.

\begin{itemize}
\item If $\pwinshort{E}=\emptyset$ (equivalently, $\abs{\swinshort{E}}=k$), then~$\pwinshort{E'}$ is composed of~$c$ and a subset $B\subseteq \sloseshort{E}$.
Then, if $c\in \w'$,~$\w'$ is also an AV winning $k$-committee of~$E$, and hence every $v\in V_{\text{M}}$ prefers~$\w'$ to~$\w$.
Otherwise, let~$\{c'\}=\w'\setminus \swinshort{E'}$. Let $w''=\w' \setminus \{c'\}\cup \{c\}$. Then,~$\w''$ is the AV winning $k$-committee of~$E$ and hence is more preferred to~$\w$ by every $v\in V_{\text{M}}$. That is, for every $v\in V_{\text{M}}$ it holds that $\abs{v\cap \w''}>\abs{v\cap w}$. As $c\not\in C^{\vee}(V_{\text{M}})$, it follows that $\abs{v\cap \w'}\geq \abs{v\cap \w''}>\abs{v\cap \w}$ holds for every $v\in V_{\text{M}}$. 

\item Otherwise, it must hold that $\abs{\pwinshort{E}}\geq 2$ and $\pwinshort{E'}=\pwinshort{E} \cup \{c\}$. If $c\in \w'$, $\w'$ is also an AV winning $k$-committee of~$E$, and we are done. Otherwise, let~$c'$ be any arbitrary candidate from $\w'\cap \pwinshort{E}$. Then, similar to the above analysis,  we can show that every $v\in V_{\text{M}}$ prefers~$\w'$ to~$\w$ (by using $\w''=\w'\setminus \{c'\}\cup \{c\}$ as an intermediate committee).
\end{itemize}

\item $c\in \pwinshort{E}$

As $\pwinshort{E} \neq \emptyset$, we know that $\abs{\pwinshort{E}}\geq 2$ and $\abs{\swinshort{E}\cup \pwinshort{E}}>k$. It follows then $c\in \sloseshort{E'}$ and, more importantly, $(\swinshort{E'} \cup \pwinshort{E'})\subseteq (\swinshort{E}\cup \pwinshort{E})$, which implies that the AV winning $k$-committees of~$E'$ is a subcollection of the AV winning $k$-committees of~$E$. As a consequence, every $v\in V_{\text{M}}$ prefers every AV winning $k$-committee of~$E'$ to~$\w$.

\item $c\in \sloseshort{E}$

In this case, $(\swinshort{E'} \cup \pwinshort{E'})= (\swinshort{E} \cup \pwinshort{E})$. In other words, the AV winning $k$-committees of~$E$ and those of~$E'$ coincide. As a result, every $v\in V_{\text{M}}$ prefers every AV winning $k$-committee of~$E'$ to~$\w$.
\end{itemize}

\item[Case~2:] $v'\subseteq C^{\vee}(V_{\text{M}})$ for all $v'\in V'$ but votes in~$V'$ do not approve the same candidates.

In the following, we compute a subset~$S\subseteq C^{\vee}(V_{\text{M}})$  so that if all manipulators turn to approve exactly the candidates in~$S$, every manipulator prefers every winning $k$-committee of the resulting election to~$\w$. First, let $S=C^{\vee}(V')\setminus \sloseshort{E}$. We consider two cases. 
\begin{itemize}
\item $\abs{S}\leq k-\abs{\swinshort{E}}$

 In this case, we let~$E'$ be the election obtained from~$E$ by replacing every $v'\in V'$ with a vote approving exactly the candidates from~$S$. Then, we have that $\swinshort{E'}\subseteq (\swinshort{E}\cup S) \subseteq (\swinshort{E} \cup \pwinshort{E})$, and $\pwinshort{E'} \subseteq \pwinshort{E}$. So, the AV winning $k$-committees of~$E'$ is a subcollection of AV winning $k$-committees of $E$; we are done.

\item $\abs{S}> k-\abs{\swinshort{E}}$

In this case, let~$S'$ be any subset of arbitrary~$k-\abs{\swinshort{E}}$ candidates from $C^{\vee}(V')\cap \pwinshort{E}$. Then, we reset $S:=(S\cap \swinshort{E})\cup S'$. After all manipulators turn to approve candidates in~$S$, the AV winning $k$-committees of the resulting election is a subcollection of those of~$E$; we are done.
\end{itemize}
\end{description}

By the above analysis, if there exists at least one $v'\in V'$ which approves at least one candidate $c\in C\setminus C^{\vee}(V_{\text{M}})$, we can removing~$c$ from the vote~$v'$ with the guarantee that all manipulators still prefer all AV winning $k$-committees of the resulting election to~$\w$. We iteratively apply such deletions until Case~1 is not satisfied. Then, all votes from $V'$ approve only candidates from $C^{\vee}(V_{\text{M}})$. Then, by the  analysis in Case~2, there exists $S\subseteq C^{\vee}(V_{\text{M}})$ so that after replacing all votes in $V'$ with~$S$, all manipulators prefer all AV winning $k$-committee to $\w$.
This completes the proof of Claim~\ref{claim-b}.
\medskip

Now we start to delineate the algorithm. Recall that for~$U\subseteq V$,~$C_V^{{\star}}(U)$ is the set of candidates from~$C$ approved exactly by votes in~$U$ among all votes in~$V$. Clearly, for distinct $U, U'\subseteq V$, it holds that $C_V^{{\star}}(U)\cap C_V^{{\star}}(U')=\emptyset$. In the remainder of the proof, for $S\subseteq V_{\text{M}}$, we use~$C^{\star}(S)$ to denote~$C^{\star}_{V_{\text{M}}}(S)$, for notational brevity. For each nonempty~$S\subseteq V_{\text{M}}$, we guess a nonnegative integer~$x_{S}\leq \abs{C^{\star}(S)}$ which indicates the number of candidates from~$C^{\star}(S)$ that are approved by all manipulators in the final election. Therefore, we guess in total at most~$2^t$ such integers, and there are at most $\prod_{S\subseteq V_{\text{M}}}(1+\abs{C^{\star}(S)})=\bigo{(m+1)^{2^t}}$ different combinations of the guesses. In addition, in light of Claim~\ref{claim-b}, we guess the number~$k'$ of candidates approved by all manipulators in the final election. In effect, these guesses split the original instance~$I$ into $\bigo{m\cdot (m+1)^{2^t}}$ subinstances each of which takes as input~$I$, a positive integer~$k'\leq \abs{C^{\vee}(V_{\text{M}})}$, and a nonnegative integer $x_S\leq \abs{C^{\star}(S)}$ for every nonempty $S\subseteq V_{\text{M}}$, and asks if there is a $k'$-subset $\w'\subseteq C^{\vee}(V_{\text{M}})$ so that the following two conditions hold:
\begin{enumerate}
\item[(1)] for every nonempty $S\subseteq V_{\text{M}}$,~$\w'$ includes exactly~$x_{S}$ candidates from~$C^{\star}(S)$, and
\item[(2)] every manipulator prefers every winning $k$-committee of the election after all manipulators turn to approve exactly the candidates in~$\w'$ to~$\w$.
\end{enumerate}
Now we show how to solve a subinstance associated with~$k'$, and $\{x_{S} \setmid S\subseteq V_{\text{M}}\}$ as given above.
First, if $\sum_{\emptyset\ne S\subseteq V_{\text{M}}} x_S \neq k'$, by Condition~(1), we conclude that the  subinstance is a {\noins}.
Otherwise, we do the following. We let all manipulators approve~$x_{S}$ certain candidates in each~$C^{{\star}}(S)$ where $\emptyset\neq S\subseteq V_{\text{M}}$ and~$x_{S}>0$.
We first prove the following claim. Let $E=(C, V)$. For each $S\subseteq V_{\text{M}}$, let~$\rhd_{S}$ be a linear order on~$C^{\star}(S)$ so that for every $c, c'\in C^{\star}(S)$ it holds that $c \rhd_S c'$ if $\score{V}{c}{E}{\text{AV}} \geq \score{V}{c'}{E}{\text{AV}}$. We use~$\rhd_S[i]$ to denote the $i$-th candidate in~$\rhd_S$. A {\emph{block}} of~$\rhd_{S}$ is set of candidates in~$C^{\star}(S)$ that are consecutive in the order~$\rhd_S$. For two positive integers $i,j$ so that $i\leq j\leq \abs{C^{\star}(S)}$, let $\rhd_S[i,j]=\{\rhd_S[x] \setmid i\leq x\leq j\}$.
\medskip

The following observation is easy to see.

\begin{observation}
\label{obs-a}
Let $X, Y\subseteq C$ be  so that $\abs{X\cap C^{\star}(S)}\geq \abs{Y\cap C^{\star}(S)}$ holds for every $S\subseteq V_{\text{M}}$. Then, for every $v\in V_{\text{M}}$, $v$ prefers~$Y$ to~$\w$ implies that~$v$ prefers~$X$ to~$\w$.
\end{observation}

\begin{claimm}
\label{claim-c}
If~$I$ is a {\yesins},  then it has a feasible solution such that for every $v'\in V'$ and every $S\subseteq V_{\text{M}}$, it holds that $v'$ induces at most two blocks of~$\rhd_S$.
\end{claimm}

{\noindent{\it{Proof of Claim~\ref{claim-c}}}.} We prove the claim by contradiction. Let~$V'$ be a feasible solution of~$I$. Let $E'=(C, V\cup V')$. By Claim~\ref{claim-b}, we may assume that all votes in~$V'$ approve exactly the same candidates. Moreover, we may also assume that none of~$V'$ approves any candidates in $\sloseshort{E'}$, since otherwise we remove from all $v'\in V'$ the candidates in~$\sloseshort{E'}$, and after the removals the AV winning $k$-committees remain the same, and hence every vote in~$V_{\text{M}}$ still prefers every AV winning $k$-committee of the resulting election to~$\w$.
Now if every vote in~$V'$ induces at most two blocks of~$\rhd_S$, we are done. Otherwise, we show below how to transform ~$V'$ into another feasible solution so that the conditions stated in the claim hold.
Let~$v'$ be a vote from~$V'$ which induces at least three blocks of~$\rhd_S$ for some $S\subseteq V_{\text{M}}$. Let~$B_1$,~$B_2$, $\dots$,~$B_z$ be the blocks induced by~$v'$ so that $B_1\rhd_S B_2 \rhd_S\cdots \rhd_S B_z$, where $z\geq 3$ and $B_i \rhd_S B_j$ means $c\rhd_S c'$ for all $c\in B_i$ and all $c'\in B_j$. Let $B=\bigcup_{i=1}^z B_i$. Without loss of generality, let $B_i=\rhd_S[i_{\text{L}}, i_\text{R}]$ where $1\leq i_{\text{L}}\leq i_{\text{R}}\leq \abs{C^{\star}(S)}$.  Obviously, it holds that $\score{V\cup V'}{c}{E'}{\text{AV}} \geq \score{V\cup V'}{c'}{E'}{\text{AV}}$ for any $c, c'\in B$ so that $c\rhd_S c'$.  Our proof proceeds by distinguishing the following cases.
\begin{description}
    \item[Case~1:] $B\subseteq \pwinshort{E'}$ or $B\subseteq \swinshort{E'}$.

    In this case, after replacing candidates from~$B\setminus B_1$ by the $\abs{B\setminus B_1}$ consecutive candidates immediately after~$\rhd_S[1_{\text{R}}]$ in every vote of~$V'$, for every $S'\subseteq V_{\text{M}}$ the number of candidates from~$C^{\star}(S')$ contained in every~AV winning $k$-committee of the resulting election equals that contained in every AV wining $k$-committee of the election before the replacement. Then, by Observation~\ref{obs-a}, after the replacement~$V'$ remains a feasible solution of~$I$.

    \item[Case~2:]  $B\cap \swinshort{E'}\neq \emptyset$ and $B\cap \pwinshort{E'}\neq \emptyset$.

    In this case, let~$c$ be the right-most candidate from~$B\cap \swinshort{E'}$ in~$\rhd_S$, and let~$c'$ be the left-most candidate from~$B\cap \pwinshort{E'}$ in~$\rhd_S$. The following observations are clear:
\begin{itemize}
\item $c\rhd_S c'$, and hence the~AV score~of $c$ is larger than that of~$c'$ in $(C, V)$; and
\item there are no other candidates from~$B$ between~$c$ and~$c'$ in~$\rhd_S$.
\end{itemize}
Let~$x$ be the number of candidates from~$B$ before~$c$ in~$\rhd_S$, and let~$y$ be the number of candidates from~$B$ after~$c'$ in~$\rhd_S$. Then, in every vote in~$V'$ we replace  all candidates in~$B$ before~$c$ by the~$x$ consecutive candidates immediately before~$c$, and replace all candidates in~$B$ after~$c'$ by the~$y$ consecutive candidates immediately after~$c'$.
By Observation~\ref{obs-a}, the new~$V'$ remains a feasible solution of~$I$.
    \end{description}
By the analysis in the above two cases, if in a feasible solution of~$I$ there are votes which induce more than two blocks of~$\rhd_S$, we can transform it into another feasible solution of~$I$ where every vote induces at most two blocks of~$\rhd_S$. This completes the proof of Claim~\ref{claim-c}.
\medskip

Armed with Claim~\ref{claim-c}, for each $S\subseteq V_{\text{M}}$, we guess whether each manipulator's new vote induces one or two blocks, and guess the starting and ending points of the block(s). The number of the combinations of the guesses is bounded from above by $(m^4)^{2^t}$ which is a constant. Each fixed combination of the guesses corresponds to a multiset~$V'$ of~$t$ votes approving the same candidates. Given such a~$V'$ and $E'=(C, V\cup V')$, we compute $\swinshort{E'}$, $\pwinshort{E'}$, and $\sloseshort{E'}$. We conclude that given instance~$I$ is a {\yesins} if and only if there exists at least one such~$V'$ so that every $v\in V_{\text{M}}$ prefers every AV winning $k$-committee of~$E'$ to~$\w$, which can be done by checking if for every $v\in V_{\text{M}}$ whether the following inequality holds:
\[\abs{\swinshort{E'}\cap v}+\max\{0, k+\abs{\pwinshort{E'}\cap v}-\abs{\swinshort{E'}\cup \pwinshort{E'}}\}>\abs{v\cap \w}.\]

The algorithms for {\sbcm} are analogous. We only outline the differences. First, as all approved candidates of a manipulator which are in the original winning~$k$-committee~$\w$ are demanded to be in the final winning~$k$-committee, we can allow first that all manipulators approve all candidates from $\w\cap C^{\vee}(V_{\text{M}})$. Then, we calculate what other candidates from~$C^{\vee}(V_{\text{M}})\setminus \w$ should be approved by the manipulators by guessing an integer $k'\leq k-\abs{\w\cap C^{\vee}(V_{\text{M}})}$, in a way similar to the above algorithm. Third, in this case a vote $v\in V_{\text{M}}$ prefers every AV winning $k$-committee of~$E'$ to~$\w$ if and only if either (1) $\abs{\pwinshort{E'}}=1$ and $\sloseshort{E'}\cap v\cap \w=\emptyset$, or (2) $\abs{\pwinshort{E'}}>1$ and $(\pwinshort{E'}\cup \sloseshort{E'})\cap (v\cap \w)=\emptyset$.
\end{proof}

For SAV and NSAV, we also have polynomial-time algorithms.

\begin{theorem}
\label{thm-maniuplation-sav-nsav-polynomial-time-solvable-constant-number-manipulators}
For $\varphi\in \{\memph{\text{SAV}}, \memph{\text{NSAV}}\}$,
{\probb{\cbcm}{$\varphi$}} and {\probb{\sbcm}{$\varphi$}} are polynomial-time solvable if there are a constant number of manipulators.
\end{theorem}

We defer the proof of Theorem~\ref{thm-maniuplation-sav-nsav-polynomial-time-solvable-constant-number-manipulators} to the Appendix.
The algorithms for SAV and NSAV are similar in principle to the one for AV but with a much larger number of subinstances to solve. The reason is that, unlike AV, for SAV and NSAV it is not always optimal for the manipulators to approve the same candidates, as shown in Example~\ref{ex-1}.
For this reason, instead of guessing only one common integer~$x_S$ for manipulators in~$S$, we need to guess many integers for manipulators in~$S$ separately. Nevertheless, as long as the number of manipulators is a constant, we are still ensured with a polynomially many combinations of guesses.

\begin{example}[Under SAV and NSAV it is not always optimal for the manipulators to approve the same candidates.]
\label{ex-1}
Consider an election~$E$ with nine candidates, seven nonmanipulative votes, and three manipulative votes as shown below, where a check mark means that the corresponding voter approves the corresponding candidate. %
  \begin{center}
    \includegraphics[width=0.65\textwidth]{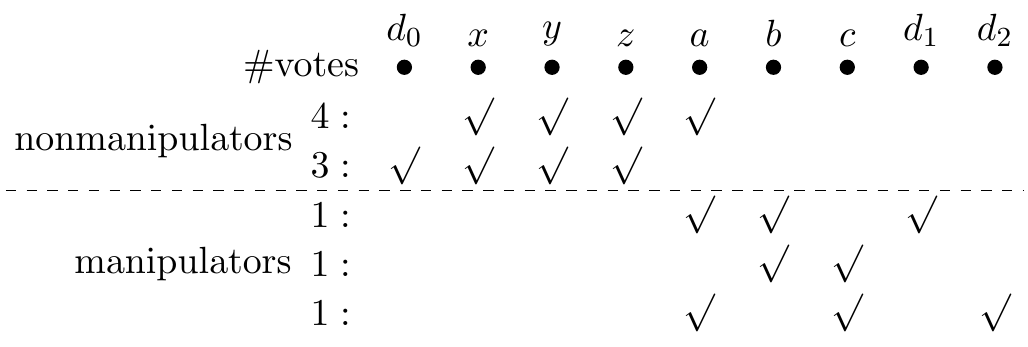}
   \end{center}
Obviously, the SAV score of each of~$x$,~$y$, and~$z$ is~$\frac{7}{4}$, that of~$a$ is $1+\frac{2}{3}=\frac{5}{3}$, that of each of~$b$ and~$c$ is~$\frac{1}{3}+\frac{1}{2}=\frac{5}{6}$, and that of each of the other candidates is strictly smaller than~$1$. Hence, the winning $2$-committees of SAV at~$E$ are exactly the $2$-subsets of $\{x, y, z\}$. To satisfy all manipulators, winning $2$-committees should be among $\{a, b\}$, $\{b, c\}$, $\{a, c\}$, $\{d_1, c\}$, and $\{d_2, b\}$. If all the three manipulators turn to approve the same set of candidates, at least one candidate from $\{b, c, d_1, d_2\}$ has SAV score at most $\frac{3}{2}$, which is strictly smaller than that of each of~$x$,~$y$, and~$z$, implying that the result is not more favorable by at least one manipulator. However, the manipulators are capable of improving the result by coordinating their votes in some other way. For instance, if one of the manipulator approves~$a$, and the other two approve~$b$, the SAV scores of~$a$ and~$b$ both become~$2$, which is strictly larger than that of every other candidate, making $\{a, b\}$ the unique winning $2$-committee. This shows that for SAV it is not always optimal for manipulators to approve the same set of candidates in order to improve the result in their favor.
By adding a large number of dummy candidates not approved by any vote, we can show a similar result for NSAV utilizing Lemma~\ref{lem-relation-sav-nsav} below.
\end{example}


For an election $(C, V)$ and a candidate~$c$, let $\textsf{SAV}_{(C, V)}(c)$ and $\textsf{NSAV}_{(C, V)}(c)$ be, respectively, the SAV score and the NSAV score of~$c$ in~$(C, V)$.

\begin{lemma}[\cite{DBLP:conf/atal/000120}]
\label{lem-relation-sav-nsav}
Let $(C, V)$ be an election where $m=\abs{C}\geq 2$ and $n=\abs{V}$. Let~$B$ be a set of at least $n\cdot m^2$ candidates disjoint from~$C$. Then, for every two candidates~$c$ and~$c'$ in~$C$, it holds that ${\memph{\textsf{SAV}}}_{(C, V)}(c)>{\memph{\textsf{SAV}}}_{(C, V)}(c')$ if and only if ${\memph{\textsf{NSAV}}}_{(C\cup B, V)}(c)>{\memph{\textsf{NSAV}}}_{(C\cup B, V)}(c')$.
\end{lemma}

\section{Control}
\label{sec-control}
In this section, we study the complexity of election control problems. We first study control by modifying the set of voters and then study control by modifying the set of candidates.

\subsection{Control of Voters}
\onlyfull{This section is devoted to the control problems.}
When considering AV as a single-winner voting rule (i.e., when $k=1$), it is known that {\probb{CCAV}{AV}} and {\probb{CCDV}{AV}} are {\nph}~\cite{DBLP:journals/ai/HemaspaandraHR07}. Notably, when~$k=1$, PAV, ABCCV, and AV are identical. As a consequence, {\prob{CCAV}} and {\prob{CCDV}} for PAV and ABCCV are also {\nph} even when~$k=1$.
For {\prob{CCAV}} and {\prob{CCDV}}, it remains to consider SAV, NSAV, and MAV\@. Note that these three rules are not equivalent to AV when $k=1$~\cite{DBLP:series/sbis/LacknerS23}. We first show the {\nphns} for SAV and NSAV even when restricted the above-mentioned special case.
Our reduction is from {\prob{RX3C}}. We provide the detailed reductions for SAV, and utilize Lemma~\ref{lem-relation-sav-nsav} to show how to adapt the reductions to make them applicable to NSAV.

\begin{theorem}
\label{thm-ccav-sav-np-hard}
For $\varphi\in \{\memph{\text{SAV}}, \memph{\text{NSAV}}\}$, {\probb{CCAV}{$\varphi$}} and {\probb{CCDV}{$\varphi$}} are {\nph} even when~$k=1$. Moreover, the hardness for {\prob{CCDV}} holds even if every vote approves at most three candidates and every candidate is approved by at most three voters. 
\end{theorem}

\begin{proof}
We prove the theorem separately for each problem. All reductions are from {\prob{RX3C}}. Let~$(\xs,\xc)$ be an instance of {\prob{RX3C}} such that~$\abs{\xs}=\abs{\xc}=3\kappa$.

\begin{itemize}
\item {\probb{CCAV}{SAV}}
\end{itemize}

We assume that~$\kappa>2$ and~$\kappa$ is divisible by~$4$. This assumption does not change the complexity of the problem{\footnote{If~$\kappa$ is not divisible by~$4$, we can modify the instance into an equivalent instance where this condition is satisfied. In particular, if $\kappa\equiv 3 \pmod 4$, we create three further elements in~$\xs$, add three $3$-subsets each of which consists of exactly these three newly introduced elements in~$\xc$, and increase~$\kappa$ by one. If $\kappa\equiv 2\pmod 4$, we create six further elements $a_1,\dots,a_6$ in~$\xs$, add three copies of $\{a_1,a_2,a_3\}$ and three copies of $\{a_4,a_5,a_6\}$ into~$\xc$, and increase~$\kappa$ by two. If $\kappa\equiv 1\pmod 4$, we create~$9$ further elements $a_1,\dots,a_9$ in~$\xs$, add three copies of each of $\{a_1,a_2,a_3\}$, $\{a_4, a_5, a_6\}$, and $\{a_7,a_8, a_9\}$ into~$\xc$, and increase~$\kappa$ by three. After the operations, the instance is still an instance of {\prob{RX3C}}, and it is easy to see that the original instance is a {\yesins} if and only if the new instance is a {\yesins}.}}. We create an instance $((C, V), k, U, J, \ell)$ of {\probb{CCAV}{SAV}} as follows. For each~$\xse\in \xs$, we create one candidate denoted still by~$\xse$ for simplicity. In addition, we create one candidate~$p$ which is the only distinguished candidate. Let~$C=\xs\cup \{p\}$ and let $J=\{p\}$. We create $\frac{3}{4} \kappa \cdot(\kappa-2)$ registered votes each of which approves all candidates except the distinguished candidate~$p$. Let~$V$ be the multiset of these registered votes. Note that as $\kappa\equiv 0 \pmod{4}$ and $\kappa>2$, $\frac{3}{4} \kappa \cdot(\kappa-2)$ is a positive integer, and hence~$V$ is well-defined. Then, for each $\xce=\{\xse_x,\xse_y,\xse_z\}\in \xc$, we create in~$U$ an unregistered vote~$v(\xce)$ which approves the four candidates~$p$,~$\xse_x$,~$\xse_y$, and~$\xse_z$. We complete the construction by setting~$k=1$ and~$\ell=\kappa$. The construction clearly can be done in polynomial time. In the following we show that the given {\prob{RX3C}} instance is a {\yesins} if and only if the above constructed instance of {\probb{CCAV}{SAV}} is a {\yesins}.

$(\Rightarrow)$ Assume that~$\xc$ contains an exact~$3$-set cover $\xc'\subseteq \xc$ of~$\xs$. Let us consider the election after adding into~$V$ all unregistered votes corresponding to the~$\kappa$ elements in~$\xc'$, i.e., the election $E=(C, V\cup U')$ where $U'=\{v(\xce)\in U \setmid \xce\in \xc'\}$. As each unregistered vote approves exactly four candidates including~$p$, the SAV score of~$p$ in~$E$ is~$\frac{\kappa}{4}$. In addition, as~$\xc'$ is an exact~$3$-set cover of~$\xs$, due to the above construction, for each candidate~$\xse\in \xs$,~$U'$ contains exactly one unregistered vote approving~$\xse$, implying that the SAV score of~$\xse$ increases by~$\frac{1}{4}$ after the addition of votes in~$U'$ to~$V$. Given that the SAV score of each candidate~$\xse\in \xs$ in $(C, V)$ is~$\frac{3}{4}\kappa \cdot (\kappa-2)\cdot \frac{1}{3\kappa}=\frac{\kappa-2}{4}$, we know that the SAV score of~$a$ in~$E$ is~$\frac{\kappa-1}{4}$. Therefore, the distinguished candidate~$p$ has the unique highest score and hence~$p$ uniquely wins~$E$.

$(\Leftarrow)$ Assume that there is a multiset~$U'\subseteq U$ of cardinality at most $\ell=\kappa$ so that~$p$ becomes the unique SAV winner of $E=(C, V\cup U')$. Let $\xc'=\{\xce\in \xc \setmid v(\xce)\in U'\}$. Since the distinguished candidate~$p$ has the unique least SAV score~$0$ in $(C, V)$, it holds that $\abs{U'}\geq 1$. Then, it is easy to see that~$\abs{U'}=\kappa$, since otherwise at least one candidate in~$\xs$ has SAV score at least $\frac{\kappa-2}{4}+\frac{1}{4}=\frac{\kappa-1}{4}$, and the distinguished candidate has SAV score at most~$\frac{\kappa-1}{4}$ in~$E$, contradicting that~$p$ is the unique winner of~$E$. It follows that the SAV score of~$p$ in~$E$ is~$\frac{\kappa}{4}$, and the SAV score of every~$\xse\in \xs$ in~$E$ is strictly smaller than~$\frac{\kappa}{4}$. Due to the construction of the votes, this means that for every~$\xse\in \xs$ there is at most one vote $v(\xce)\in U'$ which approves~$\xse$, and by the definition of~$v(\xse)$ and~$\xc'$ it holds that~$\xse\in \xce\in \xc'$. Given $\abs{\xc'}=\abs{U'}=\kappa$, it holds that~$\xc'$ is an exact~$3$-set cover of~$\xs$.

\begin{itemize}
\item {\probb{CCDV}{SAV}}
\end{itemize}

We create an instance $(C, V, k, J, \ell)$ of {\probb{CCDV}{SAV}} as follows. For each~$\xse\in \xs$, we create one candidate denoted by the same symbol for simplicity. In addition, we create four candidates~$p$,~$d_1$,~$d_2$, and~$d_3$, where~$p$ is the distinguished candidate. Let $C=\xs\cup \{p, d_1, d_2, d_3\}$ and let~$J=\{p\}$. We create the following votes in~$V$. First, we create two votes~$v_1=\{p,d_1\}$ and $v_2=\{p, d_2, d_3\}$. Then, for each $\xce=\{\xse_x,\xse_y,\xse_z\}\in \xc$, we create one vote $v(\xce)=\{\xse_x, \xse_y, \xse_z\}$. We complete the construction by setting~$k=1$ and~$\ell=\kappa$. Clearly, the construction can be done in polynomial time. It remains to show the correctness of the reduction. The SAV scores of all candidates are summarized in Table~\ref{tab-ccdv-sav-scores}.
\begin{table}
\caption{SAV scores of candidates in the instance of {\probb{CCDV}{SAV}} in the proof of Theorem~\ref{thm-ccav-sav-np-hard}.}
\begin{center}
\begin{tabular}{llllll}\toprule
&$p$ & $\xse\in \xs$ & $d_1$ & $d_2$ & $d_3$ \\ \midrule

SAV scores &$\frac{1}{2}+\frac{1}{3}=\frac{5}{6}$ & $1$ & $\frac{1}{2}$ & $\frac{1}{3}$ & $\frac{1}{3}$ \\ \bottomrule
\end{tabular}
\end{center}
\label{tab-ccdv-sav-scores}
\end{table}

$(\Rightarrow)$ Let~$\xc'\subseteq \xc$ be an exact~$3$-set cover of~$\xs$. Let $V'=\{v(\xce) \setmid \xce\in \xc'\}$ be the votes corresponding to~$\xc'$. We claim that after removing all the~$\kappa$ votes in~$V'$ from~$V$, the distinguished candidate~$p$ becomes the unique winner. Let $E=(C, V\setminus V')$. Clearly, the SAV scores of~$p$,~$d_1$,~$d_2$, and~$d_3$ in~$E$ remain the same as summarized in Table~\ref{tab-ccdv-sav-scores}. As~$\xc'$ is an exact~$3$-set cover of~$\xs$, due to the construction of the votes, for each candidate~$\xse\in \xs$, $V'$ contains exactly one vote approving~$\xse$. Therefore, the SAV score of~$\xse$ in~$E$ decreases to  $1-\frac{1}{3}=\frac{2}{3}$, resulting in~$p$ being the unique SAV winner of~$E$.

$(\Leftarrow)$ Assume that there is $V'\subseteq V$ of at most~$\ell=\kappa$ votes whose removal results in~$p$ being the unique SAV winner. We may assume that $\{v_1,v_2\}\cap V'=\emptyset$, since it is easy to see that if~$p$ uniquely wins $(C, V\setminus V')$, then~$p$ also uniquely wins $(C, V\setminus (V'\setminus \{v_1, v_2\}))$. Under this assumption, the SAV score of~$p$ in $(C,V\setminus V')$ remains $\frac{5}{6}$. To decrease the score of a candidate~$\xse\in \xs$, at lease one vote approving~$\xse$ must be in~$V'$. Given $\ell=\kappa$, we know that the submultiset $\{\xce\in \xc \setmid v(\xce)\in V'\}$ corresponding to~$V'$ is an exact~$3$-set cover of~$\xs$.

\begin{itemize}
\item {\probb{CCAV}{NSAV}} and {\probb{CCDV}{NSAV}}
\end{itemize}

Our reduction for {\probb{CCAV}{NSAV}} (respectively, {\probb{CCDV}{NSAV}}) is obtained from the above reduction for {\probb{CCAV}{SAV}} (respectively, \probb{CCDV}{SAV}) by adding additional $n\cdot m^2$ candidates not approved by any vote. Here,~$m$ and~$n$ are respectively the number of candidates and the number of votes created in the instance of {\probb{CCAV}{SAV}} (respectively, \probb{CCDV}{SAV}). Then, by Lemma~\ref{lem-relation-sav-nsav}, we know that for any $V'\subseteq V\cup U$ (respectively,~$V'\subseteq V$), a candidate has the unique highest SAV score in $(C, V')$ if and only if it has the unique NSAV score in $(C, V')$. This implies that the constructed instance of {\probb{CCAV}{SAV}} (respectively, \probb{CCDV}{SAV}) is a {\yesins} if and only if the instance of {\probb{CCAV}{NSAV}} (respectively, \probb{CCDV}{NSAV}) is a {\yesins}.
\end{proof}

Now we consider {\prob{CCAV}} and {\prob{CCDV}} for MAV.
Unlike the above results, we show that, somewhat interestingly, {\probb{CCAV}{MAV}} and {\probb{CCDV}{MAV}} have different complexity. Concretely, {\probb{CCAV}{MAV}} is {\nph} when $k=1$, while {\probb{CCDV}{MAV}} turns out to be polynomial-time solvable as long as~$k$ is a constant.
The following lemma, which characterizes the space of MAV winning $1$-committees, is useful in establishing the {\nphns} of {\probb{CCAV}{MAV}}.

\begin{lemma}
\label{lem-a}
Let~$(C, V)$ be an election, and let~$A\subseteq V$ be the submultiset of votes in~$V$ each approving the maximum number of candidates, i.e., $A=\argmax_{v\in V}\{\abs{v}\}$. Moreover, let~$C'$ be the set of candidates approved by all votes from~$A$, i.e., $C'=\bigcap_{v\in A} v$. Then, if $C'\neq \emptyset$, all candidates in~$C'$ are tied as MAV single winners of $(C, V)$. Otherwise, all candidates in~$C$ are tied as MAV single winners of $(C, V)$.
\end{lemma}

\begin{proof}
Let~$(C, V)$, $A$, and~$C'$ be as stipulated in Lemma~\ref{lem-a}. Let~$x$ be the number of candidates approved by each vote in~$A$, i.e., for every $v\in A$ it holds that $x=|v|$. Lemma~\ref{lem-a} clearly holds if $C'=C$. So, in the following let us assume that $C\setminus C'\neq\emptyset$.

If $C'\neq \emptyset$, then for any singleton committee~$\{a\}\subseteq C'$, the Hamming distance between~$\{a\}$ and every vote in~$A$ is $x-1$, and that between~$\{a\}$ and every vote not in~$A$ is at most~$x$. Hence,~$\{a\}$ has MAV score at most~$x$. Now we analyze the MAV score of a singleton committee~$\{b\}$ where $b\in C\setminus C'$. Due to the definition of~$C'$, there exists at least one vote $v\in A$ such that $b\not\in v$. The Hamming distance between~$\{b\}$ and~$v$ is $x+1$, implying that $\{b\}$ has MAV score at least $x+1$. Therefore, every candidate~$a\in C'$ is a MAV single winner.

If $C'=\emptyset$, then for any singleton committee~$\{a\}\subseteq C$, there exists at least one vote $v\in A$ such that $a\not\in v$. The Hamming distance between~$\{a\}$ and~$v$ is $x+1$. Clearly,  the Hamming distance between~$\{a\}$ and any vote not in~$A$ is at most~$x$. Hence, we know that all singleton committees of~$C$ have the same MAV score $x+1$, implying that all candidates are tied as MAV single winners of~$(C, V)$.
\end{proof}

We mention in passing that Lackner and Skowron~\cite{DBLP:series/sbis/LacknerS23} call rules which are not identical to AV when $k=1$ {\textit{non-standard}} rules, and gave two small examples to show that both MAV and SAV are nonstandard. Lemma~\ref{lem-a} fully characterizes the space of MAV winning $1$-committees, and from the characterization it is easy to see that MAV does not necessarily select candidates receiving the most approvals when $k=1$.

Now we are ready to give the {\nphns} of {\probb{CCAV}{MAV}}.

\begin{theorem}
\label{thm-ccav-mav-nph-k-1}
{\probb{CCAV}{MAV}} is {\nph}, even when $k=1$, there is only one registered vote, and all votes approve the same number of candidates.
\end{theorem}

\begin{proof}
We prove the theorem via a reduction from {\prob{RX3C}}. Let $(\xs,\xc)$ be an {\prob{RX3C}} instance where $\abs{\xs}=\abs{\xc}=3\xsize>0$. We create a {\probb{CCAV}{MAV}} instance $((C, V), k, U, J, \ell)$ as follows. For each $\xse\in \xs$, we create one candidate denoted still by~$\xse$ for simplicity. In addition, for each $\xce\in \xc$, we create three candidates denoted by~$c(\xce_1)$,~$c(\xce_2)$, and~$c(\xce_3)$. Moreover, we create one candidate~$p$ which is the only distinguished candidate. Let~$C$ be the set of all these $12\kappa+1$ created candidates, and let $J=\{p\}$. Concerning the votes, we create only one registered vote in~$V$ which approves all candidates in~$\xs$ and the distinguished candidate~$p$, and disapproves all the other candidates. Unregistered votes are created according to $\xc$. In particular, for each $\xce\in \xc$, we create one unregistered vote~$v(\xce)$ which approves exactly the four candidates~$p$,~$c(\xce_1)$,~$c(\xce_2)$,~$c(\xce_3)$, and every candidate $\xse\in\xs$ such that $\xse\not\in \xce$. Let~$U$ be the set of all the created~$3\kappa$ unregistered votes. Note that all created votes approve exactly $3\xsize+1$ candidates. We complete the reduction by setting $\ell=\xsize$. The construction can be done in polynomial time. It remains to prove the correctness of the reduction.

$(\Rightarrow)$ Suppose that there is an exact $3$-set cover $\xc'\subsetneq \xc$ of~$\xs$. Let $U'=\{v(\xce) \setmid \xce\in \xc'\}$ be the set of the~$\ell$ unregistered votes corresponding to~$\xc'$. Let $E=(C, V\cup U')$. Due to the definition of~$\xc'$ and the construction of the election,~$p$ is the unique candidate that is approved by all votes from $V\cup U'$. By Lemma~\ref{lem-a},~$p$ is the unique MAV winner of~$E$.

$(\Leftarrow)$ Suppose that there exists $U'\subseteq U$ of cardinality at most~$\ell$ so that~$p$ is the unique MAV winner of $E=(C, V\cup U')$. Due to Lemma~\ref{lem-a}, in~$E$, for every candidate $\xse\in \xs$, there must be at least one vote from~$U'$ not approving~$\xse$. Due to the construction of the unregistered votes, this means that~$U'$ contains at least one vote~$v(\xce)$ such that $\xse\in \xce\in \xc$. It follows that $\xc'=\{\xce\in \xc \setmid v(\xce)\in U'\}$ is a set cover of~$\xs$. Moreover, as every $\xce\in \xc$ is of cardinality three, it holds that $\abs{U'}=\kappa$ and~$\xc'$ is an exact $3$-set cover of~$\xs$.
\end{proof}

In contrast to the {\nphns} of {\probb{CCAV}{MAV}} even when restricted to the special case as stated in Theorem~\ref{thm-ccav-mav-nph-k-1}, {\probb{CCDV}{MAV}} is polynomial-time solvable as long as~$k$ is a constant. As far as we know, MAV is the first natural voting rule for which the complexity of {\prob{CCAV}} and {\prob{CCDV}} differs.

\begin{theorem}
\label{thm-ccav-ccdv-mav-polynomial-time-solvable-k-constant}
{\probb{CCDV}{MAV}} is polynomial-time solvable when~$k$ is a constant.
\end{theorem}

\begin{proof}
Let $I=(C, V, k, J, \ell)$ be a {\probb{CCDV}{MAV}} instance. Let $m=|C|$ denote the number of candidates. We derive an algorithm as follows.
First, we compute~$\mathcal{C}_{k, C}(J)$ and~$\overline{\mathcal{C}_{k,C}}({J})$, i.e., the collection of $k$-committees containing~$J$, and the collection of $k$-committees not containing~$J$, respectively.
As the number of $k$-committees is~$m\choose k$ and~$k$ is a constant, they can be computed in polynomial time. Then, we split the given instance~$I$ into polynomially many subinstances each of which takes as input~$I$, a nonnegative integer~$x\leq m$, and a $k$-committee $\w\in \mathcal{C}_{k, C}(J)$, and determines whether we can delete at most~$\ell$ votes from~$V$ in $(C,V)$ so that in the remaining election
\begin{enumerate}
\item[(1)] $\w$~has MAV at most~$x$, and
\item[(2)] all $k$-committees from $\overline{\mathcal{C}_{k, C}}(J)$ have MAV scores at least~$x+1$.
\end{enumerate}
The above two conditions ensure that in the remaining election winning $k$-committees must be from~$\mathcal{C}_{k, C}(J)$.
Obviously,~$I$ is a {\yesins} if and only if at least one of the subinstances is a {\yesins}.
We focus on solving a subinstance $(I, x, w)$. Our algorithm proceeds as follows.
First, by Condition~(1), all votes in~$V$ which are of Hamming distance at least~$x+1$ from~$\w$ need to be deleted; we do so and decrease~$\ell$ by the number of votes deleted. If $\ell<0$ after doing so, we immediately conclude that the subinstance is a {\noins}.
Otherwise, if there exists $\w'\in \overline{\mathcal{C}_{k, C}}(J)$ whose MAV score in the remaining election is at most~$x$, we conclude that the subinstance is a {\noins} too. (Note that in this case, the original instance~$I$ might be a {\yesins}. Nevertheless, a feasible solution of~$I$ will be captured by another subinstance associated with the same~$\w$ but with a smaller~$x$.) Otherwise, the above two conditions are satisfied, and we conclude that the subinstance is a {\yesins}.
\end{proof}

The algorithm in the proof of Theorem~\ref{thm-ccav-ccdv-mav-polynomial-time-solvable-k-constant} runs in time~$\bigos{m^k}$. In the language of parameterized complexity theory this is an {\xp}-algorithm with respect to~$k$. It is interesting to study if this algorithm can be improved to an {\fpt}-algorithm for the same parameter. We leave it as an open question.

\subsection{Control of Candidates}
Now we consider control by modifying the candidate set. Notice that for AV it is impossible to change the scores of registered candidates by adding unregistered candidates, as already observed  in the context of single-winner voting~\cite{baumeisterapproval09,DBLP:journals/jcss/HemaspaandraH07}. This implies that AV is immune to CCAC. However, this is not the case for SAV and NSAV, since in these two cases adding candidates may increase the number of approved candidates of some votes and hence affect the scores of these candidates. We show that {\probb{CCAC}{SAV}} and {\probb{CCAC}{NSAV}} are {\nph} even when treating them  as single-winner voting rules.

\begin{theorem}
\label{thm-ccac-sav-np-hard}
{\probb{CCAC}{SAV}} and {\probb{CCAC}{NSAV}} are {\nph}. Moreover, this holds even when~$k=1$, every vote approves at most four candidates, and every candidate is approved by at most three votes.
\end{theorem}

\begin{proof}
We prove the theorem by reductions from {\prob{RX3C}}. Let~$(\xs,\xc)$ be an {\prob{RX3C}} instance where $\abs{\xs}=\abs{\xc}=3\kappa>0$. We first provide the reduction for SAV, and then show how to utilize Lemma~\ref{lem-relation-sav-nsav} to adapt the reduction for NSAV.

\begin{itemize}
\item {\probb{CCAC}{SAV}}
\end{itemize}

We create an instance $((C\cup D, V), k, J, \ell)$ of {\probb{CCAC}{SAV}} as follows. For each~$\xce\in \xc$, we create one candidate~$c(\xce)$. For each $\xse\in \xs$, we create one candidate~$c(\xse)$. In addition, we create a distinguished candidate~$p$ and three dummy candidates~$d_1$,~$d_2$, and~$d_3$. Let
\[C=\{c(\xse) \setmid \xse\in \xs\}\cup \{p\}\cup \{d_1, d_2, d_3\},\] $D=\{c(\xce) \setmid \xce\in \xc\}$, and~$J=\{p\}$.
We create the following votes. First, we create three votes~$v_1=\{p\}$, $v_2=\{p, d_1\}$, and $v_3=\{p, d_2, d_3\}$. Then, for each~$\xse\in \xs$, we create two votes~$v(\xse)$ and~$v'(\xse)$. In particular,~$v'(\xse)$ approves exactly~$c(\xse)$, and~$v(\xse)$ approves exactly~$c(\xse)$ and every~$c(\xce)$ such that $\xse\in \xce\in \xc$. Hence, the vote~$v(\xse)$ approves exactly four candidates, one from~$C$ and three from~$D$. We complete the reduction by setting~$k=1$ and~$\ell=\kappa$. The above instance clearly can be constructed in polynomial time. We show the correctness of the reduction as follows. The SAV scores of the candidates in~$(C, V)$ are summarized in Table~\ref{tab-ccac-sav-nph}.
\begin{table}
\caption{The SAV scores of candidates in the election restricted to registered candidates constructed in the proof of Theorem~\ref{thm-ccac-sav-np-hard}.}
\begin{center}
\begin{tabular}{ll}\toprule
candidates & SAV scores \\ \midrule
$p$ & $1+1/2+1/3=11/6$\\

$c(\xse)$ & $1+1=2$ \\

$d_1$ & $1/2$\\

$d_2$, $d_3$ & $1/3$ \\ \bottomrule
\end{tabular}
\end{center}
\label{tab-ccac-sav-nph}
\end{table}

$(\Rightarrow)$ Let $\xc'\subseteq \xc$ be an exact~$3$-set cover of~$\xs$, and let $D'=\{c(\xce) \setmid \xce\in \xc'\}\subseteq D$ be the set of the~$\kappa$ candidates corresponding to~$\xc'$. Consider the election $E=(C\cup D', V)$. Clearly, the SAV scores of~$p$,~$d_1$,~$d_2$, and~$d_3$ in~$E$ remain the same as in~$(C, V)$ (see Table~\ref{tab-ccac-sav-nph}). Now we analyze the SAV scores of all~$c(\xse)$ where~$\xse\in \xs$. As~$\xc'$ is an exact $3$-set cover of~$\xs$, in~$E$ each vote~$v(\xse)$ approves exactly two candidates---~$c(\xse)$ and some~$c(\xce)$ such that~$\xse\in \xce\in \xc'$. Therefore, after adding the candidates from~$D'$ into~$C$, the SAV score of~$c(\xse)$ decreases to~$2-1/2=3/2$. Each candidate in~$D'$ has SAV score~$3/2$ in~$E$ too. Therefore,~$p$ is the unique SAV winner of~$E$.

$(\Leftarrow)$ Assume that~$p$ becomes the unique SAV winner after adding a set~$D'\subseteq D$ of at most~$\ell=\kappa$ candidates into~$C$. Let $E=(C\cup D', V)$. Similar to the above analysis, we know that the SAV score of~$p$ in~$E$ remains $11/6$. As~$p$ uniquely wins~$E$ under SAV, we know that the SAV score of every candidate~$c(\xse)$ in~$E$, where~$\xse\in \xs$, must be decreased compared to that in $(C, V)$. Due to the construction of the election, this means that for every candidate~$c(a)$ where~$a\in \xs$ there exists at least one candidate~$c(\xce)\in D'$ where~$\xce\in \xc$ such that~$\xse\in \xce$. After adding such a candidate~$c(\xce)$ from~$D$ into~$C$, the SAV score of~$c(\xse)$ with respect to the vote~$v(\xse)$ decreases from~$1$ to~$1/2$, leading to a final SAV score~$3/2$, smaller than the score of~$p$. As $\ell=\kappa$, it follows that $\{\xce\in \xc \setmid c(\xce)\in D'\}$ is an exact~$3$-set cover of~$\xs$.

\begin{itemize}
\item {\probb{CCAC}{NSAV}}
\end{itemize}

Our reduction for {\probb{CCAC}{NSAV}} is obtained from the above reduction for {\probb{CCAC}{SAV}} by adding~$n\cdot m^2$ new registered candidates in~$C$ not approved by any votes.  Here,~$m$ and~$n$ are respectively the number of all candidates (registered and unregistered) and the number of votes created in the instance of {\probb{CCAC}{SAV}}. The newly created registered candidates have enough small SAV scores so that no matter which unregistered candidates from~$D$ are added into~$C$, none of them is winning. Then, by Lemma~\ref{lem-relation-sav-nsav}, we know that for any $D'\subseteq D$, a candidate has the unique highest SAV score in $(C\cup D', V)$ if and only if it has the unique NSAV score in $(C\cup D', V)$. This implies that the constructed instance of {\probb{CCAC}{SAV}} is a {\yesins} if and only if the instance of {\probb{CCAC}{NSAV}} is a {\yesins}.
\end{proof}

Now we move on to the three nonadditive rules PAV, ABCCV, and MAV. The immunity of single-winner AV to CCAC also implies that PAV and ABCCV are immune to CCAC when~$k=1$. More generally, one can observe that PAV and ABCCV are immune to CCAC when the number of distinguished candidates equals~$k$\onlyfull{, the size of the desired winning $k$-committee}. In this case, the question of {\prob{CCAC}} is degenerated to whether we can add at most~$\ell$ unregistered candidates so that a given~$k$-committee~$J$ is the unique $k$-winning committee. To see that ABCCV and PAV are immune to CCAC in this special case, observe that if the given $k$-committee~$J$  is not a winning $k$-committee of $(C, V)$, there exists a committee~$\w\subseteq C$ other than~$J$ which has at least the same score as that of~$J$. As the scores of committees in $(C, V)$ do not change by adding further candidates into~$C$, the committee~$\w$ prevents~$J$ from being the unique winning $k$-committee no matter which candidates are added. This reasoning in fact applies to all multiwinner voting rules that satisfy an axiomatic property defined below, which in general states that if a committee is not uniquely winning it cannot be uniquely winning when additional candidates are introduced.

\begin{definition}[Negated Revealed Preference (NRP)]
A multiwinner voting rule~$\varphi$ satisfies NRP if for every election $(C, V)$ and every $k$-committee $\w\subseteq C$, if~$\w$ is not the unique $k$-winning committee of~$\varphi$ at $(B, V)$ for some $B\subseteq C$ such that $\w\subseteq B$, then for any $B'\subseteq C$ such that $B\subseteq B'$ the committee~$\w$ is not the unique winning $k$-committee of~$\varphi$ at $(B', V)$.
\end{definition}

It should be noted that the above definition is a variant of the notion of the unique version of weak axiom of revealed preference (unique-WARP) that has been extensively studied for single-winner voting rules. Concretely, a single-winner voting rule satisfies unique-WARP if the unique winner remains as unique winner when restricted to any subset of candidates containing this winner. Compared to unique-WARP, NRP specifies the nonwinning status of some committee without mentioning the identities of winning committees.\footnote{There are several extensions of WARP to choice correspondences in the literature (see, e.g.,~\cite{DBLP:journals/jet/BrandtH11,PetesPtheorydecision2021}), where a choice correspondence is a function that assigns to each subset $S\subseteq C$ a subset $C'\subseteq S$. So, these notions apply to resolute multiwinner voting rules which always select exactly one winning committee.}

\begin{theorem}
\label{thm-pav-abbcv-mav-immue-to-ccac-k-equal-distinguished-candidates}
Every NRP multiwinner voting rule is immune to {\memph{CCAC}} when the number of distinguished candidates is~$k$.
\end{theorem}

It is easy to see that ABCCV and PAV fulfill NRP. However, this is not the case for MAV, as shown in Example~\ref{ex-c}.

\begin{example}[MAV fails NRP.]
\label{ex-c}
Let $C=\{a,b\}$, $D=\{c, d\}$, and $V=\{\{b\}, \{a,c\}, \{a,d\}\}$. Clearly, both~$\{a\}$ and~$\{b\}$ are winning $1$-committees of~MAV at $(C, V)$. However, $\{a\}$ is the unique winning $1$-committee of~MAV at $(C\cup D, V)$.
\end{example}

The above example also shows that MAV is not immune to CCAC even when $k=1$. From the complexity point of view, we have the following result.

\begin{theorem}
\label{thm-ccac-mav-nph-k-1}
{\probb{CCAC}{MAV}} is {\nph} even when $k=1$ and there are only two registered candidates.
\end{theorem}

\begin{proof}
We prove the theorem via a reduction from {\prob{RX3C}}. Let $(\xs, \xc)$ be an instance of {\prob{RX3C}} where $\abs{\xs}=\abs{\xc}=3\kappa$. Without loss of generality, we assume that $\kappa>1$. We create an instance $((C\cup D, V), k, J, \ell)$ of {\probb{CCAC}{MAV}} as follows. First, we create two candidates denoted by~$p$ and~$q$. Then, for every $\xce\in \xc$, we create one candidate~$c(\xce)$. Let $C=\{p, q\}$ be the set of registered candidates, let $J=\{p\}$, and let $D=\{c(\xce) \setmid \xce\in \xc\}$ be the set of unregistered candidates. Regarding the votes, we first create one vote which approves only~$q$. Then, for every $\xse\in \xs$, we create one vote~$v(\xse)$ which approves all candidates except~$q$ and the three candidates corresponding to~$\xce\in \xc$ containing~$\xse$, i.e., $v(\xse)=\{p\}\cup \{c(\xce) \setmid \xse\not\in \xce, \xce\in \xc\}$. So,~$v(\xse)$ approves exactly $3\kappa-2$ candidates in $C\cup D$. Let~$V$ denote the multiset of all $3\kappa+1$ votes created above. Finally, we set $k=1$ and $\ell=\kappa$. The instance of {\probb{CCAC}{MAV}} clearly can be constructed in polynomial time. In the following, we show the correctness.

$(\Rightarrow)$ Suppose that~$\xc$ contains an exact $3$-set cover~$\xc'$ of~$\xs$. Let $D'=\{c(\xce) \setmid \xce\in \xc'\}$ be the unregistered candidates corresponding to~$\xc'$. Let $E=(C\cup D', V)$. We show below that the {\probb{CCAC}{MAV}} instance constructed above is a {\yesins} by showing that~$p$ is the unique MAV winning $1$-committee of~$E$. As~$\xc'$ is an exact set cover of~$\xs$, for every vote~$v(\xse)$ where $\xse\in \xs$ there is exactly one~$c(\xce)\in D'$ such that~$v(\xse)$ does not approve~$c(\xce)$. Therefore, every vote~$v(\xse)$ where $\xse\in \xs$ approves exactly~$\kappa$ candidates from $C\cup D'$ and, moreover,~$p$ is the only candidate that is approved by all votes corresponding to~$\xs$. Then, by Lemma~\ref{lem-a} and the assumption that $\kappa>1$, we know that~$\{p\}$ is the unique MAV winning $1$-committee of~$E$.

$(\Leftarrow)$ Suppose that there is a $D'\subseteq D$ such that $\abs{D'}\leq \ell=\kappa$ and~$\{p\}$ becomes the unique MAV winning $1$-committee of $E=(C\cup D', V)$. We prove below that $\xc'=\{\xce\in \xc \setmid c(\xce)\in D'\}$ corresponding to~$D'$ is an exact $3$-set cover of~$\xs$. For the sake of contradiction, assume that this is not the case. Let $\xse\in \xs$ be any arbitrary element in~$\xs$ that is not covered by~$\xc'$, i.e., $\xse\not\in \xce$ holds for all $\xce\in \xc'$. Then,~$v(\xse)$ approves all candidates in $C\cup D'$ except only~$q$. As~$q$ is not approved by any vote corresponding to~$\xs$, this implies that~$v(\xse)$ approves the largest number of candidates from $C\cup D'$, and any other vote approving the maximum number of candidates from $C\cup D'$ must approve exactly the same candidates among $C\cup D'$ as~$v(a)$. Then, by Lemma~\ref{lem-a}, every candidate in $C\cup D'$ except~$q$ is an MAV single winner of~$E$. This contradicts that~$\{p\}$ is the unique MAV winning $1$-committee of~$E$.
\end{proof}

One may wonder whether the restriction that the number distinguished candidates equals~$k$ is necessary for ABCCV and PAV to be immune to CCAC. Example~\ref{ex-d} below answers the question in the affirmative by illustrating that for every $k\geq 2$, ABCCV and PAV are susceptible to CCAC when there are at most~$k-1$ distinguished candidates.

\begin{example}[ABCCV and PAV are susceptible to CCAV for $k\geq 2$.]
\label{ex-d}
Let $C=\{a,b,c\}$, $D=\{d\}$, and $J=\{a\}$.
For ABCCV, we have five votes $v_1=\{a\}$, $v_2=v_3=\{b,d\}$, and $v_4=v_5=\{c,d\}$.
%
%
For PAV, we have eight votes $v_1=v_2=\{a\}$, $v_3=v_4=v_5=\{b,d\}$, and $v_6=v_7=v_8=\{c,d\}$.
It is easy to verify that with respect to~$C$ the only ABCCV/PAV winning $2$-committee is $\{b,c\}$. However, if we add the candidate~$d$, $\{a, d\}$ becomes the unique ABCCV/PAV winning $2$-committee.
We can show the susceptibility of ABCCV and PAV to CCAC for every $k\geq 3$ by slight modifying the above elections. We first we create $k-2$ copies of~$a$ in~$C$, and let~$J$ be the set consisting of~$a$ and all the copies of~$a$. Then, in addition to the above votes, for ABCCV, we create $k-2$ new votes each of which approves exactly one copy of~$a$, and each copy of~$a$ is approved by one of these $k-2$ votes, and for PAV, we create new $2(k-2)$ votes so that each of them approves one copy of~$a$, and each copy of~$a$ is approved by two of these votes.
\end{example}

Concerning the complexity, {\probb{CCAC}{ABCCV}} and {\probb{CCAC}{PAV}} are {\conph}. In fact, we can show the {\conphns} even for a class of rules and for the special case where there is only one distinguished candidate and we do not allow to add any unregistered candidate, i.e., $\abs{J}=1$ and $\ell=0$. In this case, the question becomes whether a distinguished candidate~$p$ is included in all winning $k$-committees, which is exactly the {\probb{$p$-CC}{$\varphi$}} problem.

\begin{theorem}
\label{thm-cc-abccv-pav-co-np}
For every $\omega$-Thiele rule~$\varphi$ such that $\omega(2)<2\omega(1)$, {\probb{$p$-CC}{$\varphi$}} is {\conph}, even when every vote approves at most two candidates, and every candidate is approved by  three votes.
\end{theorem}

\begin{proof}
We prove the theorem by a reduction from {\prob{Independent Set}} on regular graphs to {\probb{$p$-CC}{$\varphi$}}, where~$\varphi$ is an $\omega$-Thiele rule such that $\omega(2)<2\omega(1)$.

Let $(G, \kappa)$ be an {\prob{Independent Set}} instance where~$G=(\vset,\eset)$ is a regular graph. Let~$t$ be the degree of the vertices in~$G$.
For each vertex $u\in \vset$, we create one candidate denoted by the same symbol for simplicity. In addition, we create a candidate~$p$. Let $C=\{p\}\cup \vset$, and let $J=\{p\}$. We create the following votes. First, for each edge $\edge{u}{u'}\in \eset$, we create one vote~$v(\edge{u}{u'})$ which approves exactly~$u$ and~$u'$. Besides, we create~$t$ votes each of which approves exactly the distinguished candidate~$p$. Finally, we set $k=\kappa$. The instance of {\probb{$p$-CC}{$\varphi$}} is $((C, V), J, k)$ which can be constructed in polynomial time.
It remains to prove the correctness of the reduction.

$(\Rightarrow)$ If the instance of {\prob{Independent Set}} is a {\yesins}, then it is easy to verify that every $k$-committee corresponding to an independent set of~$\kappa$ vertices is a~$\varphi$ winning $k$-committee with~$\varphi$ score~$\kappa\cdot t\cdot \omega(1)$, implying that the above constructed instance of {\probb{$p$-CC}{$\varphi$}} is a {\noins}.

$(\Leftarrow)$ If~$G$ does not contain any independent set of~$\kappa$ vertices, we claim that the distinguished candidate~$p$ is included in all winning $k$-committees of $(C, V)$ under~$\varphi$. Assume for the sake of contradiction that there is a~$\varphi$ winning~$k$-committee~$C'$ such that $p\not\in C'$.
Clearly, $C'\subseteq\vset$. As~$C'$ is not an independent set, there exist distinct $u, u'\in C'$ which are both approved in the vote~$v(\edge{u}{u'})$. As~$p$ is approved by~$t$ votes not approving any candidates from~$C'$, if we replace (any) one of~$u$ and~$u'$ with~$p$ in~$C'$, the~$\varphi$ score of~$C'$ increases by at least $t\cdot \omega(1)-((t-1)\cdot \omega(1)+(\omega(2)-\omega(1)))>0$, which contradicts that~$C'$ is a~$\varphi$ winning $k$-committee.

As {\prob{Independent Set}} remains {\nph} when restricted to $3$-regular graphs~\cite{DBLP:journals/jct/Mohar01}, the hardness of {\probb{$p$-CC}{$\varphi$}} remains when restricted to the case where every vote approves at most two candidates, and every candidate is approved by three votes.
\end{proof}

As {\prob{Independent Set}} restricted to regular graphs is {\wah} with respect to~$\kappa$~\cite{DBLP:conf/cats/MathiesonS08}, the proof of Theorem~\ref{thm-cc-abccv-pav-co-np} implies the following corollary.

\begin{corollary}
\label{cor-p-cc-thiele-cowah-k}
For every $\omega$-Thiele rule~$\varphi$ such that $\omega(2)<2\omega(1)$, {\probb{$p$-CC}{$\varphi$}} is {\cowah} with respect to~$\kappa$, even when every vote approves at most two candidates.
\end{corollary}

For MAV, we have the following result.

\begin{theorem}
\label{thm-cc-mav-nph}
{\probb{$p$-CC}{MAV}} is {\nph}. Moreover, this holds even when every vote approves three candidates and every candidate is approved by at most three votes.
\end{theorem}

\begin{proof}
We prove Theorem~\ref{thm-cc-mav-nph} by a reduction from {\prob{RX3C}} to {\probb{$p$-CC}{MAV}}.
Let $(\xs, \xc)$ be an {\prob{RX3C}} instance such that $\abs{\xs}=\abs{\xc}=3\xsize$ for some positive integer~$\xsize$.
For each $\xce\in \xc$, we create one candidate~$c(\xce)$. In addition, we create five candidates~$p$,~$d_1$,~$d_2$,~$d_3$, and~$d_4$.
Let $C=\{p, d_1, d_2, d_3, d_4\}\cup \{c(\xce) \setmid \xce\in \xc\}$, and let $J=\{p\}$. We create the following votes. First, for each $\xse\in \xs$, we create one vote~$v(\xse)$ which approves exactly the three candidates~$c(\xce)$ such that $\xse\in \xce\in \xc$. In addition, we create two votes $v_1=\{p, d_1,d_2\}$ and $v_2=\{p, d_3, d_4\}$. Note that every vote approves exactly three candidates. Finally, we set $k=\xsize+1$. The instance of {\probb{$p$-CC}{MAV}} is $((C, V), J, k)$ which can be constructed in polynomial time.
It remains to prove the correctness of the reduction.

$(\Rightarrow)$ Suppose that there is an exact $3$-set cover $\xc'\subsetneq \xc$ of~$\xs$. Let $\w=\{c(\xce) \setmid \xce\in \xc'\}$ be the subset of candidates corresponding to~$\xc'$. It is not difficult to see that $\w\cup \{p\}$ is a winning $k$-committee with MAV score $\kappa+2$, i.e.,~$\w$ contains at least one approved candidate of every vote. For the sake of contradiction, assume that there is another winning $k$-committee~$\w'$ such that $p\not\in \w'$. First,~$w'$ must contain at least~$\xsize$ candidates corresponding to~$\xc$, since otherwise there must be at least one vote~$v(\xse)$, $\xse\in \xs$,  such that none of its approved candidates~$c(\xce)$ where $\xse\in \xce\in \xc$ is included in~$\w'$, implying that the MAV score of~$\w'$ is at least~$\kappa+4$, a contradiction. Second, if~$\w'$ contains $k=\xsize+1$ candidates corresponding to~$\xc$, then the Hamming distances from~$\w'$ to~$v_1$ and~$v_2$ are both $\kappa+4$, a contradiction too. Therefore,~$\w'$ contains exactly~$\xsize$ candidates corresponding to~$\xc$ and contains exactly one candidate from $\{d_1, d_2, d_3, d_4\}$. However, if $w'\cap \{d_1, d_2\}\neq\emptyset$, the Hamming distance from~$\w'$ to~$v_2$ is  at least~$\kappa+4$, and if $\w'\cap \{d_3, d_4\}\neq \emptyset$, the Hamming distance from~$\w'$ to~$v_1$ is at least $\kappa+4$, both contradicting that~$w'$ is a winning $k$-committee. So, we can conclude that such a winning $k$-committee~$\w'$ does not exist.

$(\Leftarrow)$ Assume that~$\xc$ does not contain any exact $3$-set cover of~$\xs$. Suppose that there is a winning $k$-committee~$\w$ which contains~$p$. Then, there exists at least one vote~$v(\xse)$ where $\xse\in \xs$ such that none of its three approved candidates~$c(\xce)$ where $\xse\in \xce\in \xc$ is in~$\w$. Hence, the MAV score of~$\w$ is exactly~$\kappa+4$. In this case, by replacing~$p$ with some candidate approved by~$v(\xse)$, we obtain another winning $k$-committee, implying that the instance of {\probb{$p$-CC}{MAV}} is a {\noins}.
\end{proof}

We mention in passing that Aziz~et~al.~\cite{DBLP:conf/atal/AzizGGMMW15} studied a problem named {\prob{$R$-TestWS}} which determines if a given $k$-committee is a winning committee of a given election under a multiwinner voting rule. This problem has a flavor of {\prob{$p$}-CC} in the sense that both problems aim to test the winning status of some particular candidates. Aziz~et~al.~\cite{DBLP:conf/atal/AzizGGMMW15} showed {\conphns} of {\prob{$R$-TestWS}} for PAV by a reduction from {\prob{Independent Set}}, but did not study ABCCV and MAV.






We would also like to point out that the {\nphns} and {\conphns} of {\prob{CCAV}} and {\prob{CCDV}} for ABCCV and PAV suggest that when~$k$ is unbounded, {\prob{CCAV}} and {\prob{CCDV}} for ABCCV and for PAV may belong to a much harder class of problems. 
We leave this as an open question for future research.

Let us move on to {\prob{CCDC}}. Unlike the immunity of AV to CCAC, it is easy to see that AV is susceptible to CCDC. Concerning the complexity, it has been shown by Meir~et~al.~\cite{DBLP:journals/jair/MeirPRZ08} that {\probb{CCDC}{AV}} is polynomial-time solvable. 
%
However, for SAV and NSAV, the complexity of {\prob{CCDC}} is the same as {\prob{CCAC}}.

\begin{theorem}
\label{thm-ccdc-sav-np-hard}
{\probb{CCDC}{SAV}} and {\probb{CCDC}{NSAV}} are {\nph} even when~$k=1$. 
\end{theorem}

\begin{proof}
We prove the theorem by reductions from {\prob{RX3C}}. Let~$(\xs, \xc)$ be an {\prob{RX3C}} instance where $\abs{\xs}=\abs{\xc}=3\xsize$. Without loss of generality, we assume $\kappa\geq 3$.
We consider first SAV.

\begin{itemize}
\item {\probb{CCDC}{SAV}}
\end{itemize}
We create an instance $((C, V), k, J, \ell)$ of {\probb{CCDC}{SAV}} as follows.
We create in total~$6\xsize+1$ candidates. In particular, for each~$\xse\in \xs$, we create one candidate~$c(\xse)$. For each~$\xce\in \xc$, we create one candidate~$c(\xce)$. For a given $\xs'\subseteq \xs$ (respectively, $\xc'\subseteq \xc$), let $\candc{\xs'}=\{c(\xse) \setmid \xse\in \xs'\}$ (respectively, $\candc{\xc'}=\{c(\xce) \setmid \xce\in \xc'\}$) be the set of candidates corresponding to~$\xs'$ (respectively, $\xc'$). In addition, we create one candidate~$p$.
Let $C=\candc{\xs} \cup \candc{\xc} \cup \{p\}$ and let~$J=\{p\}$. We create the following votes:
\begin{itemize}
    \item First, for each~$\xce\in \xc$, we create six votes~$v(\xce, 1)$,~$v(\xce, 2)$, $\dots$, $v(\xce, 6)$ each of which approves exactly~$p$ and~$c(\xce)$.
    \item Second, for each $\xse\in \xs$, we create~$12\xsize$ votes $v(\xse,1),\dots,v(\xse,12\xsize)$ each of which approves exactly~$c(\xse)$ and every~$c(\xce)$ such that $\xse\in \xce\in \xc$.
    \item Additionally,  for each $\xse\in \xs$, we create $8\xsize-2$ votes each of which approves exactly~$c(\xse)$.
    \item Finally, we create $6 (3\xsize+1)$ votes each of which approves exactly~$p$ and all the~$3\xsize$ candidates corresponding to~$\xs$. These votes together give to~$p$ and every~$c(\xse)$ where $\xse\in \xs$ six points.
\end{itemize}
Let~$V$ be the multiset of the above created votes. It is clear that $\abs{V}=60\kappa^2+30\kappa+6$. We complete the construction by setting~$k=1$ and~$\ell=\xsize$. Obviously, we can construct the above instance in polynomial time. The SAV scores of the candidates are summarized in Table~\ref{tab-ccdc-sav-nph}.
\begin{table}
\centering{
\caption{A summary of the SAV scores of candidates in the election constructed in the proof of Theorem~\ref{thm-ccdc-sav-np-hard}.}
\label{tab-ccdc-sav-nph}
\begin{tabular}{ll}\\ \toprule
candidates & SAV scores \\ \midrule

$p$ & $18\xsize \cdot \frac{1}{2}+6=9\xsize+6$ \\[2mm]

$c(\xse)$ & $12\xsize\cdot\frac{1}{4}+8\xsize-2+6=11\xsize+4$ \\[2mm]

$c(\xce)$ & $3+12\xsize\cdot \frac{3}{4}=3+9\xsize$ \\ \bottomrule
\end{tabular}
}
\end{table}
We prove the correctness of the reduction as follows.

$(\Rightarrow)$ Assume that there exists an exact $3$-set cover $\xc'\subseteq \xc$ of~$\xs$. We show that after removing the candidates corresponding to~$\xc'$, the distinguished candidate~$p$ becomes the unique winner. Let $E=(C\setminus \candc{\xc'}, V)$. First, after removing a candidate~$c(\xce)$ where $\xce\in \xc$, the SAV score of~$p$ given by~$v(\xce, i)$, $i\in [6]$, increases from~$1/2$ to~$1$. As~$\candc{\xc'}$ contains exactly~$\xsize$ candidates, the removal of these candidates leading to~$p$ having an SAV score $9\xsize+6+6\xsize\cdot \frac{1}{2}=12\xsize+6$ in~$E$. In addition, removing one candidate~$c(\xce)$ where~$\xce\in \xc$ increases the score of each~$c(\xse)$ such that~$\xse\in \xce$ by $12\xsize\cdot (1/3-1/4)=\xsize$, because after removing~$c(\xce)$ the vote~$v(\xse)$ approves three candidates (it approves four candidates in advance). As~$\xc'$ is an exact $3$-set cover, the SAV score of each~$c(\xse)$ where~$\xse\in \xs$ increases to $11\xsize+4+\xsize=12\xsize+4$ in~$E$. Analogously, we can show that the SAV score of every~$c(\xce)$ where $\xce\in \xc\setminus \xc'$ is  $3+9\xsize+\frac{1}{12}\cdot 3\cdot 12\xsize=12\xsize+3$ in~$E$. Therefore,~$p$ becomes the SAV unique-winner of~$E$.

$(\Leftarrow)$ Assume that there exists~$C'\subseteq C\setminus \{p\}$ of at most $\ell=\xsize$ candidates so that~$p$ becomes the unique winner of $E=(C\setminus C', V)$. Observe first that~$C'$ contains at least one candidate from~$\candc{\xc}$. The reason is that if this is not the case, the SAV score of~$p$ in~$E$ can be at most $9\kappa+\frac{6(3\kappa+1)}{2\kappa+1}< 9\kappa+9$, but there exists~$c(a)\in \candc{\xs}\setminus C'$ of SAV score at least $11\kappa+4$ which is larger than $9\kappa+9$ given $\kappa\geq 3$. However, this contradicts that~$p$ uniquely wins~$E$.
Then, to complete the proof, we claim that~$C'$ does not contain any candidate from~$\candc{\xs}$ and, moreover,~$C'$ contains exactly~$\xsize$ candidates.
The claim holds, since otherwise~$C'$ contains at most $\xsize-1$ candidates from~$\candc{\xc}$ which leads to some~$c(\xse)$ where $\xse\in \xs$ having at least the same SAV score as~$p$, and hence contradicts that~$p$ uniquely wins~$E$. To verify this, we distinguish the following cases. Let $C'\cap \candc{\xc}=\candc{\xc'}$ where $\xc'\subseteq \xc$ and let $C'\cap \candc{\xs}=\candc{\xs'}$ where $\xs'\subseteq \xs$. In other words,~$\xc'$ (respectively,~$\xs'$) is the submultiset (respectively, subset) of~$\xc$ (respectively,~$\xs$) corresponding to candidates contained in~$C'$ from~$\xc$ (respectively,~$\xs$).
\begin{description}
    \item[Case~1:] $\abs{\xc'}=\xsize-1$.

    In this case, the SAV score of~$p$ in~$E$ can be at most $9\xsize+6(\xsize-1)\cdot \frac{1}{2}+\frac{6(3\kappa+1)}{3\kappa}< 12\xsize+4$. As $\kappa\geq 3$,~$C'$ contains at least one candidate $c(\xce)\in \candc{\xc}$ where $\xce\in \xc'$. Then, by the construction of the votes, there exists at least one candidate~$c(\xse)\in \candc{\xs} \setminus C'$ such that $\xse\in \xce$ whose SAV score in~$E$ is at least $11\xsize+4+12\kappa (\frac{1}{3}-\frac{1}{4})=12\xsize+4$. However, this contradicts that~$p$ uniquely wins~$E$.

    \item[Case~2:] $\abs{\xc'}\leq \xsize-2$ and $\abs{\xc'}\geq \frac{2\kappa}{3}$.

    In this case, the SAV score of~$p$ in~$E$ can be at most $9\xsize+6(\xsize-2)\cdot \frac{1}{2}+\frac{6(3\kappa+1)}{2\kappa+2}< 12\xsize+3$. As $\abs{\xc'}\geq \frac{2\kappa}{3}$ and $\abs{C'}\leq \kappa$, we know that~$\abs{\xs'}\leq \frac{\kappa}{3}$. As each $\xse\in \xs$ is contained in at most three elements of~$\xc'$ and every element of~$\xc'$ is a $3$-subset, we know that~$\xc'$ covers at least $\frac{2\kappa}{3}$ elements of~$\xs$. This implies that there exists at least one candidate~$c(\xse)\in \candc{\xs}\setminus C'$ such that $\xse\in \xs\setminus \xs'$ and~$\xse$ is contained in at least one~$\xce\in \xc'$.
    By the construction of the votes, the deletion of~$c(\xce)$ increases the SAV score of~$c(\xse)$ by $12\kappa (\frac{1}{3}-\frac{1}{4})=\kappa$. It follows that the SAV score of the candidate~$c(\xse)$ in~$E$ is at least $11\kappa+4+\kappa=12\kappa+4$. However, this contradicts that~$p$ wins~$E$.

    \item[Case~3:] $\abs{\xc'}\leq \frac{2\kappa}{3}-1$.

    In this case,  the SAV score of~$p$ minors that of any candidate $c(\xse)\in \candc{\xs}\setminus C'$ where $\xse\in \xs\setminus \xs'$ in~$E$ is at most $9\xsize+6\cdot (\frac{2\kappa}{3}-1)\cdot \frac{1}{2}-(\frac{12\kappa}{4}+8\kappa-2)=-1$, which contradicts that~$p$ wins~$E$.
\end{description}
By the claim, we assume that~$C'$ consists of exactly~$\xsize$ candidates from~$\candc{\xc}$, i.e., $\abs{\xc'}=\kappa$ and $\xs'=\emptyset$.
Similar to the analysis, we know that the SAV score of~$p$ in~$E$ is $9\xsize+6+3\kappa=12\kappa+6$. This implies that for every $\xse\in \xs$, $\candc{\xc'}$ contains at most one candidate~$c(\xce)$ such that $\xse\in \xce\in \xc'$, since otherwise the SAV score of~$c(\xse)$ in~$E$ is at least $11\xsize+4+12\kappa \cdot (\frac{1}{2}-\frac{1}{4})=14\xsize+4$ which contradicts that~$p$ wins~$E$. It directly follows that~$\xc'$ is set packing, i.e., none of two elements from~$\xc'$ intersect. Then, from the facts that $\abs{\xc'}=\kappa$, every $\xce\in\xc'$ is a $3$-set, and $\abs{\xs}=3\kappa$, it follows that~$\xc'$ covers~$\xs$. Thus, the instance of {\prob{RX3C}} is a {\yesins}.

\begin{itemize}
\item {\probb{CCDC}{NSAV}}
\end{itemize}

Our reduction for {\probb{CCDC}{NSAV}} is obtained from the above reduction for {\probb{CCDC}{SAV}} by adding $n\cdot m^2+\kappa$ new candidates not approved by any vote. Here,~$m$ and~$n$ are respectively the number of candidates and the number of votes created in the instance of {\probb{CCDC}{SAV}}. Then, by Lemma~\ref{lem-relation-sav-nsav}, we know that for any $C'\subseteq C$ such that $\abs{C'}\leq \kappa$, a candidate has the unique highest SAV score in $(C\setminus C', V)$ if and only if it has the unique NSAV score in $(C\setminus C', V)$. This implies that the constructed instance of {\probb{CCDC}{SAV}} is a {\yesins} if and only if the instance of {\probb{CCDC}{NSAV}} is a {\yesins}.
\end{proof}

Now we consider the three nonadditive rules ABCCV, PAV, and MAV. Theorem~\ref{thm-cc-abccv-pav-co-np} already  implies the {\conphns} of {\prob{CCDC}} for ABCCV and PAV.
When $k=1$, the {\nphns} for ABCCV and PAV vanishes. In fact, in this special case, {\probb{CCDC}{ABCCV}} and {\probb{CCDC}{PAV}} are polynomial-time solvable because CCDC for single-winner AV is polynomial-time solvable and ABCCV, PAV, and AV are identical when functioned as single winner voting rules. However, for MAV, we can establish an {\nphns} reduction inspired by Lemma~\ref{lem-a}.

\begin{theorem}
\label{thm-ccdc-mav-nph-k-1}
{\probb{CCDC}{MAV}} is {\nph} even when $k=1$, every vote approves three candidates, and every candidate is approved by at most three votes.
\end{theorem}

\begin{proof}
We prove the theorem by a reduction from {\prob{RX3C}}. Let $(\xs, \xc)$ be an instance of {\prob{RX3C}} where $\abs{\xs}=\abs{\xc}=3\kappa>0$. We create an instance $((C, V), k, J, \ell)$ of {\probb{CCDC}{MAV}} as follows. First, we create five candidates~$p$,~$d_1$,~$d_2$,~$d_3$, and~$d_4$. Then, for every $\xce\in \xc$, we create one candidate~$c(\xce)$. Let $C=\{p, d_1, d_2, d_3, d_4\}\cup \{c(\xce) \setmid \xce\in \xc\}$. Let $J=\{p\}$, $k=1$, and $\ell=\kappa$. We create the following votes. First, we create two votes $v_1=\{p, d_1, d_2\}$ and $v_2=\{p, d_3, d_4\}$. Then, for every $\xse\in \xs$, we create one vote
$v(\xse)=\{c(\xce) \setmid \xse\in \xce\in \xc\}$ which approves exactly the three candidates corresponding to the $3$-subsets in~$\xc$ containing~$\xse$. Let~$V$ be the set of the above $3\kappa+2$ votes. By Lemma~\ref{lem-a}, all candidates are tied as MAV single winners. The above instance of {\probb{CCDC}{MAV}} clearly can be constructed in polynomial time. It remains to show the correctness of the reduction.

$(\Rightarrow)$ Assume that there is an exact $3$-set cover $\xc'\subseteq \xc$ of~$\xs$. Let $C'=\{c(\xce) \setmid \xce\in \xc'\}$ be the set of the~$\kappa$ candidates corresponding to~$\xc'$. Let $E=(C\setminus C', V)$. As~$\xc'$ is an exact set cover of~$\xs$, every vote~$v(\xse)$ approves exactly two candidates in~$C\setminus C'$. Then,~$v_1$ and~$v_2$ become the only two votes approving the maximum number of candidates in~$E$. As~$p$ is the only candidate approved by both~$v_1$ and~$v_2$, due to Lemma~\ref{lem-a}, $\{p\}$ is the unique MAV winning $1$-committee of~$E$.

$(\Leftarrow)$ Assume that there exists $C'\subseteq C$ of at most $\ell=k$ candidates so that $\{p\}$ is the unique MAV winning $1$-committee of $E=(C\setminus C', V)$. Let $\xc'=\{\xce\in \xc \setmid c(\xce)\in C'\}$. By Lemma~\ref{lem-a}, this means that for every vote $v(\xse)$ where $\xse\in \xs$, at least one of the candidates approved in~$v(\xse)$ is from~$C'$. By the definition of~$v(\xse)$, this means that~$\xc'$ contains at least one~$\xce$ such that $\xse\in \xce$. As this holds for all $\xse\in \xs$, we conclude that~$\xc'$ covers~$\xs$. As $\abs{\xc'}\leq \abs{C'}=\kappa$, we know that~$\xc'$ is an exact set cover of~$\xs$.
\end{proof}

However, when~$k$ increases to two, the complexity of {\prob{CCDC}} for ABCCV and PAV radically changes, as implied by the following theorem. 

\begin{theorem}
\label{thm-ccdc-abccv-pav-nph-k-2}
Let~$\varphi$ be an $\omega$-Thiele rule such that $\omega(2)<2\omega(1)$. Then, {\probb{CCDC}{$\varphi$}} is {\nph} even when $k=2$, there is only one distinguished candidate, and every vote approves at most two candidates.
\end{theorem}

\begin{proof}
We prove the theorem by a reduction from {\prob{Clique}} restricted to regular graphs. Let $(G, \kappa)$ be an instance of {\prob{Clique}}, where $G=(\vset, \eset)$ is a $t$-regular graph. We create an instance $((C, V), k, J, \ell)$ of {\probb{CCDC}{$\varphi$}} as follows. First, we create one candidate~$p$. Then, for every vertex $u\in \vset$, we create one candidate denoted still by the same symbol for simplicity. Let $C=\vset\cup \{p\}$ and let $J=\{p\}$. In addition, let $k=2$ and let $\ell=\abs{\vset}-\kappa$. We create the following votes. First, we create a multiset~$V_p$ of~$t$ votes each of which approves exactly~$p$. Then, for every edge $\edge{u}{u'}\in \eset$, we create one vote $v(\edge{u}{u'})=\{u, u'\}$. Let~$V$ be the set of the $t+\abs{A}$ votes created above. This completes the construction of the instance of {\probb{CCDC}{$\varphi$}}, which can be done in polynomial time. In the following, we prove the correctness of the reduction.

$(\Rightarrow)$ Assume that~$G$ has a clique~$\vset'$ of~$\kappa$ vertices. 
Let $E=(\{p\}\cup \vset', V)$. By the construction of the votes and candidates, every candidate in~$C$ is approved by exactly~$t$ votes. Given $\omega(2)<2\omega(1)$, the largest possible~$\varphi$ score of any $2$-committee is~$2t\cdot \omega(1)$, and this is achieved by any~$2$-committee containing the distinguished candidate. As~$N'$ is a clique, for any $2$-committee $\{u, u'\}$ such that $u, u'\in N'$, the vote $v(\edge{u}{u'})$ approves both~$u$ and~$u'$, implying that the~$\varphi$ score of this committee can be at most $2(t-1)\cdot \omega(1)+\omega(2)$ which is strictly smaller than $2t\cdot \omega(1)$. Therefore, we know that all~$\varphi$ winning $2$-committees of~$E$ contain~$p$.

$(\Leftarrow)$ Assume that there exists a subset $C'\subseteq \vset$ of cardinality at most $\ell=\abs{N}-\kappa$ so that all~$\varphi$ winning $2$-committees of $E=(C\setminus C', V)$ contains the distinguished candidate~$p$. Let $N'=N\setminus C'$. Clearly, $\abs{N'}\geq \kappa$. We claim that~$N'$ is a clique in~$G$. Assume, for the sake of contradiction, that this is not the case. Then, there exist $u, u'\in N'$ not adjacent in~$G$. Then, according to the construction of the votes and candidates, 
the~$\varphi$ score of $\{u, u'\}$ is~$2t\cdot \omega(1)$ in~$E$, implying that $\{\vere, \vere'\}$ is a~$\varphi$ winning $2$-committee of~$E$. However, this contradicts that every~$\varphi$ winning $2$-committee of~$E$ contains~$p$.
\end{proof}

As {\prob{Clique}} restricted to regular graphs is {\wah} with respect to~$\kappa$~\cite{DBLP:journals/cj/Cai08,DBLP:conf/iwpec/Marx04,DBLP:conf/cats/MathiesonS08}, the proof of Theorem~\ref{thm-ccdc-abccv-pav-nph-k-2} implies that for any $\omega$-Thiele rule~$\varphi$ such that $\omega(2)<2\omega(1)$, it holds that {\probb{CCDC}{$\varphi$}} is {\wah} with respect to the number of candidates not deleted even when $k=2$, there is only one distinguished candidate, and every vote approves at most two candidates. As ABCCV and PAV are both such $\omega$-Thiele rules, we arrive at the following corollary.

\begin{corollary}
    \label{cor-ccdc-pav-abccv-wa-hard-dual}
For $\varphi\in \{{\emph{\text{ABCCV}}}, {\emph{\text{PAV}}}\}$,  {\probb{CCDC}{$\varphi$}} is {\wah} when parameterized by the number of candidates not deleted. Moreover, this holds even when $k=2$, there is only one distinguished candidate, and every vote approves at most two candidates.
\end{corollary}

\section{Some Fixed-Parameter Algorithms}
\label{sec-fpt}
In the previous sections, we showed that manipulation and control problems are generally computationally hard with only a few exceptions. In this section, we consider these problems from the parameterized complexity point of view.

An important parameter that has been frequently studied in voting problems is the number of candidates (see, e.g.,~\cite{DBLP:journals/tcs/BredereckFNST20,DBLP:journals/jacm/ConitzerSL07,DBLP:journals/jair/FaliszewskiHHR09,DBLP:conf/ecai/Yang14}). In many real-world applications, this parameter is small~\cite{Fishburm05,DBLP:conf/aldt/MatteiW13}. It is easy to see that {\prob{CCAC}} and {\prob{CCDC}} for all rules studied in this paper are {\fpt} with respect to this parameter:  we enumerate all possible choices of at most~$\ell$ candidates to add ({\probb{CCAC}{$\varphi$}}) or to delete ({\probb{CCDC}{$\varphi$}}), 
and check whether at least one of the enumerations leads to a ``{\yes}''-answer. Another natural parameter is the number of votes. It is easy to see that {\probb{CCAV}{$\varphi$}} and {\probb{CCDV}{$\varphi$}} where $\varphi\in \{\text{AV},  \text{SAV}, \text{NSAV}\}$ are {\fpt} with respect to this parameter: we enumerate all possible choices of at most~$\ell$ votes to add ({\probb{CCAV}{$\varphi$}}) or to delete ({\probb{CCDV}{$\varphi$}}), 
and check whether at least one of the enumerations leads to a ``{\yes}''-answer.
For manipulation, we have similar results.

\begin{theorem}
\label{thm-manipulation-fpt-wrt-candidate}
{\probb{\cbcm}{AV}} and {\probb{\sbcm}{AV}} are {\fpt} with respect to the number of candidates~$m$. More precisely, it can be solved in time~$\bigos{2^m}$.
\end{theorem}

\begin{proof}
Let $I=((C, V\cup V_{\text{M}}), \w)$ be an instance of {\probb{\cbcm}{AV}} (respectively,~{\probb{\sbcm}{AV}}). Let $m=\abs{C}$ be the number of candidates, and let $k=\abs{w}$ be the size of the winning committee~$\w$.
We enumerate all subsets $C'\subseteq C$ of at most~$k$ candidates. For each enumerated $C'\subseteq C$, we let all manipulators approve only candidates in~$C'$, i.e., we replace every $v\in V_{\text{M}}$ with $v'=C'$. Let $V'=\{v' \setmid v\in V_{\text{M}}\}$ be the multiset of these votes of manipulators. Let~$E$ be the election after the manipulators turn to approve~$C'$. Then, the AV scores of all candidates in~$E$ are determined. Recall that $\swinspe{E}{\text{AV}}{k}$, $\slosespe{E}{\text{AV}}{k}$, and~$\pwinspe{E}{\text{AV}}{k}$ are respectively the subset of candidates that are contained in all AV winning $k$-committees of~$E$, the subset of candidates none of which is not contained in any AV winning $k$-committees of~$E$, and the subset of the remaining candidates. For notational brevity, we drop AV and~$k$ from the three notations.
Clearly, every~$k$-committee contains~$\swinshort{E}$ and any arbitrary $k-\abs{\swinshort{E}}$ candidates from~$\pwinshort{E}$ is an AV winning $k$-committee of~$E$. Therefore, for {\probb{\cbcm}{AV}}, if for all $v\in V_{\text{M}}$ it holds that $\abs{\swinshort{E} \cap v}+(k-\abs{\swinshort{E}}-\abs{\pwinshort{E} \setminus v})>\abs{v\cap w}$ we immediately conclude that the given instance~$I$ is a {\yesins}; otherwise, we discard the currently enumerated~$C'$.
For {\probb{\sbcm}{AV}}, we conclude that~$I$ is a {\yesins} if the following conditions are satisfied simultaneously:
\begin{itemize}
\item $\sloseshort{E}\cap v\cap w=\emptyset$ for all $v\in V_{\text{M}}$. If this is not satisfied, then there exists a manipulative vote~$v$ so that at least one candidate from $v\cap w$ is not contained in any AV winning $k$-committees of~$E$.

\item Either $v\cap \w\cap \pwinshort{E}=\emptyset$ for all $v\in V_{\text{M}}$ or $\abs{\swinshort{E} \cup \pwinshort{E}}=k$. In fact, if this condition fails, there exists $v\in V_{\text{M}}$ so that $v\cap \w\cap \pwinshort{E} \neq \emptyset$ and $\abs{\swinshort{E} \cup \pwinshort{E}}>k$. Let~$c$ be any arbitrary candidate from $v\cap \w\cap \pwinshort{E}$, then any $k$-committee containing~$\swinshort{E}$ and any arbitrary $k-\abs{\swinshort{E}}$ candidates from $\pwinshort{E} \setminus \{c\}$ is an AV winning $k$-committee of~$E$. However,~$v$ does not prefer such a winning committee to~$w$.
\end{itemize}
If at least one of the above two conditions fails, we discard the enumerated~$C'$.
For both {\probb{\cbcm}{AV}} and {\probb{\sbcm}{AV}}, if all enumerations are discarded, we conclude that the given instance~$I$ is a {\noins}.
Regarding the time complexity, as we have at most~$2^m$ enumerations~$C'$ to consider, the algorithms run in time~$\bigos{2^m}$.
\onlyfull{
For the bribery problem, we guess all possible~$k$-committee $w\subseteq C$ such that $J\subseteq w$. Then, for every voter which is more satisfied with~$\w$ than with the current winning~$k$-committee, we let the this voter approves exactly the candidates in~$\w$. Then we return ``{\yes}'' if doing so leading to~$\w$ to be the winning $k$-committee. If this is not the case, we discard this guess and proceed to the next. If no guess gives us a ``{\yes}'' answer, we conclude that the given instance is a {\noins}.}
\end{proof}

Now we present a general algorithm for {\probb{\cbcm}{$\varphi$}} and {\probb{\sbcm}{$\varphi$}} for all polynomial computable additive rules. A difficulty is that it is not essentially  true that all manipulators need to approve the same candidates in order to improve the result in their favor, as already shown in Example~\ref{ex-1}. However, as the number of candidates is bounded, we can enumerate all possible winning committees in desired running time, and exploit ILP to formulate the question of assigning approved candidates to manipulators.

\begin{theorem}
\label{thm-manipulation-sav-nsav-fpt-wrt-candidate}
For all polynomial computable additive rules~$\varphi$,
{\probb{\cbcm}{$\varphi$}} and {\probb{\sbcm}{$\varphi$}} are {\fpt} with respect to the number of candidates.
\end{theorem}

\begin{proof}
Let $I=(C, V, V_M, \w)$ be an instance of {\probb{\cbcm}{$\varphi$}}/{\probb{\sbcm}{$\varphi$}}. Let $m=\abs{C}$ and let $k=\abs{\w}$. We first enumerate all collections~$\mathcal{W}$ of $k$-committees of~$C$. Each enumerated~$\mathcal{W}$ corresponds to a guess that~$\mathcal{W}$ is exactly the set of~$\varphi$ winning $k$-committees of the final election. There are at most $2^{m\choose k}$ different collections to consider. Let~$\mathcal{W}$ be an enumerated collection. If there exists a committee $\w'\in \mathcal{W}$ and a vote $v\in V_{\text{M}}$ so that~$v$ does not prefer~$\w'$ to~$\w$ (recall that for {\probb{\cbcm}{$\varphi$}},~$v$ prefers~$\w'$ to $\w$ if and only if $\abs{v\cap \w'}>\abs{v\cap \w}$, and for {\probb{\sbcm}{$\varphi$}},~$v$ prefers~$\w'$ to~$\w$ if and only if $(v\cap \w) \subsetneq (v\cap \w')$), we discard~$\mathcal{W}$. Otherwise, we determine if manipulators can cast their votes so that committees in~$\mathcal{W}$ are exactly the winning $k$-committees. Clearly, the given instance~$I$ is a {\yesins} if and only if there is at least one enumerated~$\mathcal{W}$ so that the answer to the question for~$\mathcal{W}$ is ``\yes''. In the following, we give an ILP formulation of the above question with a limited number of variables.

For each subset $S\subseteq C$, let~$V_{\text{M}}^S$ be the subset of votes from~$V_{\text{M}}$ which approve exactly the candidates in~$S$. That is, $V_{\text{M}}^S=\{v\in V_{\text{M}} \setmid v=S\}$. 
For every $S, T\subseteq C$, we create a nonnegative integer variable~$x_{S\rightarrowshort T}$ which indicates the number of votes in $V_{\text{M}}^S$ that turn to approve exactly the candidates in~$T$.
 For $T\subseteq C$, let $X(T)=\sum_{S\subseteq C}x_{S\rightarrowshort T}$, which indicates the number of manipulators whose truthfully approved candidates are exactly those from~$S$ and are demanded to approve candidates from~$T$ in order to improve the result. Let $E=(C, V)$, and for each $S\subseteq C$, let $E^{S}=(C, \{S\})$ be the election contains only one vote approving exactly the candidates in~$S$. We create the following constraints.
\begin{itemize}
\item First, we have the natural constraints $0\leq x_{S\rightarrowshort T}\leq \abs{V_{\text{M}}^S}$ for every variable~$x_{S\rightarrowshort T}$.
\item Second, for each $S\subseteq C$, we have that $\sum_{T\subseteq C}x_{S\rightarrowshort T}=\abs{V_{\text{M}}^S}$.
\item Third, as committees from~$\mathcal{W}$ are supposed to be the winning $k$-committees in the final election, they should have the same~$\varphi$ score. Therefore, for every $T, T'\in \mathcal{W}$, we have that
\[\score{V}{T}{E}{\varphi}+\sum_{S\subseteq C} \score{S}{T}{E^S}{\varphi} \cdot X(S) = \score{V}{T'}{E}{\varphi}+\sum_{S\subseteq C} \score{S}{T'}{E^S}{\varphi} \cdot X(S).\]
\item Let~$T$ be any arbitrary $k$-committee from~$\mathcal{W}$. To ensure that committees from~$\mathcal{W}$ are exactly the~$\varphi$ winning $k$-committees of the final election, for every $k$-committee $T'\subseteq C$ such that $T'\not\in \mathcal{W}$, we have
    \[\score{V}{T}{E}{\varphi}+\sum_{S\subseteq C} \score{S}{T}{E^S}{\varphi} \cdot X(S) > \score{V}{T'}{E}{\varphi}+\sum_{S\subseteq C} \score{S}{T'}{E^S}{\varphi} \cdot X(S).\]
\end{itemize}
The correctness of the ILP formulation is fairly easy to see. Due to Lenstra's algorithm for ILP~\cite{DBLP:journals/mor/Lenstra83}, the above ILP can be solved in {\fpt}-time in~$m$.
\end{proof}

For {\probb{SDCM}{$\varphi$}}, we could obtain the same fixed-parameter tractability result by utilizing a similar algorithm.

\begin{theorem}
\label{thm-new-manipulation-sav-nsav-fpt-wrt-candidate}
For all polynomial computable additive rules~$\varphi$,
{\probb{\sdcm}{$\varphi$}} is {\fpt} with respect to the number of candidates.
\end{theorem}

\begin{proof}
Let $I=(C, V, V_M, k)$ be an instance of {\probb{\sdcm}{$\varphi$}}. Let $m=\abs{C}$. We first enumerate all collections~$\mathcal{W}$ of $k$-committees of~$C$. Each enumerated~$\mathcal{W}$ is in essence a guess that~$\mathcal{W}$ is exactly the set of~$\varphi$ winning $k$-committees of the final election. In addition, let~$\mathcal{W}'$ denote the collection of all winning $k$-committees of~$\varphi$ at $(C, V)$. We first determine if~$\mathcal{W}$ stochastically dominates~$\mathcal{W}'$, which can be done in {\fpt}-time in~$m$ according to the definition of stochastic domination. If this is not the case, we discard this enumerated~$\mathcal{W}$, and proceed to the next one, if there are any. Otherwise, we further determine if it is possible for the manipulators to cast their votes to make~$\mathcal{W}$ exactly the collection of winning $k$-committees, which can be done in {\fpt}-time by solving the ILP described in the proof of Theorem~\ref{thm-manipulation-sav-nsav-fpt-wrt-candidate}. If the ILP has a feasible solution, we conclude that~$I$ is a {\yesins}. If all possible collections of~$\mathcal{W}$ are enumerated without providing us with a conclusion on~$I$, we conclude that~$I$ is a {\noins}.
\end{proof}

In the following, we study fixed-parameter algorithms for election control by adding/deleting voters with respect to the number of candidates. In particular, we show that a natural generalization of both {\probb{CCAV}{$\varphi$}} and {\probb{CCDV}{$\varphi$}} formally defined below is {\fpt} with respect to this parameter.

\EP
{{\prob{Constructive Control by Adding and Deleting Voters}} for~$\varphi$}{\probb{CCADV}{$\varphi$}}
{A set~$C$ of candidates, two multisets~$V$ and~$U$ of votes over~$C$, a positive integer~$k\leq \abs{C}$,  a nonempty subset~$J\subseteq C$ of at most~$k$ distinguished candidates, and two nonnegative integers~$\ell_{\text{AV}}$ and $\ell_{\text{DV}}$ such that $\ell_{\text{AV}}\leq \abs{U}$ and~$\ell_{\text{DV}}\leq \abs{V}$.}
{Are there $V'\subseteq V$ and~$U'\subseteq U$ such that~$|V'|\leq \ell_{\text{DV}}$, $|U'|\leq \ell_{\text{AV}}$, and~$J\subseteq \w$ for all $\w\in \varphi(C, V\setminus V'\cup U', k)$?}

Obviously, both {\probb{CCAV}{$\varphi$}} and {\probb{CCDV}{$\varphi$}} are special cases of {\probb{CCADV}{$\varphi$}}.
Now we show our {\fpt}-results for {\probb{CCADV}{$\varphi$}}, beginning with those for additive rules.


\begin{theorem}
\label{thm-ccav-ccdv-av-sav-nsav-fpt-candidates}
For~$\varphi$ being a polynomial computable additive rule, {\probb{CCADV}{$\varphi$}} is {\fpt} with respect to the number of candidates.
\end{theorem}

\begin{proof}
Let $\varphi$ be a polynomial computable additive rule, and let $I=(C, V, U, k, J, \ell_{\text{AV}}, \ell_{\text{DV}})$ be an instance of {\probb{CCADV}{$\varphi$}}. Let $m=\abs{C}$.
To solve the instance, we first guess a subset~$A\subseteq C\setminus J$ of at least~$m-k$ candidates and a candidate~$b\in J$. The guessed candidate~$b$ is expected to be a candidate in~$J$ with the minimum~$\varphi$ score in the final election among all candidates in~$J$. In addition, the guessed candidates in~$A$ are expected to be the candidates not in~$J$ and have strictly smaller scores than that of~$b$ in the final election. To be more precise, we split the given instance~$I$ into {\fpt}-many subinstances, each of which takes as input~$I$, a subset $A\subseteq C\setminus J$ such that $\abs{A}\geq m-k$, and a candidate~$b\in J$, and asks whether there exist $V'\subseteq V$ and $U'\subseteq U$ such that $\abs{V'}\leq \ell_{\text{DV}}$, $\abs{U'}\leq \ell_{\text{AV}}$, and in the election $(C, (V\setminus V')\cup U')$ the candidate~$b$ has the smallest~$\varphi$ score among all candidates in~$J$, and all candidates in~$A$ have strictly smaller~$\varphi$ scores than that of~$b$. It is easy to see that~$I$ is a {\yesins} if and only if at least one of the subinstances is a {\yesins}.

In the following, we show how to solve a subinstance associated with guessed~$b$ and~$A$ in {\fpt}-time by giving an ILP formulation with a bounded number of variables.

For each subset~$S\subseteq C$, we create two integer variables~$x_{S}$ and~$y_{S}$. So, there are~$2^{m+1}$ variables in total.
The variables~$x_{S}$ and~$y_{S}$ indicate respectively the number of deleted votes from~$\vaes{V}{S}$ and the number of votes added from~$\vaes{U}{S}$ in a certain desired feasible solution.
Recall that~$\vaes{V}{S}$ and~$\vaes{U}{S}$ are respectively the multiset of votes from~$V$ and the multiset of votes from~$U$ that approve exactly the candidates in~$S$. 
Let $E=(C, V)$, and for each $S\subseteq C$, let $E^S=(C, \{S\})$ be the election which contains only one vote approving exactly the candidates in~$S$.
For each~$c\in C$, let
\[\textsf{sc}(c) = \score{V}{c}{E}{\varphi} - \sum_{S\subseteq C} \score{S}{c}{E^S}{\varphi} \cdot x_{S}+\sum_{S\subseteq C} \score{S}{c}{E^S}{\varphi} \cdot y_{S}.\]

The constraints are as follows:
\begin{enumerate}
\item[(1)] As we add at most~$\ell_{\text{AV}}$ unregistered votes and delete at most~$\ell_{\text{DV}}$ votes, we have that
$\sum_{S\subseteq C}x_{S}\leq \ell_{\text{AV}}$ and $\sum_{S\subseteq C}y_{S}\leq \ell_{\text{DV}}$.

\item[(2)] For each $S\subseteq C$, we naturally have~$0\leq x_{S}\leq \abs{\vaes{V}{S}}$ and $0\leq y_{S}\leq \abs{{\vaes{U}{S}}}$.

\item[(3)] To ensure that~$b$ has the minimum~$\varphi$ score among all candidates in~$J$, for each $b'\in J\setminus \{b\}$,  we have that $\textsf{sc}(b')\geq \textsf{sc}(b)$.

\item[(4)] To ensure that all candidates from~$A$ have  scores strictly smaller than that of~$b$, for each~$a\in A$, we have that~$\textsf{sc}(a)<\textsf{sc}(b)$.
\end{enumerate}
This ILP can be solved in {\fpt}-time with respect to~$m$~\cite{DBLP:journals/mor/Lenstra83}. As we have at most $m\cdot 2^m$ subinstances to solve, the whole algorithm runs in {\fpt}-time in~$m$.
%
\end{proof}

For some nonadditive rules considered in the paper, we can obtain similar results.

\begin{theorem}
\label{thm-ccadv-abccv-pav-fpt-m}
For each $\omega$-Thiele rule~$\varphi$, {\probb{CCADV}{$\varphi$}} is {\fpt} with respect to the number of candidates.
\end{theorem}

\begin{proof}
Let $I=(C, V, U, k, J, \ell_{\text{AV}}, \ell_{\text{DV}})$ be an instance of {\probb{CCADV}{$\varphi$}}. Let $m=\abs{C}$.
We derive {\fpt}-algorithms for the problems stated in the theorem  based on Lenstra's theorem on ILP.
First, we compute~$\mathcal{C}_{k, C}(J)$ which can be done in time~$\bigo{2^m}$.
%
Then, we split the instance~$I$ into {\fpt}-many subinstances each taking as input~$I$ and a nonempty $\mathcal{W}\subseteq \mathcal{C}_{k, C}(J)$. In particular,~$\mathcal{W}$ is supposed to be exactly the collection of all winning $k$-committees of~$\varphi$ at the final election. The question of the subinstance is whether there exist~$U'\subseteq U$ and $V'\subseteq V$  so that $\abs{U'}\leq \ell_{\text{AV}}$, $\abs{V'}\leq \ell_{\text{DV}}$, and  at the  election~$(C, (V\setminus V')\cup U')$ all $k$-committees in~$\mathcal{W}$ have the same~$\varphi$ score which is higher than that of any $k$-committee not in~$\mathcal{W}$. Clearly,~$I$ is a {\yesins} if and only if at least one of the subinstances is a {\yesins}.

In the following, we give an ILP formulation for the subinstance associated with a nonempty~$\mathcal{W}\subseteq \mathcal{C}_{k, C}(J)$. Similar to the proof of Theorem~\ref{thm-ccav-ccdv-av-sav-nsav-fpt-candidates}, for each $S\subseteq C$, we create two integer variables~$x_{S}$ and~$y_{S}$.
Regarding the constraints, we first adopt the constraints described in~(1) and~(2) in the proof of Theorem~\ref{thm-ccav-ccdv-av-sav-nsav-fpt-candidates}. Then, we create the following constraints.
%
%
%
For each $k$-committee~$\w$, let ${\textsf{sc}}(\w)$ denote the~$\varphi$ score of~$\w$ with respect to the multiset~$V$ of registered votes, i.e., $\textsf{sc}(\w)=\score{V}{\w}{(C, V)}{\varphi}$.
To ensure that all $k$-committees from~$\mathcal{W}$ have the same score in the final election, for every $w, w'\in \mathcal{W}$, we require
\[{\sf{sc}}(\w)+\sum_{\substack{\w\cap S\neq\emptyset\\ S\subseteq C}} \omega(\abs{w\cap {S}})\cdot (y_{S} - y_{S}) =
{\sf{sc}}(\w')+\sum_{\substack{\w'\cap S\neq\emptyset\\ S\subseteq C}} \omega(\abs{\w'\cap S})\cdot (y_{S}-x_{S}).\]
(If~$\mathcal{W}$ consists of only one $k$-committee, we do not create such constraints.)

Finally, to ensure that the~$\varphi$ score of every $k$-committee in~$\mathcal{W}$ is higher than that of any $k$-committee not in~$\mathcal{W}$, we fix a $k$-committee~$\w$ in~$\mathcal{W}$, and for every $k$-committee~$w'$ not in~${\mathcal{W}}$, we require
\[{\sf{sc}}(w)+\sum_{\substack{w\cap S\neq\emptyset\\ S\subseteq C}} \omega(\abs{w\cap S})\cdot (y_{S}-x_{S}) >
{\sf{sc}}(w')+\sum_{\substack{w'\cap S\neq\emptyset\\ S\subseteq C}} \omega(\abs{\w'\cap S})\cdot (y_{S}-x_{S}).\]

As we have $2^{m+1}$ variables, by a theorem of Lenstra~\cite{DBLP:journals/mor/Lenstra83}, the above ILP can be solved in {\fpt}-time in~$m$.
\end{proof}

Now we present the algorithms for MAV. In the proof of Theorem~\ref{thm-ccav-ccdv-mav-polynomial-time-solvable-k-constant}, we presented a polynomial-time algorithm for {\probb{CCDV}{MAV}} for~$k$ being a constant. The algorithm runs in time $\bigos{m^k}=\bigos{m^m}$, and hence the following corollary holds.

\begin{corollary}
\label{cor-ccdv-mav-fpt-m}
{\probb{CCDV}{MAV}} is {\fpt} with respect to the number of candidates.
\end{corollary}

For {\probb{CCAV}{MAV}}, we present a natural {\fpt}-algorithm with respect to the same parameter.

\begin{theorem}
\label{thm-ccav-mav-fpt-m}
{\probb{CCAV}{MAV}} is {\fpt} with respect to the number of candidates.
\end{theorem}

\begin{proof}
%
Let $I=(C, V, U, k, J, \ell)$ be an instance of {\probb{CCAV}{MAV}}.
Observe that if two unregistered votes approve exactly the same candidates, we can remove any of them without changing the answer to the instance. In light of this observation, we assume that all unregistered votes are distinct, and thus $\abs{U}\leq 2^m$. We enumerate all subsets $U'\subseteq U$ of cardinality at most~$\ell$, and check if all MAV winning $k$-committees of the election $(C, V\cup U')$ contain~$J$ (this can be done in $\bigos{2^m}$ time by enumerating all $k$-committees). If this is the case for at least one of the enumerations, we conclude that~$I$ is a {\yesins}; otherwise, we conclude that~$I$ is a {\noins}.
\end{proof}

Finally, we present color-coding based {\fpt}-algorithms for control by modifying the candidate set when parameterized by the number of voters plus the number of added or deleted candidates. As a matter of fact,  our algorithm is for a natural combination of {\probb{CCAC}{$\varphi$}} and {\probb{CCDC}{$\varphi$}}, formally defined below.

\EP
{Constructive Control by Adding and Deleting Candidates for~$\varphi$}{\probb{CCADC}{$\varphi$}}
{Two disjoint sets~$C$ and~$D$ of candidates, a multiset~$V$ of votes over~$C\cup D$, a positive integer $k\leq \abs{C}$, a nonempty subset $J\subseteq C$ of at most~$k$ distinguished candidates, and two nonnegative integers~$\ell_{\text{DC}}$ and~$\ell_{\text{AC}}$ such that ~$\ell_{\text{DC}}\leq \abs{C}$ and $\ell_{\text{AC}}\leq \abs{D}$.}
{Are there~$C'\subseteq C$ and~$D'\subseteq D$ such that $\abs{C'}\leq \ell_{\text{DC}}$, $\abs{D'}\leq \ell_{\text{AC}}$, and~$J\subseteq \w$ for all $\w\in \varphi(C\setminus C'\cup D', V, k)$?}

The color-coding technique was first used to derive an {\fpt}-algorithm for the~$k$-{\prob{Path}} problem~\cite{DBLP:journals/jacm/AlonYZ95}. At a high level, this technique first randomly colors the ``units'' in the solution space with~$k$ different colors, and then utilizes dynamic programming to find a colored solution. Thanks to a theory on perfect hash functions, such a randomized algorithm can be derandomized without sacrificing the fixed-parameter tractability. In our problems, the solution space are collections of subsets of candidates. Hence, we first randomly color the candidates, and then we explore a certain solution where no two candidates have the same color.


To describe our algorithm formally, we need the following notions.
For a universe~$X$ and a positive integer $\kappa\leq \abs{X}$, an $(X, \kappa)$-perfect class of hash functions is a set of functions $f_i: X\rightarrow [\kappa]$, $i\in [t]$, where~$t$ is an integer, such that for every $\kappa$-subset~$A\subseteq X$, there exists at least one~$f_i$, $i\in [t]$, such that $\bigcup_{a\in A}f_i(a)=[\kappa]$. It is known that there always exists an $(X, \kappa)$-perfect class of hash functions of cardinality at most~$g(\kappa)$ where~$g$ is a function in~$\kappa$ and, moreover, such functions can be constructed in {\fpt}-time in~$\kappa$~\cite{DBLP:journals/jacm/AlonYZ95}.

Our algorithms hinge upon algorithms for  {\probb{$J$-CC}{$\varphi$}} running in {\fpt}-times in the number of voters.

\begin{theorem}
\label{thm-J-CC-FPT-n}
For $\varphi\in \{\memph{\text{ABCCV}}, \memph{\text{PAV}}, \memph{\text{MAV}}\}$, {\probb{$J$-CC}{$\varphi$}} is {\fpt} with respect to the number of voters.
\end{theorem}

For not distracting the reader, we defer to Appendix the proof of Theorem~\ref{thm-J-CC-FPT-n}.
At a high level, our algorithms first compute the optimal score~$s$ of winning $k$-committees which can be done in {\fpt}-time in the number of voters for all concrete rules considered in the paper~\cite{DBLP:journals/jair/BetzlerSU13,DBLP:conf/atal/MisraNS15,DBLP:journals/aamas/YangW23}. Having this optimal score~$s$, the question is then whether there exists at least one $k$-committee which does not contain~$J$ and has score at least (ABCCV and PAV) or at most (MAV)~$s$. Obviously, the {\probb{$J$-CC}{$\varphi$}} instance is a {\yesins} if and only if we have at least one ``{\yes}'' answer. We show that this question can be answered in {\fpt}-time by giving ILP formulations, analogous to those for solving the winner determination problems for these rules studied in~\cite{DBLP:journals/jair/BetzlerSU13,DBLP:conf/atal/MisraNS15,DBLP:journals/aamas/YangW23}.

Empowered with Theorem~\ref{thm-J-CC-FPT-n}, we are ready to present our {\fpt}-algorithms for {\probb{CCADC}{$\varphi$}}.

\begin{theorem}
\label{thm-ccac-ccdc-many-rules-fpt-ell-plus-n}
For $\varphi\in \{\memph{\text{SAV}}, \memph{\text{NSAV}}, \memph{\text{ABCCV}}, \memph{\text{PAV}},  \memph{\text{MAV}}\}$, \probb{CCADC}{$\varphi$} is {\fpt} with respect to the combined parameter~$\ell_{\memph{\text{AC}}}+\ell_{\memph{\text{DC}}}+n$, where~$n$ is the number of votes.
\end{theorem}

\begin{proof}
We derive an algorithm for {\probb{CCADC}{$\varphi$}} as follows.
Let $I=(C, D, V, k, J, \ell_{\text{AC}}, \ell_{\text{DC}})$ be an instance of {\probb{CCADC}{$\varphi$}}. Let $n=\abs{V}$ be the number of votes.

First, we guess two nonnegative integers~$\ell'_{\text{AC}}\leq \ell_{\text{AC}}$ and $\ell'_{\text{DC}}\leq \ell_{\text{DC}}$. Each guessed pair $\{\ell'_{\text{AC}}, \ell'_{\text{DC}}\}$ corresponds to a subinstance of~$I$ which asks whether there is a subset $C'\subseteq C$ of exactly~$\ell'_{\text{DC}}$ candidates and a subset $D'\subseteq D$ of exactly~$\ell'_{\text{AC}}$ candidates so that in election restricted to $C\setminus C'\cup D'$ all candidates from~$J$ are in all winning $k$-committees. Obviously, there are polynomially many subinstances and, moreover,~$I$ is a {\yesins} if and only if at least one of the subinstances is a {\yesins}. To complete the proof, it suffices to show that we can solve each subinstance in polynomial time, which is the focus of the remainder of the proof.

Let~$\{\ell'_{\text{AC}}, \ell'_{\text{DC}}\}$ be a guessed pair of integers. We construct a $(C,\ell'_{\text{DC}})$-perfect class~$\mathcal{F}$ of hash functions whenever $\ell'_{\text{DC}}\geq 1$, and construct a $(D, \ell'_{\text{AC}})$-perfect class~$\mathcal{G}$ of hash functions whenever $\ell'_{\text{AC}}\geq 1$. According to~\cite{DBLP:journals/jacm/AlonYZ95},~$\mathcal{F}$ and~$\mathcal{G}$ can be constructed in {\fpt}-time in~$\ell'_{\text{AC}}+\ell'_{\text{DC}}$. Our algorithm considers all pairs of $(f, g)$ one by one where $f\in \mathcal{F}$ and $g\in \mathcal{G}$. If $\mathcal{F}= \emptyset$ (respectively, $\mathcal{G}= \emptyset$), our algorithms only considers functions in~$\mathcal{G}$ (respectively,~$\mathcal{F}$) one by one.

Let $(f, g)$ be a considered pair. For each~$i\in [\ell'_{\text{DC}}]$ and each $j\in [\ell'_{\text{AC}}]$, let~$C_i$ be the subset of candidates of~$C$ assigned the value~$i$ by~$f$, and let~$D_j$ be the subset of candidates of~$D$ assigned the value~$j$ by~$g$, i.e., $C_i=\{c\in C \setmid f(c)=i\}$ and $D_j=\{c\in D \setmid g(c)=j$\}.
We aim to explore a feasible solution of the subinstance corresponding to $\{\ell'_{\text{AC}}, \ell'_{\text{DC}}\}$ which contains exactly one candidate from each~$C_i$, $i\in [\ell'_{\text{DC}}]$, and contains exactly one candidate from each~$D_j$, $j\in [\ell'_{\text{AC}}]$. With respect to such a solution, if there are two candidates~$c$ and~$c'$ from the same set~$C_i$ or~$D_j$ such that the voters approving them are exactly the same, then~$c$ and~$c'$ are indistinguishable.
In view of this observation, we partition each~$C_i$ (respectively,~$D_j$) into~$2^n$ subsets~$\{C_i^{V'}\}_{V'\subseteq V}$ (respectively, $\{D_j^{V'}\}_{V'\subseteq V}$), so that~$C_i^{V'}$ (respectively,~$D_j^{V'}$) consists of all candidates $c\in C_i$  (respectively,~$c\in D_j$) such that $V(c)=V'$. By the above discussion, for each~$C_i$ (respectively,~$D_j$), it only matters which element from $\{C_i^{V'}\}_{V'\subseteq V}$ (respectively, $\{D_j^{V'}\}_{V'\subseteq V}$) intersects the feasible solution.
In light of this fact, for each~$C_i$ (respectively,~$D_j$), we guess a~$V^i\subseteq V$ (respectively,~$U^i\subseteq V$) such that~$C_i^{V^i}\neq \emptyset$ (respectively, $D_j^{U^j}\neq \emptyset$), which indicates that the feasible solution contains exactly one candidate in~$C_i$ which is from~$C_i^{V^i}$, and contains exactly one candidate in~$D_j$ which is from~$D_j^{U^j}$. As we have at most~$2^n$ choices for each~$C_i$ and each~$D_j$, and we have at most $\ell'_{\text{DC}}\leq \ell_{\text{DC}}$ many~$C_i$s and at most $\ell'_{\text{AC}}\leq \ell_{\text{AC}}$ many~$D_j$s to consider, there are in total at most $2^{n\cdot (\ell_{\text{AC}}+\ell_{\text{DC}})}$ combinations of guesses.
Each combination $\{\{C_i^{V^i}\}_{i\in [\ell'_{\text{DC}}]}, \{D_j^{U^j}\}_{j\in [\ell'_{\text{AC}}]}\}$ determines an instance $((C', V), k, J)$  of {\probb{$J$-CC}{$\varphi$}}, where~$C'$ is obtained from~$C$ by deleting any arbitrary candidate from~$C_i^{V^i}$ and including any arbitrary candidate from~$D_j^{U^j}$, for all $i\in [\ell'_{\text{DC}}]$ and all $j\in [\ell'_{\text{AC}}]$. Then, we check if the instance of {\probb{$J$-CC}{$\varphi$}} is a {\yesins}, which can be done in {\fpt}-time for ABCCV, MAV, and PAV with respect to~$n$ (see Theorem~\ref{thm-J-CC-FPT-n}), and can be trivially done in polynomial time for SAV and NSAV%
. If at least one of the no more than $2^{n\cdot (\ell_{\text{AC}}+\ell_{\text{DC}})}$ instances of {\probb{$J$-CC}{$\varphi$}} is a {\yesins}, the subinstance corresponding to~$\{\ell'_{\text{AC}}, \ell'_{\text{DC}}\}$ is a {\yesins}; otherwise, we consider the next pair $(f', g')$ where $f'\in \mathcal{F}$ and $g'\in \mathcal{G}$, if there are any.

If none of the pairs $(f, g)$ where $f\in \mathcal{F}$ and $g\in \mathcal{G}$ results in a conclusion that the subinstance corresponding to~$\{\ell'_{\text{AC}}, \ell'_{\text{DC}}\}$ is a {\yesins}, we conclude that the subinstance corresponding to~$\{\ell'_{\text{AC}}, \ell'_{\text{DC}}\}$ is a {\noins}.
\end{proof}

\section{Concluding Remarks}
\label{sec-conclusion}
In this paper, we have studied the complexity of several manipulation\onlyfull{, bribery,} and control problems for numerous sought-after {\abmv} rules, namely AV, SAV, NSAV, ABCCV, PAV, and MAV. We showed that these rules generally resist these strategy problems by giving many intractability results. However, it should be pointed out that our study is purely based on worst-case analysis. Whether these problems are difficult to solve in practice demands further investigations. In addition to the hardness results, we also derived several {\fpt}-algorithms with respect to natural parameters and polynomial-time algorithms for some special cases of these problems. We refer to Table~\ref{tab-results-summary} for a summary of our results.

{
There are a number of interesting topics for future research. First, in the control problems studied in this paper, the goal of the external agent
 is to include the given distinguished candidates into all winning $k$-committees. In real-world applications, a tie-breaking scheme is applied so that only one winning $k$-committee is selected.
 When the external agent knows which deterministic tie-breaking scheme is used, a more natural goal is to make the distinguished candidates  be included in the unique winning $k$-committee with respect to the tie-breaking scheme. It is interesting to see whether the complexity of the problems changes with respect to different tie-breaking schemes. Notably, for single-winner control problems, it has been observed that tie-breaking schemes may significantly change the complexity of the problems (see, e.g.,~\cite{DBLP:conf/aaai/AzizGMNW13,DBLP:conf/atal/FaliszewskiHS08,DBLP:conf/ecai/MatteiNW14,DBLP:conf/atal/ObraztsovaEH11}). Second, one could study faster or combinatorial {\fpt}-algorithms for {\fpt} problems studied in this paper. Third, for the {\nph} problems proved in the paper, it is natural to explore their approximation algorithms. In addition, in our manipulation problems, the manipulators are allowed to change their votes in any possible way. It is interesting to see if the complexity changes if manipulators' actions are restricted somehow, e.g., if they are only allowed to drop approved candidates or only allowed to approve some previously disapproved candidates. Fourth, it is interesting to see if the complexity of the problems changes when restricted to specific domains of dichotomous preferences. We refer to~\cite{DBLP:conf/ijcai/ElkindL15,DBLP:journals/corr/abs-2205-09092,DBLP:journals/iandc/FaliszewskiHHR11,DBLP:conf/ijcai/Yang19a} for the notions of several restricted domains.
 Finally, it has been shown that for single-winner voting rules, the problems of control by adding/deleting candidates are already {\nph} when there is only a constant number of voters~\cite{DBLP:journals/jair/ChenFNT17}. It is interesting to explore whether similar results hold for {\probb{CCAC}{$\varphi$}} and {\probb{CCDC}{$\varphi$}} for multiwinner voting rules.


\section*{Appendix}
In this appendix, we give the proofs for Theorem~\ref{thm-maniuplation-sav-nsav-polynomial-time-solvable-constant-number-manipulators} and Theorem~\ref{thm-J-CC-FPT-n}.
\bigskip

{\noindent{\bf{Theorem~\ref{thm-maniuplation-sav-nsav-polynomial-time-solvable-constant-number-manipulators}}}} {\textit{For $\varphi\in \{{\memph{SAV}}, {\memph{NSAV}}\}$, {\probb{\cbcm}{$\varphi$}} and {\probb{\sbcm}{$\varphi$}} are polynomial-time solvable if there are a constant number of manipulators.}}

\begin{proof}
Let $\varphi\in \{{\text{SAV}}, {\text{NSAV}}\}$.
We first present the algorithm for {\probb{\cbcm}{SAV}}. After this, we show how to modify the algorithm to make it work for {\probb{\cbcm}{NSAV}}, {\probb{\sbcm}{SAV}}, and {\probb{\sbcm}{NSAV}}.

Let $I=(C, V, V_{\text{M}}, \w)$ be an instance of {\probb{\cbcm}{$\varphi$}}, where~$\w$ is an SAV winning~$k$-committee of $(C, V\cup V_{\text{M}})$, and~$\varphi$ dentoes SAV. Let $t=\abs{V_{\text{M}}}$ be the number of manipulators, and let $m=\abs{C}$. Observe that if~$I$ is a {\yesins}, it admits a feasible solution~$V'$ so that $C^{\vee}(V')\subseteq C^{\vee}(V_{\text{M}})$.
Based on this observation, for each $S\subseteq V_{\text{M}}$ and each~$v\in V_{\text{M}}$, we guess the number~$m(S, v)$ of candidates from~$C_{V_{\text{M}}}^{\star}(S)$ that the manipulator corresponding to~$v$ turn to approve. For notational brevity, in the following, we use~$C^{\star}(S)$ to denote $C^{\star}_{V_{\text{M}}}(S)$, and let $m^{\star}(S)=\abs{C^{\star}(S)}$.
The number of different combinations of these guesses is $\bigo{(m+1)^{2^t\cdot t}}$ which is a polynomial in~$m$ given~$t$ being a constant. For a fixed combination of the guesses on~$m(S, v)$ for $S\subseteq V_{\text{M}}$ and $v\in V_{\text{M}}$, we then guess the winning-threshold~$s$ of SAV at the final election. As there are only a constant number of manipulators and for each manipulator we have guessed the total number of candidates the manipulator approves (i.e., $\sum_{S\subseteq V_{\text{M}}}m(S, v)$ for a manipulator whose true vote is~$v$),~$s$ may have only polynomially many different values. To be precise, for each $v\in V_{\text{M}}$, let $m_{\sum}(v)=\sum_{S'\subseteq V_{\text{M}}}m(S', v)$. Moreover, let $E=(C, V)$. Then, the values of~$s$ are from
\[\{\score{V}{c}{E}{\varphi} \setmid c\in C\}\cup \{\score{V}{c}{E}{\varphi}+\sum_{v\in S', m_{\sum}(v)>0}\frac{1}{m_{\sum}(v)} \setmid c\in C^{\vee}(V_{\text{M}}), S'\subseteq V_{\text{M}}\}.\] Each combination of the guesses on $m(S, v)$ for all $S\subseteq V_{\text{M}}$ and $v\in V_{\text{M}}$, and the winning-threshold~$s$ corresponds to a subinstance which determines if there exists a multiset~$V'$ of~$t$ votes, one-to-one corresponding  to~$V_{\text{M}}$, so that
\begin{enumerate}
\item[(1)] every $v\in V_{\text{M}}$ prefers all SAV winning $k$-committees of $(C,V\cup V')$ to $\w$,
\item[(2)] for every $v\in V_{\text{M}}$ and every $S\subseteq V_{\text{M}}$, the vote $v'\in V'$ corresponding to~$v$ approves exactly $m(S, v)$ candidates from~$C^{\star}(S)$, and
\item[(3)] the winning-threshold of SAV at $(C, V\cup V')$ is~$s$.
\end{enumerate}
Apparently, the original instance~$I$ is a {\yesins} if and only if at least one of the subinstances is a {\yesins}.
In addition, as discussed above, there are polynomially many subinstances to consider.
Therefore, to show that {\probb{\cbcm}{SAV}} is polynomial-time solvable, it suffices to give a polynomial-time algorithm for solving each subinstance. We propose such an algorithm based on dynamic programming as follows.

Let~$I'$ denote the subinstance corresponding to a combination $\{\{m(S, v)\}_{S\subseteq V_{\text{M}}, v\in V_{\text{M}}}, s\}$ of guesses as described above. For each $v\in V_{\text{M}}$, let $m_{\sum}(v))=\sum_{S\subseteq V_{\text{M}}}m(S, v)$ be as defined above with respect to $\{m(S, v)\}_{S\subseteq V_{\text{M}}, v\in V_{\text{M}}}$. Besides, for $S\subseteq V_{\text{M}}$ and $c\in C$, let
\begin{equation}
\label{eq-hsc}
h(S, c)=
\begin{cases}
\score{V}{c}{E}{\varphi} & S=\emptyset~\text{or}~\sum_{v\in S} m_{\sum}(v)=0,\\
\score{V}{c}{E}{\varphi}+\sum_{\substack{v\in S\\ m_{\sum}(v)>0}}\frac{1}{m_{\sum}(v)} & \text{otherwise}. \\
\end{cases}
\end{equation}
Briefly put, $h(S, c)$ is the SAV score of candidate~$c$ if among all manipulators exactly those corresponding to~$S$ turn to approve~$c$.
We maintain a binary table for each $S\subseteq V_{\text{M}}$ such that $C^{\star}(S)\neq \emptyset$ as follows.
Let $(c^S_1, c^S_2, \dots, c^S_z)$ be any arbitrary linear order on~$C^{\star}(S)$, and for each $x\in [z]$, let $C^S_x=\{c^S_1, c^S_2, \dots, c^S_x\}$ be the set of the first~$x$ candidates in the order.  As~$C^{\star}(S)\neq\emptyset$, it holds that $z\geq 1$. The binary table is denoted by $T_S(x, i, j, {\bf{b}})$ where $x\in [z]$,~$i$ and~$j$ are nonnegative integers such that $i+j\leq m^{\star}(S)$, and~${\bf{b}}$ is a vector of~$t$ nonnegative integers, one for each $v\in V_{\text{M}}$. Particularly,~$i$ and~$j$ indicate respectively the number of candidates from~$C^S_x$ having SAV score strictly larger than~$s$, and the number of candidates from~$C^S_x$ having SAV score exactly~$s$ in the final election. We use~$b(v)$ to denote the integer in~$\bf{b}$ for~$v\in V_{\text{M}}$ which indicates the number of candidates from~$C^S_x$ approved by the vote replacing the manipulative vote~$v$ in the final election. So, we require that $b(v)\leq m(S, v)$. We say that a vote~$v$ is an extension of another vote~$u$ if~$u\subseteq v$. Moreover, a multiset~$V'$ of votes is an extension of a multiset~$U'$ if elements of~$V'$ one-to-one correspond to elements of~$U'$ so that every $v\in V'$ corresponding to~$u\in U'$ is an extension of~$u$.
The entry $T_S(x, i, j, {\bf{b}})$ is defined to be~$1$ if and only if there is a multiset~$\widetilde{V}$ of~$t$ votes over~$C^S_x$ which one-to-one correspond to~$V_{\text{M}}$ by $\rho: V_{\text{M}}\rightarrow \widetilde{V}$ such that the following conditions hold simultaneously:
\begin{enumerate}
\item[(i)] the number of candidates in~$C^S_x$ whose SAV scores are strictly larger than~$s$ in any election $(C, V\cup V')$, where~$V'$ satisfies Condition~(2) given above and is an extension of~$\widetilde{V}$, is exactly~$i$,
 i.e.,
\[\abs{\{c\in C^S_x \setmid h_S(B(c), c)>s\}}=i,\]
where
\[B(c)=\{v\in V_{\text{M}} \setmid c\in \rho(v)\}.\]
\item[(ii)] the number of candidates from~$C^S_x$ whose SAV scores are exactly~$s$ in any election $(C, V\cup V')$, where~$V'$ satisfies Condition~(2) given above and is an extension of~$\widetilde{V}$, is exactly~$j$, i.e.,
\[\abs{\{c\in C^S_x \setmid h_S(B(c), c)=s\}}=j,\]
where~$B(c)$ is defined as in~(i), and
\item[(iii)] for every $v\in V_{\text{M}}$,~$\rho(v)$ approves exactly~$b(v)$ candidates in~$C^S_x$.
\end{enumerate}
The values of the base entries are  set as follows. Let $B=\{v\in V_{\text{M}} \setmid b(v)=1\}$.
\begin{itemize}
    \item[\textcircled{a}] $T_S(1, 0, 0, {\bf{b}})=0$ if and only if at least one of the following conditions holds:
    \begin{itemize}
        \item there exists $v\in V_{\text{M}}$ such that $b(v)>1$;
        \item $h(B, c_1)\geq s$.
    \end{itemize}

    \item[\textcircled{b}] $T_S(1, 1, 0, {\bf{b}})=0$ if and only if at least one of the following conditions holds:
    \begin{itemize}
        \item there exists $v\in V_{\text{M}}$ such that $b(v)>1$;
        \item $h(B, c_1)\leq s$.
    \end{itemize}

    \item[\textcircled{c}] $T_S(1, 0, 1, {\bf{b}})=0$ if and only if at least one of the following conditions holds:
    \begin{itemize}
        \item there exists $v\in V_{\text{M}}$ such that $b(v)>1$;
        \item $h(B, c_1)\neq s$.
    \end{itemize}

    \item[\textcircled{d}] $T_S(1, i, j, {\bf{b}})=0$ for all other entries such that $i+j>1$.
\end{itemize}
We update other entries $T_S(x, i, j, {\bf{b}})$ where $x\geq 2$ as follows.

We let $T_S(x, i, j, {\bf{b}})=1$ if and only if there exists $U\subseteq V_{\text{M}}$ so that at least one of the following conditions holds:
\begin{itemize}
\item[(a)] $T_S(x-1, i, j, {\bf{b'}})=1$ and $h(U, c_x)<s$,
\item[(b)] $T_S(x-1, i-1, j, {\bf{b'}})=1$ and $h(U, c_x)>s$, 
\item[(c)] $T_S(x-1, i, j-1, {\bf{b'}})=1$ and $h(U, c_x)=s$,  
\end{itemize}
where ${\bf{b'}}$ is a vector of~$t$ integers so that $b'(v)=b(v)$ for every $v\in V_{\text{M}}\setminus U$ and $b'(v)=b(v)-1$ for every $v\in U$.

The above three conditions correspond to that exactly the manipulators corresponding to~$U$ turn to approve the $x$-th candidate~$c^S_x$, and the SAV score of~$c^S_x$ in the final election is, respectively, smaller than, greater than, or equal to~$s$.

After all tables are computed for all submultisets of~$V_{\text{M}}$, we determine if the subinstance~$I'$ is a {\yesins} as follows. For each $S\subseteq V_{\text{M}}$, let
\[\mathcal{T}(S)=\{T_S(m^{\star}(S), i, j, {\bf{b}}) \setmid T_S(m^{\star}(S), i, j, {\bf{b}})=1, (\forall{v\in V_{\text{M}}})[b(v)=m(S, v)]\}.\]
Clearly, the cardinality of $\mathcal{T}(S)$ is $\bigo{m^{t+3}}$. For each entry $T_S(m^{\star}(S), i, j, {\bf{b}})\in \mathcal{T}(S)$, let $i_S(i, j)$ be the value of~$i$, and let~$j_S(i, j)$ be the value of~$j$ in the entry. For an $x$-tuple $A=(a_1, a_2, \dots, a_x)$, we use $a\in A$ to denote that~$a$ is a component (element) of~$A$, i.e., $a=a_i$ for some $i\in[x]$.
For each~$\Omega$ in the $2^t$-fold Cartesian product
 $\bigtimes_{S\subseteq V_{\text{M}}} \mathcal{T}(S)$, let
\[i_{\Omega}=\sum_{T_S(m^{\star}(S), i, j, {\bf{b}})\in \Omega}i_{S}(i, j),\]
and let
\[j_{\Omega}=\sum_{T_S(m^{\star}(S), i, j, {\bf{b}})\in \Omega}j_{S}(i, j).\]
In addition, for each $v\in V_{\text{M}}$, let
\[i_{\Omega}(v)=\sum_{\substack{T_S(m^{\star}(S), i, j, {\bf{b}})\in \Omega\\ v\in S}}i_{S}(i, j),\]
and let
\[j_{\Omega}(v)=\sum_{\substack{T_S(m^{\star}(S), i, j, {\bf{b}})\in \Omega\\ v\in S}}j_{S}(i, j).\]
Let~$i^{>s}=\{c\in C\setminus C^{\vee}(V_{\text{M}}) \setmid h(\emptyset, c)>s\}$ be the number of candidates in $C\setminus C^{\vee}(V_{\text{M}})$ whose SAV scores in $(C, V)$ are larger than~$s$. Analogously, let~$i^{=s}=\{c\in C\setminus C^{\vee}(V_{\text{M}}) \setmid h(\emptyset, c)=s\}$ be the number of candidates in $C\setminus C^{\vee}(V_{\text{M}})$ whose SAV scores in $(C, V)$ are exactly~$s$.
We consider all  $\Omega\in \bigtimes_{S\subseteq V_{\text{M}}} \mathcal{T}(S)$ such that
\begin{equation}
\label{eq-d}
i_{\Omega}+i^{>s}\leq k-1~\text{and}~j_{\Omega}+i^{=s}>0.
\end{equation}
By the definition of the tables, we know that for each such an~$\Omega$ such that $i_{\Omega}\leq k-i^{>s}$ and $j_{\Omega}+i^{=s}>0$, there exists a multiset~$V'$ of~$t$ votes over~$C$ that satisfies Conditions~(2) and~(3). If there exists no such~$\Omega$, we conclude that the subinstance~$I'$ is a {\noins}. Otherwise, we conclude that~$I'$ is a {\yesins} if and only if there exists $\Omega \in \bigtimes_{S\subseteq V_{\text{M}}} \mathcal{T}(S)$ which satisfies Inequality~\eqref{eq-d} and, moreover, for every $v\in V_{\text{M}}$ it holds that
\begin{equation}
\label{eq-e}
i_{\Omega}(v)+\max\{0, (k+j_{\Omega}(v)-(i_{\Omega}+i^{>s}+j_{\Omega}+i^{=s}))\}>\abs{v\cap \w}.
\end{equation}
This ensures that Condition~(1) given above is satisfied by a multiset of~$t$ votes corresponding to~$\Omega$.
This completes the description of the algorithm for {\probb{\cbcm}{SAV}}.
\medskip

Let us use~$\mathcal{A}$ to denote the algorithm described above for  {\probb{\cbcm}{SAV}}. We show below how to modify the algorithm for other manipulation problems.
In the following, let $I=(C, V, V_{\text{M}}, \w)$ be an instance of the respective problems.

\begin{itemize}
\item {\probb{\cbcm}{NSAV}}
\end{itemize}

We let~$\varphi$ be NSAV, and in the description of~$\mathcal{A}$ we let the values of~$s$ be from
\[\{\score{V}{c}{E}{\varphi} \setmid c\in C\}\cup \{\score{V}{c}{E}{\varphi} +\sum_{\substack{v\in S'\\ m_{\sum}(v)>0}}\frac{1}{m_{\sum}(v)} - \sum_{\substack{v\in V_{\text{M}}\setminus S'\\ m_{\sum}(v)<m}}\frac{1}{m-m_{\sum}(v)} \setmid c\in C^{\vee}(V_{\text{M}}), S'\subseteq V_{\text{M}}\}.\]
Besides, we redefine
\begin{equation}
\label{eq-h}
h(S, c)=
\begin{cases}
\score{V}{c}{E}{\varphi} & S=\emptyset~\text{or}~\sum_{v\in S} m_{\sum}(v)=0,\\
\score{V}{c}{E}{\varphi}+\sum_{\substack{v\in S\\ m_{\sum}(v)>0}}\frac{1}{m_{\sum}(v)}-
\sum_{\substack{v\in V_{\text{M}}\setminus S\\ m_{\sum}(v)<m}}\frac{1}{m-m_{\sum}(v)} & \text{otherwise}. \\
\end{cases}
\end{equation}

\begin{itemize}
\item {\probb{\sbcm}{SAV}}
\end{itemize}
For each $S\subseteq V_{\text{M}}$, let $C_{\cap \w}(S)=\w\cap C^{\star}(S)$. Obviously, $\w\cap C^{\vee}(V_{\text{M}})=\bigcup_{S\subseteq {V(_{\text{M}})}}C_{\cap \w}(S)$.
To solve~$I$ for {\probb{\sbcm}{SAV}}, in addition to the guesses on $m(S, v)$ for all $S\subseteq V_{\text{M}}$ and all $v\in V_{\text{M}}$, and the winning-threshold~$s$, we also guess whether the final election admits a unique SAV winning $k$-committee, which is equivalent to determining if there are exactly~$k$ candidates in the final election whose SAV scores are at least~$s$.

For the case where our guess is that the final election admits a unique SAV winning $k$-committee, we define an entry $T_S(x, i, j, {\bf{b}})$ to be~$1$ if, in addition to Conditions (i)--(iii) it also satisfies that every candidate in $C_{\cap w}(S)$ has SAV score at least~$s$ in any election $(C, V\cup V')$, where~$V'$ satisfies Condition~(2) and is an extension of~$\widetilde{V}$, i.e., for every $c\in C_{\cap \w}(S)$ it holds that $h(B(c), c)\geq s$. 
We compute the base entries $T_S(x, i,j, {\bf{b}})$ where $x=0$ by distinguishing between whether~$c_x$ is in~$C_{\cap \w}(S)$. If $c_x\not\in C_{\cap \w}(S)$, then we use the same steps (\textcircled{a}--\textcircled{d}) as for the computation of base entries for {\probb{\cbcm}{SAV}} described above. If  $c_x\in C_{\cap \w}(S)$, then we set~$T_S(1, 1, 0, {\bf{b}})$ and $T_S(1, 0, 1, {\bf{b}})$ by Steps~\textcircled{b} and~\textcircled{c}, and set $T_S(1, i, j, {\bf{b}})=0$ for all the other possible entries.
In addition, we update $T_S(x, i,j, {\bf{b}})$ by distinguishing two cases described below.

\begin{description}
\item[Case~1.] $c_x\not\in C_{\cap \w}(S)$

In this case, we set $T_S(x, i,j, {\bf{b}})=1$ if and only if there exists $U\subseteq V_{\text{M}}$ so that at least one of Conditions (a)--(c) described above holds.

\item[Case~2.] $c_x\in C_{\cap \w}(S)$

In this case, we set $T_S(x, i,j, {\bf{b}})=1$ if and only if there exists $U\subseteq V_{\text{M}}$ so that Condition~(b) or Condition~(c) described above holds. The reason that we drop Condition~(1) in this case is because that~$c_x$ should be contained in the unique SAV winning $k$-committee, and hence its SAV score in the final election should be no smaller than~$s$.
\end{description}
After all tables are computed, if there exists $\Omega\in \bigtimes_{S\subseteq V_{\text{M}}} \mathcal{T}(S)$ such that $i_{\Omega}+j_{\Omega}+i^{>s}+i^{=s}=k$, and for every $v\in V_{\text{M}}$ it holds that $i_{\Omega}(v)+j_{\Omega}(v)>\abs{v\cap \w}$, we conclude that the given instance~$I$ is a {\yesins}.

If none of the guesses so far leads to the conclusion that~$I$ is a {\yesins}, we consider the case where the final election has multiple SAV winning $k$-committees. In this case, to contain $w\cap C^{\vee}(V_{\text{M}})$ into all SAV winning $k$-committees, every candidate from $w\cap C^{\vee}(V_{\text{M}})$ needs an SAV score strictly larger than~$s$ in the final election. Therefore, in this case, we define $T_S(x, i, j, {\bf{b}})$ to be~$1$ if, in addition to Conditions (i)--(iii), it also satisfies that every candidate in $w\cap C^{\vee}(V_{\text{M}})$ has SAV score strictly larger than~$s$ in the final election.

To update $T_S(x, i,j, {\bf{b}})$, we distinguish between the same two cases as above. The condition for updating $T_S(x, i,j, {\bf{b}})$ in the case where $c_x\not\in C_{\cap \w}(S)$ is exactly the same as described above. For the case where $c_x\in C_{\cap \w}(S)$, we let $T_S(x, i,j, {\bf{b}})=1$ if and only if there exists $U\subseteq V_{\text{M}}$ so that Condition~(b) described above holds.
After all tables are computed,  we conclude that~$I$ is a {\yesins} if there exists $\Omega\in \bigtimes_{S\subseteq V_{\text{M}}} \mathcal{T}(S)$ such that
\begin{itemize}
\item $i_{\Omega} + i^{>s}\leq k-1$,
\item $j_{\Omega}+i^{=s}>0$, and
\item for every $v\in V_{\text{M}}$ it holds that either
 \begin{itemize}
     \item $i_{\Omega}(v)>\abs{v\cap w}$ or
     \item $i_{\Omega}(v)=\abs{v\cap w}$ and $k+j_{\Omega}(v)-(i_{\Omega}+j_{\Omega}+i^{>s}+i^{=s})>0$.
 \end{itemize}
\end{itemize}
The equality $i_{\Omega}(v)>\abs{v\cap w}$ indicates that all candidates in $v\cap w$ have SAV scores larger than~$s$ in the final election, and in addition to these candidates, there exists at least one extra candidate from $v\setminus w$ which also has SAV score larger than~$s$ in the final election. Besides, $i_{\Omega}(v)=\abs{v\cap w}$ indicates that among all candidates approved in~$v$, exactly those from $v\cap w$ have SAV scores larger than~$s$. Therefore, to satisfy the manipulator, every winning $k$-committee of the final election should contain at least one extra candidate from $v\setminus w$, which is ensured by $k+j_{\Omega}(v)-(i_{\Omega}+j_{\Omega}+i^{>s}+i^{=s})>0$.

If none of the guesses leads to a conclude that~$I$ is a {\yesins}, we conclude that~$I$ is a {\noins}.

\begin{itemize}
\item {\probb{\sbcm}{NSAV}}
\end{itemize}

We adapt the above algorithm for {\probb{\sbcm}{SAV}} by letting~$\varphi$ denote NSAV, and defining~$h(S, c)$ as in Equation~\eqref{eq-h}.
\end{proof}

\bigskip

{\noindent{\bf{Theorem~\ref{thm-J-CC-FPT-n}}}} {\textit{For $\varphi\in \{{\memph{ABCCV}}, {\memph{PAV}},{\memph{MAV}}\}$, {\probb{$J$-CC}{$\varphi$}} is {\fpt} with respect to the number of voters.}}

\begin{proof}
Let $I=(C, V, k, J)$ be an instance of {\probb{$J$-CC}{$\varphi$}}. Let $n=\abs{V}$ be the number of votes, and let $m=\abs{C}$ be the number of candidates.
We derive algorithms for~$\varphi$ being ABCCV, PAV, or MAV separately below. However, these algorithms  share the same skeleton:
\begin{description}
\item[Step~1] Compute an optimal $k$-committee of $(C, V)$ with respect to~$\varphi$.
\item[Step~2] Determine if there is a~$\varphi$ winning $k$-committee of~$(C, V)$ not containing~$J$.
\end{description}
Obviously, the given instance~$I$ is a {\yesins} if and only if the answer to the question in Step~2 is {\no}.
The first step can be done in {\fpt}-time in~$n$ for ABCCV~\cite{DBLP:journals/jair/BetzlerSU13}, PAV~\cite{DBLP:journals/aamas/YangW23}, and MAV~\cite{DBLP:conf/atal/MisraNS15,DBLP:journals/algorithmica/GrammNR03}. Let~$s$ denote the~$\varphi$ score of the optimal $k$-committee computed in Step~1. In the following, we provide details of how to solve the question in Step~2 in {\fpt}-time in~$n$ by utilizing integer linear programmings. The following notions are used in our algorithms.

Recall that for each $V'\subseteq V$, $C_V^{\star}(V')$ denotes the set of candidates approved exactly by votes in~$V'$ among votes from~$V$, and $m^{\star}_V(V')=\abs{C_V^{\star}(V')}$. For notational brevity, let  $C^{\star}(V')=C_V^{\star}(V')$ and let $m^{\star}(V')=m^{\star}_V(V')$. Let $\mathcal{V}(J)=\{V'\subseteq V \setmid C^{\star}(V')\cap J\neq\emptyset\}$,
and let $k'=\sum_{V'\in \mathcal{V}(J)} \abs{C^{\star}(V')}$, both of which can be computed in  time~$\bigos{2^n}$.

\begin{claimm}
\label{claim-a}
Let~$\varphi\in \{\memph{\text{ABCCV}}, \memph{\text{PAV}}, \memph{\text{MAV}}\}$.
There is an optimal $k$-committee $w\subseteq C$ of~$\varphi$ at $(C, V)$ such that $J\setminus w\neq \emptyset$ if and only if there is an optimal $k$-committee~$w'$ of~$\varphi$ at $(C, V)$ and a submultiset $V'\subseteq V$ such that
\begin{enumerate}
\item[(1)] $J\cap C^{\star}(V')\neq \emptyset$ and
\item[(2)] $\abs{w'\cap C^{\star}(V')}< m^{\star}(V')$.
\end{enumerate}
\end{claimm}
\smallskip

\noindent{\textit{Proof of Claim~\ref{claim-a}.}}
Suppose that there is an optimal $k$-committee $w\subseteq C$ of~$\varphi$ at~$(C, V)$ such that $J\setminus w\neq \emptyset$. Let~$c$ be any arbitrary candidate in $J\setminus w$. Let $V'=V(c)$ be the multiset of votes approving~$c$, and let $w'=w$. Then, it is easy to see that the above two conditions hold with respect to~$w'$ and~$V'$. To prove the opposite direction, we assume that the two conditions in the claim hold with respect to some optimal $k$-committee~$w'$ of~$\varphi$ at $(C, V)$ and some $V'\subseteq V$. If $C^{\star}(V')\subseteq J$, we are done because the second condition implies $J\setminus \w'\neq \emptyset$. In addition, if $(C^{\star}(V')\cap J)\setminus w'\neq \emptyset$, we are done too. Now, we may assume that $(C^{\star}(V')\cap J)\subseteq w'$. Then, by the second condition, we know that there exists $c\in C^{\star}(V')\setminus w'$. By the first condition,
we know that there exists $c'\in C^{\star}(V')\cap J$, and hence $c'\in w'$.
Then, as $V(c)=V(c')$, it holds that, for $\varphi\in \{{\text{ABCCV}}, {\text{PAV}}, {\text{MAV}}\}$, replacing~$c$ with~$c'$ in~$w'$ yields another optimal $k$-committee of~$\varphi$ at $(C, V)$ too.
As the new $k$-committee excludes~$c'$ which is in~$J$, we know that the claim holds.
\smallskip

\begin{itemize}
\item ABCCV
\end{itemize}
By Claim~\ref{claim-a}, if $k'>k$, there exists at least one  ABCCV winning $k$-committee of~$(C, V)$ which does not contain~$J$; we are done (the given instance~$I$ is a {\noins}).
Otherwise, for each $\widetilde{V}\in \mathcal{V}(J)$ and each integer $i=\{0, 1, \dots, m(\widetilde{V})-1\}$, defining $\widehat{V}=\emptyset$ if $i=0$ and defining $\widehat{V}=\widetilde{V}$ otherwise, we create one ILP formulating the question of whether the election $(C\setminus C^{\star}(\widetilde{V}), V\setminus \widehat{V})$ admits a $(k-i)$-committee of ABCCV score $s-\abs{\widehat{V}}$. This can be done in {\fpt}-time in $\abs{V\setminus \widehat{V}}\leq n$~\cite{DBLP:journals/jair/BetzlerSU13}. Clearly, there exists a $k$-committee of ABCCV score~$s$ in $(C, V)$ that does not contain $C^{\star}(\widetilde{V})$ if and only if at least one of the~$m(\widetilde{V})\leq m$ ILPs does not admit any feasible solution. As $\abs{\mathcal{V}(J)}\leq 2^n$, we have in total at most $m\cdot 2^n$ ILPs to solve, which can be done in {\fpt}-time in~$n$.

After all ILPs are solved, we conclude that~$I$ is a {\noins} if and only if at least one of the ILPs admits a feasible solution.

\begin{itemize}
\item PAV
\end{itemize}

To answer the question in Step~2 for PAV, we adapt the algorithm presented in~\cite{DBLP:journals/aamas/YangW23} which is designed for solving a winner determination problem for PAV. We need the following problem and several other new notions.

\EPP
{Integer Programming With Simple Piecewise Linear Transformations (IPWSPLT)}
{A collection $\{f_{i,j} \setmid i\in [p], j\in [q]\}$ of $p\cdot q$ piecewise linear concave functions, and a vector $b\in \mathbb{Z}^p$.}
{Is there a vector $x\in \mathbb{Z}^q$ such that for every~$i\in [p]$ it holds that
\begin{equation}
\label{equ-generalized-ilp}
\sum_{j=1}^q f_{i,j}(x_j)\leq b_i?
\end{equation}}

It is known that {\prob{IPWSPLT}} is solvable in time $\bigo{{{\textsf{poly}}}(\abs{I}, t) \cdot q^{2.5q+o(q)}}$, where~$I$ is the number of bits encoding the input, and~$t$ is the maximum number of pieces per function~\cite{DBLP:journals/tcs/BredereckFNST20}.
Similar to the above analysis for ABCCV, we need only to focus on determining if there exists $\widetilde{V}\in \mathcal{V}(J)$ such that there exists a  winning $k$-committee of PAV at $(C, V)$ that does not contain~$C^{\star}(\widetilde{V})$. To this end, we give an IPWSPLT formulation as follows.
First, we create an integer variable~$x_{V'}$ for every $V'\subseteq V$.
Second, for each $v\in V$, we create one integer variable~$x_v$ indicating the number of candidates in~$v$ that are contained in a desired $k$-committee.
The constraints are as follows.
\renewcommand\labelenumi{(\theenumi)}
\begin{enumerate}
\item[(1)] For every $V'\subseteq V$ such that $V'\neq \widetilde{V}$, we naturally have that $0\leq x_{V'}\leq m(V')$.

\item[(2)] To ensure the second condition in Claim~\ref{claim-a}, we have $0\leq x_{\widetilde{V}} < m(\widetilde{V})$.

\item[(3)] As we are seeking a $k$-committee, we have that $\sum_{V'\subseteq V}x_{V'}=k$.

\item[(4)] For each~$v\in V$, we have that $0\leq x_v\leq k$.

\item[(5)] For each~$v\in V$, we have that $x_v=\sum_{v\in V'\subseteq V} x_{V'}$.

\item[(6)] For the last constraint, we need to define a piecewise linear concave function $f: \mathbb{R}_{\geq 0}\rightarrow \mathbb{R}_{\geq 0}$ as follows.
First, $f(0)=0$. Second, for each positive integer~$x$, we define $f(x)=\sum_{i=1}^x\frac{1}{i}$.
Third, for each real~$x$ such that $y< x< y+1$ and~$y$ is a nonnegative integer, we define
\[f(x)=f(y)+(x-y)\cdot (f(y+1)-f(y)).\]

The last constraint is then $\sum_{v\in V} f(x_v)=s$. This ensures the optimality of the desired $k$-committee.
\end{enumerate}
Clearly, the number of variables is $n+2^n$, and hence due to~\cite{DBLP:journals/tcs/BredereckFNST20}, the above IPWSPLT  can be solved in {\fpt}-time in~$n$.

After all the $\abs{\mathcal{V}(J)}\leq 2^n$ IPWSPLTs are solved, we conclude that~$I$ is a {\noins} if and only if at least one of the IPWSPLTs has a feasible solution.

\begin{itemize}
\item MAV
\end{itemize}

Based on Claim~\ref{claim-a}, to answer the question in Step~2, we determine if there exists a $k$-committee~$w'\subseteq C$ of MAV score~$s$ and a multiset $V'\subseteq V$ that fulfill the two conditions given in the claim. To this end, we generate~$\abs{\mathcal{V}(J)}$ many ILPs, one for each element in~$\mathcal{V}(J)$ with constraints being generated to verifying the two conditions in Claim~\ref{claim-a}. Concretely, for each $\widetilde{V}\in \mathcal{V}(J)$, we create one ILP as follows.
For each $V'\subseteq V$, we create one nonnegative integer variable~$x_{V'}$ indicating the number of candidates from~$C^{\star}(V')$ that are contained in a desired $k$-committee.
Regarding the constraints, we first adopt the first three classes of constraints ((1)--(3)) in the proof for PAV. Then, to ensure that the desired $k$-committee has MAV score~$s$, for every vote~$v\in V$ we create the following constraints:
\[\sum_{V'\subseteq V, v\in V'}(m^{\star}(V')-x_{V'})+\sum_{V'\subseteq V, v\not\in V'} x_{V'} \leq s.\]
As we have at most~$2^n$ variables, the above ILP can be solved in {\fpt}-time in~$n$. So, solving all the $\abs{\mathcal{V}(J)}\leq 2^n$ ILPs still can be done in {\fpt}-time in~$n$. After all ILPs are solved, we conclude that the given instance~$I$ is a {\noins} if and only if at least one of the ILPs has a feasible solution.
\end{proof}
\end{document}